\documentclass{amsart}
\usepackage{amsmath,amssymb,amsthm,amsfonts,fullpage,graphicx,mathrsfs,color}
\newtheorem{proposition}{Proposition}
\newtheorem{corollary}{Corollary}
\newtheorem{theorem}{Theorem}
\begin{document}

\title{The semiclassical modified nonlinear Schr\"odinger equation II:  asymptotic analysis of
the Cauchy problem.  The elliptic region for transsonic initial data.}
\author{Jeffery C. DiFranco}
\address{Department of Mathematics, Seattle University, 901 12th Ave., P.O. Box 222000, Seattle, WA 98122}
\email{difranco@seattleu.edu}
\author{Peter D. Miller}
\address{Department of Mathematics, University of Michigan, East Hall, 530 Church St., Ann Arbor, MI 48109}  
\email{millerpd@umich.edu}
\urladdr{www.math.lsa.umich.edu/$\sim$millerpd}
\thanks{This work was partially supported by the National Science Foundation under grant numbers DMS-0807653 and DMS-1206131.}

\begin{abstract}
We begin a study of a multi-parameter family of Cauchy initial-value problems for the modified nonlinear Schr\"odinger equation, analyzing the solution in the semiclassical limit.  We use the inverse scattering transform for this equation, along with the steepest descent method of Deift and Zhou.  The initial conditions are selected both to allow all relevant scattering data to be calculated without approximation and also to place the governing equation in a transsonic state in which the quantum fluid dynamical system formally approximating it  is of hyperbolic type for some $x$ and of elliptic type for other $x$.  Our main result is a global approximation theorem valid in a maximal space-time region connected to the elliptic part of the initial data.
\end{abstract}
\maketitle

\section{Introduction}
Let $\epsilon$ and $\alpha$ be positive parameters.  The modified nonlinear Schr\"odinger (MNLS) equation
\begin{equation}
i\epsilon\frac{\partial\phi}{\partial t} +\frac{\epsilon^2}{2}\frac{\partial^2\phi}{\partial x^2} +
|\phi|^2\phi + i\alpha\epsilon\frac{\partial}{\partial x}\left(|\phi|^2\phi\right)=0,\quad x\in\mathbb{R},\quad t>0
\label{eq:MNLS}
\end{equation}
is a completely integrable generalization of the focusing nonlinear Schr\"odinger  equation, to which the MNLS equation reduces upon setting $\alpha=0$.  The perturbation term proportional to $\alpha$ models the effect of \emph{nonlinear dispersion} and is one of a suite of three terms arising in small-amplitude perturbation theory one order beyond the focusing nonlinear Schr\"odinger
equation in the analysis of short pulse propagation in weakly nonlinear optical fibers \cite{Doktorov02, Doktorov06, DoktorovK01}.  Although the MNLS equation appears to be a perturbation of the focusing nonlinear Schr\"odinger equation, there is a sense in which it can also be considered to be a perturbation of the defocusing nonlinear Schr\"odinger equation.  This is related to broken Galilean symmetry introduced for $\alpha\neq 0$, which also shows that the MNLS equation can be reduced to the so-called derivative nonlinear Schr\"odinger equation of Kaup and Newell \cite{KaupN78}.  The complete integrability of the MNLS equation is connected with its representation as the compatibility condition of a Lax pair based on the WKI spectral problem of Wadati, Konno, and Ichikawa \cite{WadatiKI79}.  The existence of a Lax pair representation makes it possible to analyze quite general solutions of the MNLS equation with remarkable accuracy; as an example, see the long-time asymptotic analysis of Kitaev and Vartanian \cite{KitaevV97, KitaevV99}.

The Cauchy initial-value problem for \eqref{eq:MNLS} is to find a solution $\phi$ subject to the initial condition:
\begin{equation}
\phi(x,0)=A_0(x)e^{iS_0(x)/\epsilon},\quad x\in\mathbb{R}.
\label{eq:MNLS-IC}
\end{equation}
Here, $A_0(\cdot)$ and $S_0(\cdot)$ are real-valued amplitude and phase functions respectively.
Under suitable conditions on these two functions, the Cauchy initial-value problem can be
studied with the help of an \emph{inverse-scattering transform} derived from Lax pair representation of the MNLS equation \eqref{eq:MNLS}.

Our main interest
is in the asymptotic behavior of $\phi=\phi_\epsilon(x,t)$ in the \emph{semiclassical limit} where $\epsilon\downarrow 0$ with
$\alpha>0$ and the functions $A_0(\cdot)$ and $S_0(\cdot)$ held fixed.  This limit is obviously very singular, but the problem can be recast in a form that appears somewhat more tractable by introducing Madelung's fluid dynamical variables \cite{Madelung26}:
\begin{equation}
\begin{split}
\rho_\epsilon(x,t)&:=|\phi_\epsilon(x,t)|^2,\quad\text{(a quantum fluid density)}\\
u_\epsilon(x,t)&:=\epsilon\Im\left\{\frac{\partial}{\partial x}\log(\phi_\epsilon(x,t))\right\},\quad\text{(a quantum fluid velocity)}.
\end{split}
\label{eq:Madelung}
\end{equation}
Here of course $\log(\phi_\epsilon(x,t))$ is defined (assuming $\rho_\epsilon(x,t)$ is non-vanishing) by continuation to be a differentiable function of $x$.
The initial values of these fields are independent of $\epsilon$, since according to \eqref{eq:MNLS-IC},
\begin{equation}
\rho_\epsilon(x,0)=\rho_0(x):=A_0(x)^2\quad\text{and}\quad u_\epsilon(x,0)=u_0(x):=S_0'(x).
\label{eq:Madelung-IC}
\end{equation}
Most of what we will discuss in this paper has to do with quite special choices for the functions $A_0(\cdot)$ and $S_0(\cdot)$.  However, to ensure the validity of certain general asymptotic formulae for scattering data (see \eqref{eq:taudef} and \eqref{eq:PhidefWKB}), we assume that the functions 
$\rho_0(\cdot)$ and $u_0'(\cdot)$ are Schwartz-class real-analytic functions of $x\in\mathbb{R}$, and that $\rho_0(\cdot)$ is nowhere vanishing.  In particular, these assumptions imply that  $u_0(x)\to u_\pm$ as
$x\to\pm\infty$.  Then
\begin{equation}
S_0(x)=S_0(0)+\int_0^x u_0(y)\,dy \quad\implies\quad
S_0(x)=u_\pm x + S_\pm + o(1),\quad x\to\pm\infty,
\end{equation}
where
\begin{equation}
S_+:=S_0(0)+\int_0^{+\infty}\left[u_0(y)-u_+\right]\,dy\quad\text{and}\quad
S_-:=S_0(0)-\int_{-\infty}^0\left[u_0(y)-u_-\right]\,dy.
\end{equation}

In terms of the Madelung fields defined by \eqref{eq:Madelung}, the MNLS equation \eqref{eq:MNLS} can be rewritten without approximation in the form of a coupled system:
\begin{equation}
\frac{\partial\rho_\epsilon}{\partial t} +\frac{\partial}{\partial x}\left(\rho_\epsilon u_\epsilon +
\frac{3}{2}\alpha\rho_\epsilon^2\right)=0\quad\text{and}\quad
\frac{\partial u_\epsilon}{\partial t} +\frac{\partial}{\partial x}\left(\frac{1}{2}u_\epsilon^2-\rho_\epsilon +\alpha\rho_\epsilon u_\epsilon\right)=\frac{1}{2}\epsilon^2\frac{\partial F[\rho_\epsilon]}{\partial x},
\label{eq:MNLSrewrite}
\end{equation}
where $F[\rho_\epsilon]$ is given by the differential rational expression
\begin{equation}
F[\rho_\epsilon]:=
\frac{1}{2\rho_\epsilon}\frac{\partial^2\rho_\epsilon}{\partial x^2}-\left(\frac{1}{2\rho_\epsilon}\frac{\partial\rho_\epsilon}{\partial x}\right)^2.
\label{eq:DispersiveCorrection}
\end{equation}
The only explicit dependence on $\epsilon$ in either the equations of motion \eqref{eq:MNLSrewrite} or the initial conditions \eqref{eq:Madelung-IC} lies in the dispersive correction term on the right-hand side of the equation governing $u_\epsilon$ in \eqref{eq:MNLSrewrite}, and it seems reasonable to try to neglect that term
and solve the corresponding quasilinear system of local conservation laws, the \emph{dispersionless MNLS system}
\begin{equation}
\frac{\partial\rho}{\partial t} +\frac{\partial}{\partial x}\left(\rho u +
\frac{3}{2}\alpha\rho^2\right)=0\quad\text{and}\quad
\frac{\partial u}{\partial t} +\frac{\partial}{\partial x}\left(\frac{1}{2}u^2-\rho +\alpha\rho u\right)=0,
\label{eq:DispersionlessMNLS}
\end{equation}
with the $\epsilon$-independent initial data $\rho(x,0)=\rho_0(x)$ and $u(x,0)=u_0(x)$ given by\eqref{eq:Madelung-IC}.
Although this seems attractive, proving that such a procedure yields an accurate approximation when $\epsilon\ll 1$ is not at all straightforward for several reasons.  Perhaps the greatest obstruction lies in the fact that the dispersionless MNLS system \eqref{eq:DispersionlessMNLS} is not strictly hyperbolic, leading to a certain ill-posedness of the $\epsilon$-independent Cauchy initial value problem purported to approximate the true dynamics.  This means that without strong assumptions on the initial data, there may exist no corresponding solution of \eqref{eq:DispersionlessMNLS} at all.  Even if there is a solution, linearization about a constant state reveals unbounded exponential growth rates corresponding to modes that are likely to be seeded by the dispersive correction term \eqref{eq:DispersiveCorrection}; it is not obvious at all in such a situation whether the solution for small nonzero $\epsilon$ should resemble for $t$ strictly positive that obtained by simply setting $\epsilon=0$.

The characteristic velocities of the dispersionless MNLS system \eqref{eq:DispersionlessMNLS} satisfy a quadratic equation with real coefficients depending on the local values of the Madelung variables $\rho$ and $u$.  The discriminant of this quadratic is proportional via a positive factor to 
\begin{equation}
Q:=\alpha^2\rho+\alpha u -1.
\label{eq:Qdefine}
\end{equation}
At any given $t\ge 0$ the fields $\rho(\cdot,t)$ and $u(\cdot,t)$ can cause the system to be in any one of three different states (with terminology borrowed from the language of stationary flows in gas dynamics):
\begin{itemize}
\item
\emph{Globally supersonic}.  In this case the dispersionless MNLS system is strictly hyperbolic, corresponding to $Q>0$ and hence real distinct characteristic velocities, for all $x\in\mathbb{R}$.
\item
\emph{Globally subsonic}.  In this case, the dispersionless MNLS system is elliptic, corresponding to $Q<0$ and hence distinct complex-conjugate characteristic velocities, for all $x\in\mathbb{R}$.
\item 
\emph{Transsonic}.  In this case, the dispersionless MNLS system is hyperbolic for some values of $x$ and elliptic for others, and $Q$ changes sign as a function of $x$.
\end{itemize}
While the correspondence is not the obvious one, it is shown in \cite{DiFrancoMM11} that
for initial data \eqref{eq:Madelung-IC} that makes the dispersionless MNLS system \eqref{eq:DispersionlessMNLS} globally supersonic (respectively subsonic) at $t=0$, the semiclassical analysis of the MNLS Cauchy problem can essentially be reduced to that of the defocusing (respectively focusing) cubic nonlinear Schr\"odinger equation.
This result generalizes to arbitrary genus a fact that was found for genus zero in \cite{DiFrancoM08} and for genus one in \cite{KuvshinovL94}.

Evidently only for transsonic initial data might the MNLS equation \eqref{eq:MNLS} behave unlike either the focusing or defocusing cubic nonlinear Schr\"odinger equation in the semiclassical limit.
\emph{This observation is the key motivation for the study that we begin in this paper.}

To analyze the solution of the MNLS equation \eqref{eq:MNLS} subject to suitable initial data of the form
\eqref{eq:MNLS-IC} by means of the inverse-scattering transform, one must in general study 
the small-$\epsilon$ asymptotics of two types of problems:
\begin{itemize}
\item The direct scattering problem.  This amounts to the asymptotic analysis of a singularly perturbed
linear differential equation (the WKI spectral problem).  The natural tools applicable to this problem include the classical WKB method and its generalizations.
\item The inverse scattering problem.  This can be formulated as a matrix Riemann-Hilbert problem (see \S\ref{sec:RHP}) with a highly oscillatory or rapidly exponentially growing and decaying jump matrix.  The natural tools applicable to this problem include the noncommutative steepest descent method of Deift and Zhou \cite{DeiftZ93} and its generalizations (in particular the mechanism of the so-called \emph{$g$-function} first introduced in \cite{DeiftVZ97}).
\end{itemize}
Although the asymptotic analysis of the direct scattering problem can be carried out in part using classical methods, what really makes the problem difficult is (i) that one requires asymptotics that are uniform with respect to the spectral parameter, including near exceptional values where turning points collide, and (ii) that for some initial data it is required to approximate some exponentially small quantities that --- like the semiclassical above-barrier reflection coefficient for the Schr\"odinger operator --- cannot be suitably estimated without assumptions of analyticity and the use of turning point analysis in the complex $x$-plane.  For a flavor of the type of calculations required to analyze the WKI direct scattering problem for the MNLS equation in the semiclassical limit, see \cite{DiFrancoMM11}.  It should be stressed however, that with very few exceptions in the literature\footnote{A notable exception is the analysis of the Korteweg-de Vries equation in the small dispersion limit by Claeys and Grava \cite{ClaeysG09,ClaeysG10a,ClaeysG10b}.  Of course in this case the direct scattering problem corresponds to the well-studied and self-adjoint Schr\"odinger operator, and sufficiently accurate semiclassical approximations of the scattering data for this problem were obtained (only fairly recently) by Ramond \cite{Ramond96}.} the asymptotic information that is available from the direct problem is frequently of insufficient quality and/or accuracy to continue with the inverse scattering problem without unjustified formal approximations.  This is especially true when (as for the WKI problem or the Zakharov-Shabat problem for the focusing nonlinear Schr\"odinger equation) the direct spectral problem cannot be cast as an eigenvalue problem for a self-adjoint operator.

On the other hand, the completely rigorous asymptotic analysis of matrix Riemann-Hilbert problems is now a rather well-developed science.  In order to bypass the difficulties mentioned above in regard to the direct problem and hence enable rigorous analysis of the Cauchy problem for the MNLS equation \eqref{eq:MNLS}, we will choose the initial condition functions $A_0(x)$ and $S_0(x)$ to be of the specific form
\begin{equation}
A_0(x)=\nu\,\mathrm{sech}(x)\quad\text{and}\quad S_0(x)=S_0(0) + \delta x + \mu\log(\cosh(x))
\label{eq:data}
\end{equation}
for real parameters $\nu\neq 0$, $S_0(0)$, $\delta$, and $\mu$.  The corresponding initial values for the Madelung fields are
\begin{equation}
\rho_0(x)=\nu^2\mathrm{sech}^2(x)\quad\text{and}\quad u_0(x)=\delta + \mu\tanh(x).
\label{eq:Madelungdata}
\end{equation}
In fact, without loss of generality, we will assume that $\nu=1$ and $S_0(0)=0$, as this can be accomplished by making the substitutions
\begin{equation}
\phi = \nu e^{iS_0(0)/\epsilon}\tilde{\phi},\quad\epsilon =| \nu|\tilde{\epsilon},\quad\alpha = |\nu|^{-1}\tilde{\alpha},\quad t=|\nu|\tilde{t}
\end{equation}
($x$ is not scaled) in both \eqref{eq:MNLS} and \eqref{eq:data} and dropping the tildes.  Unlike the cubic nonlinear Schr\"odinger equation, the MNLS equation is not invariant under Galilean boosts \cite{DiFrancoM08} so the parameter $\delta$ cannot be removed even by going to a moving frame of reference.  Therefore, both $\delta$ and $\mu$ are  essential parameters in the family of initial data of the form \eqref{eq:data}.  The advantage of working with the  family \eqref{eq:data} is that this choice renders the direct scattering problem for the MNLS equation equivalent for all $\epsilon>0$ to a hypergeometric equation, and this in turn allows the corresponding scattering data to be obtained without any approximation whatsoever \cite{DiFrancoM08}. 

For the initial data \eqref{eq:Madelungdata} with $\nu=1$, the quantity $Q$ defined by  \eqref{eq:Qdefine} becomes (for $t=0$)
\begin{equation}
Q(x)=\alpha^2\mathrm{sech}^2(x) +\alpha\delta +\alpha\mu\tanh(x)-1=-\alpha^2T^2 + \alpha\mu T +\alpha^2+\alpha\delta-1,\quad T:=\tanh(x).
\label{eq:specialQ}
\end{equation}
To be in the transsonic case, we need $Q(x)$ to change sign.  As a quadratic function of $T=\tanh(x)$ it is clear that there are at most two roots.  For simplicity we want to arrange that there is exactly one simple root of this quadratic in the interval $T\in (-1,1)$, and this will occur if and only if it takes opposite signs for $T=\pm 1$; hence we assume the condition
\begin{equation}
\alpha |\mu|>|1-\alpha\delta|.
\label{eq:condition-transsonic}
\end{equation}
For technical reasons that we will explain in \S\ref{sec:analytic-properties} we also assume the following two conditions:
\begin{equation}
\mu^2>4(1-\alpha\delta),
\label{eq:condition-no-eigenvalues}
\end{equation}
and
\begin{equation}
\mu>0.
\label{eq:condition-no-zeros}
\end{equation}
If we introduce the combined parameters
\begin{equation}
A:=\frac{1-\alpha\delta}{4\alpha^2}\quad\text{and}\quad B:=\frac{\mu}{4\alpha},
\label{eq:ABdefs}
\end{equation}
then conditions \eqref{eq:condition-transsonic}--\eqref{eq:condition-no-zeros} imply that
$(A,B)$ should occupy the shaded region in the diagram pictured in Figure~\ref{fig:ABplane}.
\begin{figure}[h]
\begin{center}
\includegraphics{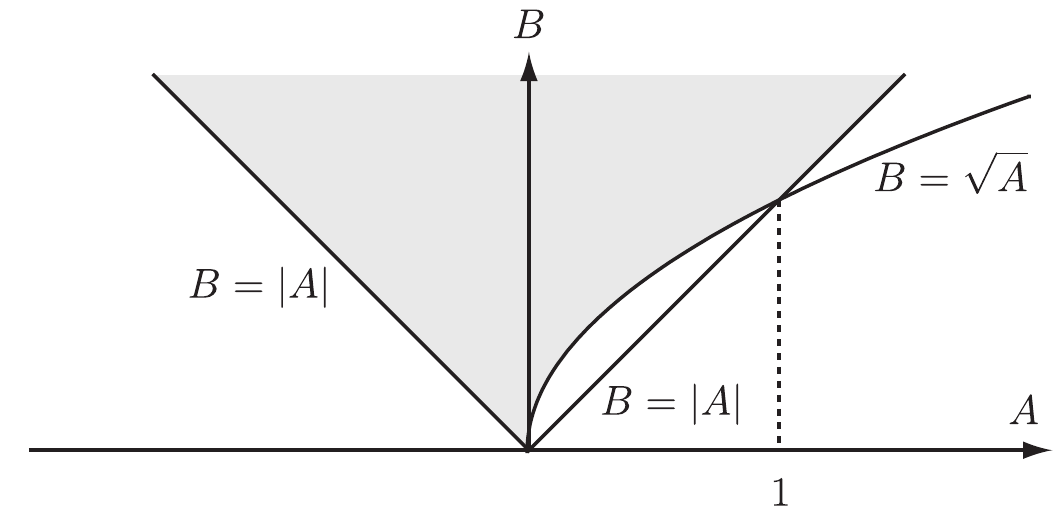}
\end{center}
\caption{The shaded region corresponds to the inequalities \eqref{eq:condition-transsonic}--\eqref{eq:condition-no-zeros}.}
\label{fig:ABplane}
\end{figure}
As an example, some specific parameter values that satisfy these conditions are $\alpha=\delta=1$ and $\mu=2$, corresponding to $A=0$ and $B=1/2$.  Under the inequalities
\eqref{eq:condition-transsonic}--\eqref{eq:condition-no-zeros}, the unique simple root in the interval $(-1,1)$ of the quadratic \eqref{eq:specialQ} is exactly
\begin{equation}
T=T_\mathrm{c}:=2B-\sqrt{4B^2-4A+1}
\label{eq:Tcrit}
\end{equation}
(the other root exceeds $T=1$ for $(A,B)$ in the admissible region), and it follows that when $t=0$, $Q<0$ (elliptic case) for $x<x_\mathrm{c}$ and $Q>0$ (hyperbolic case) for $x>x_\mathrm{c}$, where $x_\mathrm{c}:=\mathrm{arctanh}(T_\mathrm{c})$.  Our main result is then the following.
\begin{theorem}
\label{thm:main}
Let $\phi_\epsilon(x,t)$ denote the solution of the Cauchy initial-value problem for the MNLS equation \eqref{eq:MNLS} subject to initial data of the form \eqref{eq:MNLS-IC} with $A_0(\cdot)$
and $S_0(\cdot)$ given by \eqref{eq:data}, where (without loss of generality) $\nu=1$ and $S_0(0)=0$,
and where $\alpha>0$, $\delta$, and $\mu$ are subject to the inequalities \eqref{eq:condition-transsonic}--\eqref{eq:condition-no-zeros}.  Then there exists a smooth curve $x=x_\mathrm{c}(t)$, $t\ge 0$, with $x_\mathrm{c}(0)=x_\mathrm{c}$ such that for all $t\ge 0$ and all $x<x_\mathrm{c}(t)$,
\begin{equation}
\phi_\epsilon(x,t)=A(x,t)e^{iS(x,t)/\epsilon} + O(\epsilon),\quad\epsilon\to 0,\quad \epsilon>0,
\end{equation}
where the error term is uniform for $(x,t)$ in compact subsets
and where $A(x,t)$ and $S(x,t)$ are smooth, real-valued functions independent of $\epsilon$ that satisfy 
$A(x,0)=A_0(x)$ and $S(x,0)=S_0(x)$.  Also, $Q<0$ holds strictly for all $x<x_\mathrm{c}(t)$ and $t\ge 0$, while $Q\to 0$ as $x\to x_\mathrm{c}(t)$, $x<x_\mathrm{c}(t)$.  Finally, whenever $x<x_\mathrm{c}(t)$ and $t>0$, the Madelung-type fields 
\begin{equation}
\rho(x,t):=A(x,t)^2\quad\text{and}\quad u(x,t):=\frac{\partial S}{\partial x}(x,t)
\end{equation}
exactly satisfy the dispersionless MNLS system \eqref{eq:DispersionlessMNLS}.
\end{theorem}
The functions $A(x,t)$ and $S(x,t)$ will be specified precisely as part of the proof; see \eqref{eq:A-S-xt}.  Significantly, these functions are obtained without directly analyzing the dispersionless MNLS system, but rather arise from the solution of certain algebraic equations (the expressions given in \eqref{eq:A-S-xt} depend on the solution of \eqref{eq:momentscx} and the identities \eqref{eq:kappa-x} and \eqref{eq:KappaCrit}).  Our result is 
interesting in part because it is a \emph{global} description of the solution $\phi_\epsilon(x,t)$
valid for all time $t\ge 0$ in a semi-infinite spatial domain $x<x_\mathrm{c}(t)$.   Unlike the semiclassical limit for the focusing nonlinear Schr\"odinger equation (for which the dispersionless system analogous to \eqref{eq:DispersionlessMNLS} is elliptic for all initial data) \cite{KamvissisMM03,TovbisVZ04}, the solution of the MNLS equation with the initial data under consideration does not ``break'' at some caustic curve beyond which the Madelung ansatz fails and a higher-genus multiphase behavior ensues.  Indeed, the boundary curve $x=x_\mathrm{c}(t)$ is merely
a \emph{sonic line} across which the asymptotic solution changes from subsonic type (for $x<x_\mathrm{c}(t)$) to supersonic type (for $x>x_\mathrm{c}(t)$).  While we do not give the proof in this paper, the sonic transition turns out not to be catastrophic for the Madelung fields, which extend smoothly into the supersonic region (although in this region the Madelung theory turns out to no longer be globally valid in time, that is, the MNLS solution ``breaks'' in the hyperbolic region for the dispersionless MNLS system).  The analysis corresponding to the supersonic region $x>x_\mathrm{c}(t)$ is the subject of our forthcoming paper \cite{DiFrancoM12b} in this series.

From the point of view of technique, this paper contains several novelties that we wish to point out.
The basic approach is to formulate the Riemann-Hilbert problem of inverse-scattering for the MNLS Cauchy problem with available explicit formulae for the scattering data (the reflection coefficient is the key quantity).  Being as the formula for the reflection coefficient involves the Euler gamma function with arguments proportional to the large parameter $\epsilon^{-1}$, one can hardly resist applying Stirling-type approximations.  While much of the analysis corresponds to working with these approximations, we are very careful about not discarding any error terms until they can be rigorously controlled (which always occurs at the very end of the process in the steepest descent method).  In particular, note that the function $\phi_\epsilon(x,t)$ in Theorem~\ref{thm:main} \emph{is the exact solution of the Cauchy initial-value problem for the specified initial data}; we do not replace the initial condition with a nearby function corresponding to approximate scattering data as is a common approach in the subject (see, for example \cite{SG1},
\cite{KamvissisMM03}, and \cite{TovbisVZ04}).  There are points where the Stirling approximation fails (these correspond to the exceptional values of the spectral parameter where turning points collide in WKB theory), and we deal with these points with the help of a new kind of ``lens-opening'' deformation of the Riemann-Hilbert problem that completely sidesteps the need for any local analysis near these points.  This same technique also completely solves certain problems arising from the fact that in cases where one is lucky enough that the reflection coefficient admits an analytic continuation away from the real axis, there frequently exist ``phantom poles'' that do not correspond to eigenvalues of the spectral problem but nonetheless that obstruct the sort of contour deformations that are at the heart of the steepest descent method of Deift  and Zhou.

In \S\ref{sec:scattering-data} we summarize the key facts about the scattering data for the WKI
spectral problem in the case that the initial conditions for the MNLS equation are given in the special form \eqref{eq:data}, incidentally correcting a sign error in the original reference \cite{DiFrancoM08}.  Then in \S\ref{sec:RHP} we recall the formulation of the inverse scattering problem for the MNLS equation as a Riemann-Hilbert problem of analytic matrix factorization.
To prepare the Riemann-Hilbert problem for analysis in the semiclassical limit we need to introduce an appropriate ``$g$-function'' $g(z)$ and a closely related function $h(z)$; these functions are constructed and completely characterized in \S\ref{sec:gfunction}.  We should point out that it is the count of the zeros of the auxiliary function $Y(z)$ related to derivatives of $g$ and $h$ (see Proposition~\ref{prop:Yzeroscx} in \S\ref{sec:gfunction}) that ultimately gives rise to the global-in-time nature of our main result.  Finally, in \S\ref{sec:RHP-deform} we use the $g$-function along with the new ``lens-opening'' method to reduce the Riemann-Hilbert problem formulated in \S\ref{sec:RHP} to a form where it can be treated by means of Neumann series for an associated system of singular integral equations of small-norm type.  The estimates resulting from this analysis complete the proof of Theorem~\ref{thm:main}.

Regarding notation, we use $\Re\{z\}$ and $\Im\{z\}$ to denote the real and imaginary parts, respectively, of a complex number $z$.  We will typically use the symbol $\eta$ to refer to an arbitrarily small positive quantity.
We write all matrices in boldface (e.g. $\mathbf{M}$), with the exception of the Pauli spin matrices:
\begin{equation}
\sigma_1:=\begin{bmatrix}0 & 1\\1 & 0\end{bmatrix},\quad
\sigma_2:=\begin{bmatrix}0 & -i\\i & 0\end{bmatrix},\quad
\sigma_3:=\begin{bmatrix}1 & 0 \\ 0 & -1\end{bmatrix}.
\end{equation}
In particular, the expression $a^{b\sigma_3}$ appears frequently; this is simply the diagonal matrix 
\begin{equation}
a^{b\sigma_3}:=\begin{bmatrix}a^b & 0\\0 & a^{-b}\end{bmatrix},\quad a,b\in\mathbb{C}.
\end{equation}
If $f(z)$ denotes a harmonic or analytic scalar or matrix-valued function defined on the complement of an oriented arc $\mathcal{A}$ in some neighborhood of the complex $z$-plane, then for $z\in\mathcal{A}$ we define $f_+(z)$ (respectively $f_-(z)$) as the boundary value taken by $f(w)$ as $w\to z\in\mathcal{A}$ from the left (respectively right) as $\mathcal{A}$ is traversed according to its orientation.  Throughout our paper, when we write $z^p$ for generally complex $z$ and $p$ we \emph{always} mean the principal branch defined by $z^p:=e^{p\log(z)}$ with $-\pi<\Im\{\log(z)\}<\pi$.  Finally, we denote complex conjugation with an asterisk:  $z^*$.
 
\section{The Scattering Data and Properties Thereof}
\label{sec:scattering-data}
\subsection{Exact Formulae}
According to the theory of the inverse-scattering transform as described in the Appendix of \cite{DiFrancoM08}, the key \emph{scattering coefficients} needed to formulate the inverse-scattering problem are denoted $S_{12}(k)$ (defined for $\Im\{k^2\}=0$) and $S_{22}(k)$ (defined for $\Im\{k^2\}\ge 0$ and analytic in the interior of this region).  For initial data of the special form \eqref{eq:MNLS-IC} with $A(\cdot)$ and $S(\cdot)$ given by \eqref{eq:data}, the
linear scattering problem was reduced to a Gauss hypergeometric equation in \cite{DiFrancoM08}, which allowed the relevant scattering coefficients to be calculated explicitly in terms of the Euler gamma function \cite{DLMF}.  In the case of $\nu=1$ and $S_0=0$ the scattering coefficients found in this way are exactly\footnote{Equation \eqref{eq:S12} corrects a sign error in \cite{DiFrancoM08}.  In that paper, the sign error originates in equation (146) and subsequently propagates into (151), (156), and (158)--(160).}
\begin{equation}
S_{12}
(k)
:=
\frac{2^{-i\mu/\epsilon}i\epsilon}{2
k
}\cdot
\frac{\Gamma(\tfrac{1}{2}-\tfrac{i\mu}{2\epsilon}-\tfrac{\Omega}{\epsilon})\Gamma(\tfrac{1}{2}-\tfrac{i\mu}{2\epsilon}+\tfrac{\Omega}{\epsilon})}{\Gamma(-\tfrac{i\mu}{2\epsilon}-\tfrac{R}{2\epsilon})\Gamma(-\tfrac{i\mu}{2\epsilon}+\tfrac{R}{2\epsilon})},\quad \Im\{k^2\}=0,
\label{eq:S12}
\end{equation}
and
\begin{equation}
S_{22}
(k)
:=\frac{\Gamma(\tfrac{1}{2}+\tfrac{\Omega}{\epsilon}+\tfrac{i\mu}{2\epsilon})\Gamma(\tfrac{1}{2}+\tfrac{\Omega}{\epsilon}-\tfrac{i\mu}{2\epsilon})}{\Gamma(\tfrac{1}{2}+\tfrac{\Omega}{\epsilon}-\tfrac{R}{2\epsilon})\Gamma(\tfrac{1}{2}+\tfrac{\Omega}{\epsilon}+\tfrac{R}{2\epsilon})},\quad \Im\{k^2\}\ge 0.
\label{eq:S22}
\end{equation}
Here
\begin{equation}
\Omega=\Omega(
k
):=\frac{2}{i\alpha}\left(
k^2
-\frac{1-\alpha\delta}{4}\right)
\label{eq:Omegadef}
\end{equation}
and\footnote{The function $R(k)$ defined in \eqref{eq:Rdef} is actually the Schwarz reflection of the function with the same name used in \cite{DiFrancoM08}.}
\begin{equation}
R=R(
k
):=e^{i\pi/4}(-i(16
k^2
-\mu^2))^{1/2}.
\label{eq:Rdef}
\end{equation}
The corresponding \emph{reflection coefficient} is defined generally in terms of the scattering coefficients as follows:
\begin{equation}
r(k):=-\frac{S_{12}(k)}{S_{22}(k)},\quad k^2\in\mathbb{R}.
\label{eq:reflectioncoefficient1}
\end{equation}

In \S\ref{sec:RHP} we will show that we only need to consider $k$ to lie on the boundary of the first quadrant in the complex plane (and for some purposes to admit analytic continuation into the open first quadrant), so we introduce a new coordinate $z$ as follows:
\begin{equation}
z=k^2\quad\text{and}\quad k=e^{i\pi/4}(-iz)^{1/2}.
\label{eq:zkrelation}
\end{equation}
Thus, the closed upper half $z$-plane\footnote{The coordinate $z$ defined in \eqref{eq:zkrelation} is the negative of the variable $z$ used in \cite{DiFrancoM08}.} corresponds in a one-to-one fashion to the closed first quadrant in the $k$-plane.
In terms of $z$, the specific reflection coefficient under consideration is then
\begin{equation}
s(z):=r(e^{i\pi/4}(-iz)^{1/2})
=
-
\frac{2^{-i\mu/\epsilon}i\epsilon}{2e^{i\pi/4}(-iz)^{1/2}}
\frac{\Gamma(\tfrac{1}{2}-\tfrac{i\mu}{2\epsilon}-\tfrac{\Omega}{\epsilon})}
{\Gamma(\tfrac{1}{2}+\tfrac{i\mu}{2\epsilon}+\tfrac{\Omega}{\epsilon})}
\frac{\Gamma(\tfrac{1}{2}+\tfrac{\Omega}{\epsilon}-\tfrac{R}{2\epsilon})
\Gamma(\tfrac{1}{2}+\tfrac{\Omega}{\epsilon}+\tfrac{R}{2\epsilon})}
{\Gamma(-\tfrac{i\mu}{2\epsilon}-\tfrac{R}{2\epsilon})\Gamma(-\tfrac{i\mu}{2\epsilon}+
\tfrac{R}{2\epsilon})},\quad z\in\mathbb{R},
\label{eq:s-define}
\end{equation}
where $\Omega$ and $R$ are written explicitly in terms of $z=k^2$.

A simple exact formula is available for $|s(z)|^2$ for $z\in\mathbb{R}$ as a result of three observations.  Firstly,
since $\Omega$ is imaginary for $z\in\mathbb{R}$, and since $\Gamma(w^*)=\Gamma(w)^*$, it follows that
\begin{equation}
\left|\frac{\Gamma(\tfrac{1}{2}-\frac{i\mu}{2\epsilon}-\frac{\Omega}{\epsilon})}
{\Gamma(\tfrac{1}{2}+\frac{i\mu}{2\epsilon}+\frac{\Omega}{\epsilon})}\right|=1,\quad z\in\mathbb{R}.
\end{equation}
Secondly, the function
\begin{equation}
m(z):=\frac{\Gamma(\tfrac{1}{2}+\frac{\Omega}{\epsilon}-\frac{R}{2\epsilon})
\Gamma(\tfrac{1}{2}+\frac{\Omega}{\epsilon}+\frac{R}{2\epsilon})}
{\Gamma(-\frac{i\mu}{2\epsilon}-\frac{R}{2\epsilon})\Gamma(-\frac{i\mu}{2\epsilon}+\frac{R}{2\epsilon})}
\end{equation}
is meromorphic in a horizontal strip containing the real $z$ axis, and hence so is the function $M(z):=m(z)m(z^*)^*$.  Of course, for $z\in\mathbb{R}$ we have $M(z)=|m(z)|^2$, and we may calculate
$M(z)$ for all real $z$ by obtaining a formula valid for $z$ in some interval of the real axis and applying
analytic continuation to assert that the same formula holds for all $z\in\mathbb{R}$.  For $z<\mu^2/16$,
 $R$ and $\Omega$ are both imaginary, so using $\Gamma(w^*)=\Gamma(w)^*$ we obtain
 \begin{equation}
 M(z)=\frac{[\Gamma(\frac{1}{2}+\frac{\Omega}{\epsilon}-\frac{R}{2\epsilon})\Gamma(\frac{1}{2}-\frac{\Omega}{\epsilon}+\frac{R}{2\epsilon})][\Gamma(\frac{1}{2}+\frac{\Omega}{\epsilon}+\frac{R}{2\epsilon})\Gamma(\frac{1}{2}-\frac{\Omega}{\epsilon}-\frac{R}{2\epsilon})]}
 {[\Gamma(-\frac{i\mu}{2\epsilon}-\frac{R}{2\epsilon})\Gamma(\frac{i\mu}{2\epsilon}+\frac{R}{2\epsilon})]
 [\Gamma(-\frac{i\mu}{2\epsilon}+\frac{R}{2\epsilon})\Gamma(\frac{i\mu}{2\epsilon}-\frac{R}{2\epsilon})]},
 \quad z<\frac{\mu^2}{16}.
 \label{eq:M-for-some-z}
 \end{equation}
Thirdly, we recall the two reflection identities for the gamma function \cite{DLMF}:  
\begin{equation}
\Gamma(w)\Gamma(-w)=-\frac{\pi}{w\sin(\pi w)}\quad\text{and}\quad
\Gamma(\tfrac{1}{2}+w)\Gamma(\tfrac{1}{2}-w) = \frac{\pi}{\cos(\pi w)}.
\label{eq:GammaReflection}
\end{equation}
Using these in \eqref{eq:M-for-some-z}, we obtain
\begin{equation}
\begin{split}
M(z)&=\frac{(-\frac{i\mu}{2\epsilon}-\frac{R}{2\epsilon})(-\frac{i\mu}{2\epsilon}+\frac{R}{2\epsilon})
\sin(\pi(-\frac{i\mu}{2\epsilon}-\frac{R}{2\epsilon})))\sin(\pi(-\frac{i\mu}{2\epsilon}+\frac{R}{2\epsilon}))}
{\cos(\pi(\frac{\Omega}{\epsilon}-\frac{R}{2\epsilon}))\cos(\pi(\frac{\Omega}{\epsilon}+\frac{R}{2\epsilon}))}\\
&= \frac{\mu^2-(-iR)^2}{4\epsilon^2}\frac{\cosh^2(\frac{\pi\mu}{2\epsilon})-\cosh^2(\frac{\pi R}{2i\epsilon})}{\cosh^2(\frac{\pi\Omega}{i\epsilon})+\cosh^2(\frac{\pi R}{2i\epsilon})-1}.
\end{split}
\end{equation}
It is clear that the latter is an even function of $R$ and hence defines the analytic continuation 
of $M(z)$ to the whole real $z$-axis.
From these three facts it therefore follows that
\begin{equation}
|s(z)|^2 =\mathrm{sgn}(z)\frac{\cosh^2(\frac{\pi\mu}{2\epsilon})-\cosh^2(\frac{\pi R}{2i\epsilon})}{\cosh^2(\frac{\pi\Omega}{i\epsilon})+\cosh^2(\frac{\pi R}{2i\epsilon})-1} = 
\frac{|\cosh^2(\frac{\pi\mu}{2\epsilon})-\cosh^2(\frac{\pi R}{2i\epsilon})|}{\cosh^2(\frac{\pi\Omega}{i\epsilon})+\cosh^2(\frac{\pi R}{2i\epsilon})-1},\quad z\in\mathbb{R}.
\label{eq:modssquared}
\end{equation}
Since $\cosh^2(\frac{\pi R}{2i\epsilon})$ is even in $R$, it follows that the product $\mathrm{sgn}(z)|s(z)|^2$ has a meromorphic continuation from $\mathbb{R}$ to the whole complex $z$-plane.

\subsection{Analytic Properties of the Scattering Data}
\label{sec:analytic-properties}
When the scattering coefficients are given by the specific formulae \eqref{eq:S12}--\eqref{eq:S22},
it is obvious that $S_{12}$ and $S_{22}$ are even functions of $R$ and hence they have no branch point at $z=\mu^2/16$, the square-root branch point of $R$ considered as a function of $z=k^2$.  This makes $kS_{12}(k)$ and $S_{22}(k)$ both meromorphic functions of $z=k^2$.
From this it follows that $s(z)$ has a meromorphic continuation from the real $z$-axis into the full upper half $z$-plane.  The point $z=0$ is a branch point of $s(z)$, however, as is evidenced by the fact that the product $s(z)(-iz)^{-1/2}$ is analytic and non-vanishing near $z=0$ (the apparent
singularity at $z=0$ from the factor $(-iz)^{-1/2}$ is cancelled by the simple zero of $\Gamma(-\frac{i\mu}{2\epsilon}+\frac{R}{2\epsilon})^{-1}$).  This shows that $s(z)$ does not have any single-valued continuation into the lower half $z$-plane due to the mismatch of boundary values taken on the branch cut of $(-iz)^{1/2}$ along the negative imaginary axis.  In this paper, we will refer to both the meromorphic continuation of $s(z)$ into the upper half $z$-plane as well as its boundary value taken on the real $z$-axis by the same notation:  $s(z)$.  Now we consider the analytic nature of $s(z)$ for $\Im\{z\}\ge 0$.

In general, any zeros of the scattering coefficient $S_{22}(k)$ for $\Im\{k^2\}\ge 0$ (equivalently for $\Im\{z\}\ge 0$) represent discrete spectrum (eigenvalues) of the direct scattering problem and ultimately poles in the matrix unknown of the Riemann-Hilbert problem of inverse-scattering theory.  For technical reasons we would like to avoid having to include such poles.
This motivates the condition \eqref{eq:condition-no-eigenvalues} on the parameters $\alpha$, $\delta$, and $\mu$; indeed in the case that  $S_{22}(k)$ is given by the formula \eqref{eq:S22}, 
it was shown in \cite{DiFrancoM08} that under the condition \eqref{eq:condition-no-eigenvalues} there are no zeros in the closed upper half $z$-plane for any $\epsilon>0$. In this case, $S_{22}$ is an analytic function of $z$ for $\Im\{z\}>0$ that extends continuously to the real $z$-axis and that for each $\epsilon>0$ is bounded away from zero.

The scattering coefficient $S_{12}$ generally has no analytic continuation from the axes $\Im\{k^2\}=0$ (or the real axis $\Im\{z\}=0$).  However, in the case of the special formula \eqref{eq:S12}, it is clear that $S_{12}$ can be continued into the upper half $z$-plane as a meromorphic function.  The factors $\Gamma(\tfrac{1}{2}-\tfrac{i\mu}{2\epsilon}+\tfrac{\Omega}{\epsilon})$ and $\Gamma(-\tfrac{i\mu}{2\epsilon}+\tfrac{R}{2\epsilon})$ are analytic
and nonvanishing for $\Im\{z\}\ge 0$, but the factor $\Gamma(\tfrac{1}{2}-\tfrac{i\mu}{2\epsilon}-\tfrac{\Omega}{\epsilon})$ contributes to $S_{12}$ an infinite array of simple poles with small spacing 
$\alpha\epsilon/2$ along the vertical ray $\Re\{z\}=z_\mathrm{P}$  and $\Im\{z\}>0$ where
\begin{equation}
z_\mathrm{P}:=\frac{1}{4}(1-\alpha\delta+\alpha\mu).
\label{eq:zpdef}
\end{equation}
These poles are singularities of $s(z)$
that unlike zeros of $S_{22}$ do not correspond to eigenvalues of the direct scattering problem; in \cite{DiFrancoM08} they are called \emph{phantom poles}.   The presence of such poles in the reflection coefficient of an analogous inverse-scattering problem for the focusing nonlinear Schr\"odinger equation has been previously dealt with \cite{TovbisVZ04} by means of the installation of an implicit local parametrix;  by contrast the deformations of the Riemann-Hilbert problem that we will carry out in \S\ref{sec:RHP-deform} are specially designed to render this unfortunate feature of the reflection coefficient completely harmless without the use of any local parametrix at all.  The remaining factor 
$\Gamma(-\tfrac{i\mu}{2\epsilon}-\tfrac{R}{2\epsilon})$ can only contribute zeros to $S_{12}$.
For technical reasons we wish to prevent these zeros from lying in the upper half $z$-plane,
which is the purpose of the condition \eqref{eq:condition-no-zeros} on $\mu$.  Were this condition
not satisfied, the factor in question would contribute to $S_{12}$ 
infinitely many simple zeros lying along (and densely filling out as $\epsilon\downarrow 0$) a parabolic curve in the upper half-plane, and their presence would cause difficulties with the use of
Stirling's formula for  asymptotic analysis (see \S\ref{sec:asymptotic-properties}).

Under assumptions \eqref{eq:condition-no-eigenvalues} and \eqref{eq:condition-no-zeros},
$s(z)$ is analytic and non-vanishing in the open upper half-plane with the exception of the ``phantom poles'' along the line $\Re\{z\}=z_\mathrm{P}$.  Note also that the condition \eqref{eq:condition-no-eigenvalues} implies that the zero $z=\mu^2/16$ of $R$ is greater than
the zero $z=(1-\alpha\delta)/4$ of $\Omega$, and that conditions \eqref{eq:condition-transsonic}
and \eqref{eq:condition-no-zeros} taken together imply that $z_\mathrm{P}>0$.

\subsection{Asymptotic Properties of the Scattering Data}
\label{sec:asymptotic-properties}

\subsubsection{Asymptotic behavior of $|s(z)|^2$}
The three conditions \eqref{eq:condition-transsonic}--\eqref{eq:condition-no-zeros} also determine the qualitative asymptotic behavior of $|s(z)|^2$ given by \eqref{eq:modssquared} for $z\in\mathbb{R}$ in the limit $\epsilon\downarrow 0$ as we will now show.  We have $\Omega\in i\mathbb{R}$ for all $z\in\mathbb{R}$, and $\mu$ is a real constant. 
The relative size of the positive real (quadratic and constant, respectively) quantities 
\begin{equation}
q(z):=(-i\Omega)^2 = \frac{1}{4\alpha^2}(4z-(1-\alpha\delta))^2\quad\text{and}\quad
c(z):=\left(\frac{\mu}{2}\right)^2=\frac{\mu^2}{4}
\end{equation}
will therefore play a role in determining the asymptotic behavior of $|s(z)|^2$ for real $z$ in
the limit $\epsilon\downarrow 0$.   If 
$z>\mu^2/16$, then $R\in\mathbb{R}$ and these two functions determine the asymptotic behavior completely; on the other hand, if $z<\mu^2/16$,
then $R\in i\mathbb{R}$, and we also need to take into account the positive real (linear) quantity
\begin{equation}
l(z):=\left(-i\frac{R}{2}\right)^2 = \frac{\mu^2}{4}-4z,\quad z<\frac{\mu^2}{16}.
\end{equation}
Extending the definition of $l(z)$ by setting $l(z):=0$ for $z\ge \mu^2/16$, we observe that
the asymptotic behavior of $|s(z)|^2$ is determined by the maximum of $q(z)$, $l(z)$, and $c(z)$:
\begin{itemize}
\item For those $z$ for which $c(z)=\max\{q(z),l(z),c(z)\}$, $|s(z)|^2$ is exponentially large in the limit $\epsilon\to 0$ (exponential growth of $s(z)$).
\item For those $z$ for which $l(z)=\max\{q(z),l(z),c(z)\}$, $|s(z)|^2-1$ is exponentially small in the limit $\epsilon\to 0$ (pure oscillation of $s(z)$).
\item For those $z$ for which $q(z)=\max\{q(z),l(z),c(z)\}$, $|s(z)|^2$ is exponentially small in the limit $\epsilon\to 0$ (exponential decay of $s(z)$).
\end{itemize}
It is easy to see that for $z<0$, $l(z)>c(z)$, and that for $z>0$, $l(z)<c(z)$, while $l(0)=c(0)$.  It follows that there are only three possible scenarios for how the real $z$-axis may be partitioned into intervals of growth, pure oscillation, and decay of $s(z)$, as illustrated in Figure~\ref{fig:zregions}.
\begin{figure}[h]
\begin{center}
\includegraphics{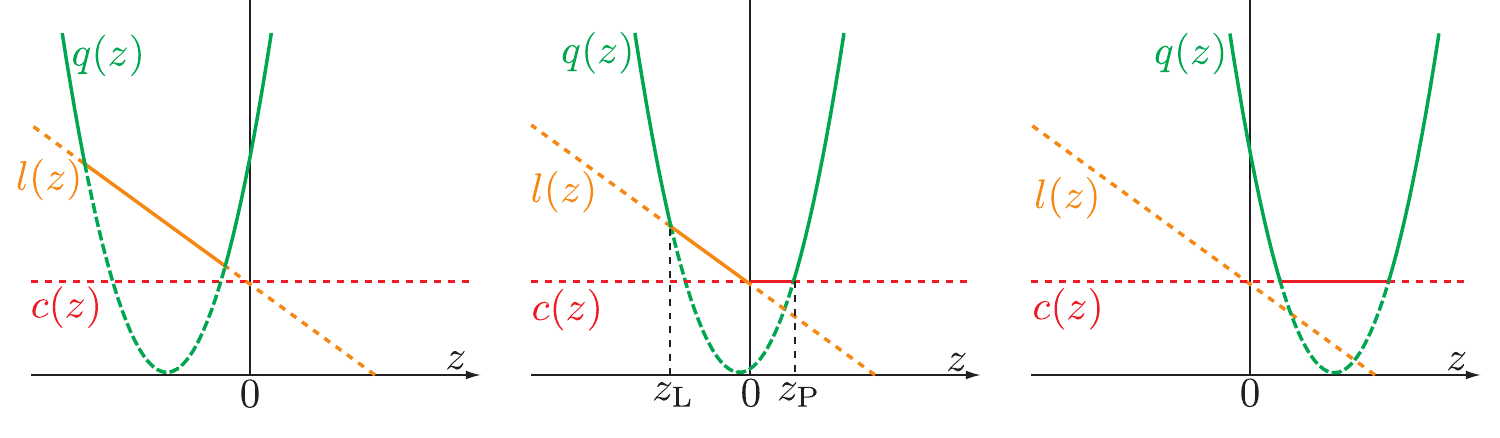}
\end{center}
\caption{The three configurations of the reflection coefficient.  Left:  $1-\alpha\delta+\alpha|\mu|<0$.
Center:  $\alpha|\mu|>|1-\alpha\delta|$.  Right:  $1-\alpha\delta-\alpha|\mu|>0$.  The subdominant arcs of the graphs of $q(z)$, $l(z)$, and $c(z)$ are shown with dashed curves.}
\label{fig:zregions}
\end{figure}
Noting that the points $z=(1-\alpha\delta\pm\alpha|\mu|)/4$ are the solutions of
the equation $q(z)=c(z)$, these scenarios are:
\begin{itemize}
\item
If $1-\alpha\delta+\alpha|\mu|<0$, then there are two real and negative roots of $q(z)=l(z)$, and
for $z$ between these two roots $s(z)$ is purely oscillatory, while for $z$ outside of this interval
$s(z)$ is exponentially small.  See Figure~\ref{fig:zregions}, left-hand panel.
In such cases one expects the asymptotic analysis of the
Riemann-Hilbert problem of inverse scattering to resemble that which has been carried out 
for the Korteweg-de Vries equation in the small dispersion limit (see \cite{DeiftVZ97} as well
as \cite{ClaeysG09}--\cite{ClaeysG10b}).  In such problems all of the necessary deformations of the Riemann-Hilbert problem that are required for the steepest descent technique are purely local to the real axis in the $z$-plane, and hence analyticity of the scattering data does not play any central role in the theory.  This condition really corresponds to the case of globally supersonic initial 
data for the MNLS equation, for which there is a clear analogy with the inverse-scattering theory of the defocusing nonlinear Schr\"odinger equation in the semiclassical limit as shown in \cite{DiFrancoMM11}.
\item
If  $1-\alpha\delta-\alpha|\mu|>0$, then $s(z)$ is exponentially large for $(1-\alpha\delta-\alpha|\mu|)/4<z<(1-\alpha\delta+\alpha|\mu|)/4$ (where $c(z)$ is dominant) and otherwise is exponentially small.  See Figure~\ref{fig:zregions}, right-hand panel.  In such cases one expects that to handle the exponentially large reflection coefficient it will be necessary to introduce
the meromorphic continuation of $s(z)$ from its interval of growth into the complex $z$-plane.  This type of analysis was
carried out for the semiclassical limit of the focusing nonlinear Schr\"odinger equation in 
\cite{TovbisVZ04}.  In fact, this condition corresponds to the case of globally subsonic initial data for the
MNLS equation, and the analogy with the focusing nonlinear Schr\"odinger equation in the semiclassical limit is also considered in \cite{DiFrancoMM11}.  
\item 
If $\alpha|\mu|>|1-\alpha\delta|$, (that is, if \eqref{eq:condition-transsonic} holds) then there is a unique negative root of $q(z)=l(z)$, given by
$z=z_\mathrm{L}:=z^-$, where
\begin{equation}
z^\pm:=\frac{1}{2}\left[\frac{1}{2}(1-\alpha\delta)-\alpha^2\pm
\sqrt{\alpha^4+\frac{\alpha^2\mu^2}{4}-\alpha^2(1-\alpha\delta)}\right],
\label{eq:z-plus-minus}
\end{equation}
and $s(z)$ is exponentially small for $z<z_\mathrm{L}$ and for $z>(1-\alpha\delta+\alpha|\mu|)/4$,
is purely oscillatory for $z_\mathrm{L}<z<0$, and is exponentially large for $0<z<(1-\alpha\delta+\alpha|\mu|)/4$.  See Figure~\ref{fig:zregions}, center panel.  Note that under the assumption \eqref{eq:condition-no-zeros}, the point of transition from exponential growth to exponential decay is exactly $z=z_\mathrm{P}$.
\end{itemize}
These results show that the assumption \eqref{eq:condition-transsonic}, which was originally imposed to ensure that
the initial conditions were of transsonic type, is also exactly what is required to ensure that in
the spectral transform domain there exist both a negative interval of pure oscillation and also an abutting positive interval of exponential growth of the reflection coefficient $s(z)$.  

To handle the Riemann-Hilbert problem of inverse scattering in the transsonic case under consideration will therefore require in the same problem a combination of techniques from both the category of ``modulationally stable'' semiclassical limits (e.g. Korteweg-de Vries, defocusing nonlinear Schr\"odinger) and also the category of ``modulationally unstable'' semiclassical limits (e.g. focusing nonlinear Schr\"odinger).

Assuming the three conditions \eqref{eq:condition-transsonic}--\eqref{eq:condition-no-zeros}, note that as shown in the diagram in the central panel of Figure~\ref{fig:zregions}, there are five distinguished points on the real $z$-axis, in order from left-to-right:
\begin{itemize}
\item $z=z_\mathrm{L}$, the point of transition from exponential decay to oscillatory behavior of $s(z)$.  This is the negative root of the quadratic equation $q(z)=l(z)$.
\item $z=z_\infty:=\frac{1}{4}(1-\alpha\delta-\alpha\mu)$, a point at which two subdominant exponentials in $|s(z)|^2$ exchange roles.  This point is the negative root of the quadratic equation $q(z)=c(z)$ and it also has significance with respect to the \emph{turning point curve} to be explained in \S\ref{sec:tpc}. 
\item $z=0$, the point of transition from oscillatory behavior to exponential growth of $s(z)$.
This is the unique root of the linear equation $l(z)=c(z)$.
\item $z=z^+$, another point at which two subdominant exponentials in $|s(z)|^2$ exchange roles.
This point is the positive root of the quadratic equation $q(z)=l(z)$ and is given explicitly by
\eqref{eq:z-plus-minus}.
\item $z=z_\mathrm{P}:=\frac{1}{4}(1-\alpha\delta+\alpha\mu)$, the point of transition from exponential growth to exponential decay of $s(z)$.  This is the positive root of the quadratic equation $q(z)=c(z)$.
\end{itemize}
That $z^\pm$ are real and distinct follows
from condition \eqref{eq:condition-no-eigenvalues} along with $\alpha>0$, and that they have opposite signs is exactly equivalent to the condition \eqref{eq:condition-transsonic}.  Note also that conditions \eqref{eq:condition-transsonic} and \eqref{eq:condition-no-zeros} together imply that $z^+<z_\mathrm{P}$.  For later convenience, let us define the linear exponent $f(z)$ by
\begin{equation}
f(z):=\frac{2\pi}{\alpha}(z-z_\mathrm{P}).
\label{eq:fexpdefine}
\end{equation}

The assumptions \eqref{eq:condition-transsonic}--\eqref{eq:condition-no-zeros} imply that simple bounds for $|s(z)|^2$ then follow from the exact formula \eqref{eq:modssquared} and the central graph in Figure~\ref{fig:zregions}.  Indeed, since $\cosh^2(\pi R/2i\epsilon)\ge\cosh^2(\pi\mu/2\epsilon)$ and $\cosh^2(\pi\Omega/i\epsilon)-1\ge 0$ for $z\le 0$, we easily obtain the inequality:
\begin{equation}
|s(z)|^2\le 1,\quad z\le 0.
\label{eq:sboundznegative}
\end{equation}
On the other hand, for $z\ge 0$ we have instead that
\begin{equation}
\begin{split}
|s(z)|^2&\le\frac{\cosh^2(\frac{\pi\mu}{2\epsilon})}{\cosh^2(\frac{\pi\Omega}{i\epsilon})} \\ &= e^{4\pi(\alpha\mu/4-|z-(1-\alpha\delta)/4|)/(\alpha\epsilon)}\left(\frac{1+e^{-\pi\mu/\epsilon}}{1+e^{-4\pi|z-(1-\alpha\delta)/4|/(\alpha\epsilon)}}\right)^2\\ & \le 4e^{4\pi(\alpha\mu/4-|z-(1-\alpha\delta)/4|)/(\alpha\epsilon)}\\ &
\le 4e^{4\pi(z_\mathrm{P}-z)/(\alpha\epsilon)}\\
&=4e^{-2f(z)/\epsilon},\quad z\ge 0.
\end{split}
\label{eq:sboundzpositive}
\end{equation}
Note that by the Mean Value Theorem applied to \eqref{eq:modssquared}, the estimates \eqref{eq:sboundznegative} and \eqref{eq:sboundzpositive} can be replaced by an improved estimate valid in a neighborhood of the origin as follows:  for each sufficiently small $\eta>0$ there exist constants $C>0$ and $K>0$ independent of $\epsilon$ such that
\begin{equation}
|s(z)|^2\le \frac{C|z|}{\epsilon}e^{-K/\epsilon},\quad |z|<\eta.
\label{eq:sboundorigin}
\end{equation}
Here we have used the fact that $z_\mathrm{P}>0$ to ensure that $K>0$.

Similar reasoning produces the following asymptotic formulae involving $|s(z)|^2$.  Firstly, we have
\begin{equation}
1-|s(z)|^2 = e^{-\tau(z)/\epsilon}(1+\text{exponentially small in $\epsilon$}),\quad \text{uniformly for 
$z\in [z_\mathrm{L}+\eta,z_\infty-\eta]\cup[z_\infty+\eta,0]$}
\label{eq:oneminusmodssquared}
\end{equation}
for all $\eta>0$,
where $\tau(z)>0$ is defined by
\begin{equation}
\tau(z)=\begin{cases}
\displaystyle\tau_{z_\mathrm{L}}(z):=-i\pi R - 2\pi|\Omega|=\pi\sqrt{\mu^2-16z}+\frac{4\pi}{\alpha}\left(z-\frac{1-\alpha\delta}{4}\right),&\displaystyle \; z_\mathrm{L}<z\le z_\infty\\\\
\tau_0(z):=-i\pi R -\pi\mu=\pi\sqrt{\mu^2-16z}-\pi\mu,&\displaystyle\; z_\infty\le z<0.
\end{cases}
\label{eq:taudef}
\end{equation}
The function $\tau(z)$ extends (by zero) to a continuous function for $z\in\mathbb{R}$.
Secondly, we have
\begin{equation}
|s(z)|^2=e^{-2f(z)/\epsilon}(1+\text{exponentially small in $\epsilon$}),\quad
\text{uniformly for $z>z^++\eta$, $\forall \eta>0$.}
\label{eq:modssquaredlargez}
\end{equation}

\subsubsection{Stirling asymptotics for $s(z)$}
Under conditions \eqref{eq:condition-transsonic}--\eqref{eq:condition-no-zeros}, Stirling's formula \cite{DLMF} yields accurate asymptotics for  $s(z)$ that are uniformly accurate to a relative error of order $O(\epsilon)$ in the closed set $S$ consisting of the closed upper half-plane with the following sets omitted:
\begin{itemize}
\item The vertical strip $|\Re\{z\}-z_\mathrm{P}|<\eta$ with $\Im\{z\}\ge 0$.
This strip contains the accumulation locus of the phantom poles.  The Stirling approximation of
the ratio $\Gamma(\tfrac{1}{2}-\tfrac{i\mu}{2\epsilon}-\tfrac{\Omega}{\epsilon})/\Gamma(\tfrac{1}{2}+\tfrac{i\mu}{2\epsilon}+\tfrac{\Omega}{\epsilon})$ fails (only) in this strip as $\epsilon\downarrow 0$.
\item  The open upper half-disk of radius $\eta$ centered at $z=0$.  The point $z=0$ is the unique root
of $R-i\mu$ for $\Im\{z\}\ge 0$ when $\mu>0$ (and $R+i\mu$ is nonvanishing for $\Im\{z\}\ge 0$).
The Stirling approximation of the product $\Gamma(-\tfrac{i\mu}{2\epsilon}-\tfrac{R}{2\epsilon})\Gamma(-\tfrac{i\mu}{2\epsilon}+\tfrac{R}{2\epsilon})$ fails (only) in
this half-disk as $\epsilon\downarrow 0$.
\item The open upper half-disks of radius $\eta$ centered at the real points $z^-<0<z^+$ defined by \eqref{eq:z-plus-minus}.
These are the roots of the quadratic equation $R^2=(2\Omega)^2$.  The Stirling approximation of the
product $\Gamma(\tfrac{1}{2}+\tfrac{\Omega}{\epsilon}-\tfrac{R}{2\epsilon})\Gamma(\tfrac{1}{2}+\tfrac{\Omega}{\epsilon}+\tfrac{R}{2\epsilon})$ fails (only) in the union of these half-disks as $\epsilon\downarrow 0$.
\end{itemize}
Here $\eta$
is any arbitrarily small fixed value.  The omitted vertical strip containing the phantom poles divides $S$ into two disjoint
subregions:  $S_\mathrm{L}$ to the left of the strip and $S_\mathrm{R}$ to the right.  In each of these
two subregions Stirling's formula will supply a different analytic approximation to the ratio $s(z)$.
The resulting formulae are as follows.  Firstly, define (here $\log(\cdot)$ denotes the principal branch with $|\Im\{\log(\cdot)\}|<\pi$)
\begin{equation}
\begin{split}
\Phi(z)&:=-\mu\log(2)+\frac{4}{\alpha}(z-z_\mathrm{P})\log\left(\frac{2}{\alpha}(z_\mathrm{P}-z)\right)\\
&{}\quad\quad
+\frac{\mu-iR}{2}\log\left(\frac{-i\mu-R}{2}\right)+\frac{\mu+iR}{2}\log\left(\frac{-i\mu+R}{2}\right)\\
&{}\quad\quad
-i\left(\Omega-\frac{R}{2}\right)\log\left(\Omega-\frac{R}{2}\right)-i
\left(\Omega+\frac{R}{2}\right)\log\left(\Omega+\frac{R}{2}\right),\quad\Im\{z\}> 0.
\end{split}
\label{eq:PhidefStirling}
\end{equation}
Under the conditions \eqref{eq:condition-transsonic}--\eqref{eq:condition-no-zeros}, $\Phi(z)$ is an analytic function of $z$ for $\Im\{z\}>0$. 
Now, define:
\begin{equation}
E(z):=
-
s(z)e^{-i\Phi(z)/\epsilon},\quad
\text{for $\Im\{z\}> 0$.}
\end{equation}
The function $E(z)$ is meromorphic for $\Im\{z\}>0$ and its poles are confined to the line $\Re\{z\}=z_\mathrm{P}$.  In particular it is analytic in the interior of $S=S_\mathrm{L}\cup S_\mathrm{R}$.
Stirling's formula implies that
\begin{equation}
E(z)=\begin{cases}
1+O(\epsilon),&\quad z\in S_\mathrm{L},\\
e^{
-2f(z)/\epsilon
}\left(1+O(\epsilon)\right),&\quad z\in S_\mathrm{R},
\end{cases}
\end{equation}
where $f(z)$ is defined by \eqref{eq:fexpdefine},
with both estimates holding uniformly in the indicated region 
(including for the boundary value $E_+(z)$ taken on the real axis).
Next, define
\begin{equation}
\tilde{E}(z):=\left(1+e^{
2f(z)/\epsilon
}\right)E(z),\quad
\Im\{z\}> 0.
\label{eq:tildeEdef}
\end{equation}
The explicit prefactor exactly cancels the poles in
$E(z)$ for $\Im\{z\}>0$, and hence $\tilde{E}(z)$ 
is
an analytic function in the whole upper half complex plane.  Letting $\tilde{S}$ denote the closed upper half-plane with
the open upper half-disks of radius $\eta$ centered at $z=z_\mathrm{L}$, $z=0$, $z=z^+$, and
$z=z_\mathrm{P}$ omitted, we can obtain from Stirling's formula (by using first the reflection identities \eqref{eq:GammaReflection} for $\Gamma(\cdot)$) that 
\begin{equation}
\tilde{E}(z)=1+O(\epsilon),\quad z\in\tilde{S}.
\label{eq:tildeEasymp}
\end{equation}
Again, the estimate holds uniformly 
(including for the boundary value $\tilde{E}_+(z)$).
 It will be
important later to record the imaginary part of the boundary values taken by $\Phi$ on the real axis:
\begin{equation}
\Im\{\Phi_+(z)\}=\begin{cases}
\displaystyle \frac{i\pi}{2}\left(R(z)-2\Omega(z)\right),&\quad
z<z_\mathrm{L}=z^-\\\\
0,&\quad
z_\mathrm{L}=z^-<z<0\\\\
\displaystyle -\frac{i\pi}{2}\left(R(z)-i\mu\right),&\quad 0<z<z^+\\\\
\displaystyle -\frac{2\pi}{\alpha}|z-z_\mathrm{P}|
=-|f(z)|,
&\quad z>z^+.
\end{cases}
\label{eq:ImPhiPlus}
\end{equation}
It therefore follows that $\Phi(z)$ extends through the interval $z_\mathrm{L}<z<0$ to an analytic function defined also in the open lower half-plane, and henceforth we consider $\Phi(z)$ to be
analytic in the slit domain $z\in\mathbb{C}\setminus((-\infty,z_\mathrm{L}]\cup[0,+\infty))$ and to satisfy
the Schwarz symmetry condition $\Phi(z^*)=\Phi(z)^*$.
By direct calculation, one sees from \eqref{eq:PhidefStirling} that
\begin{equation}
\Phi'(z)= -\frac{2\pi i}{\alpha} -\pi z^{-1/2} -\frac{\mu}{2}z^{-1} +O(z^{-3/2}),\quad z\to\infty,\quad \Im\{z\}>0.
\label{eq:PhiPrimeExpansion}
\end{equation}

An important observation is that while here the formulae \eqref{eq:taudef} and \eqref{eq:PhidefStirling} for $\tau(z)$ and $\Phi(z)$ respectively come directly from analysis of the exact formula \eqref{eq:s-define} for the
reflection coefficient valid in the case of the particular initial conditions under consideration, these formulae agree \emph{exactly} with predictions based on WKB theory formally valid for more general initial data.
This formal semiclassical spectral analysis of the direct scattering problem can be found in
all details in \S2 of \cite{DiFrancoMM11}.  The more general formulae corresponding to \eqref{eq:taudef}
and \eqref{eq:PhidefStirling} arise from WKB theory as follows.
The \emph{turning point curve} of WKB theory is, by definition, the locus of real $x$ and complex $z$ satisfying the characteristic equation
\begin{equation}
\chi(x;z):=(4z-1+\alpha u_0(x))^2 + 16\alpha^2z\rho_0(x)=0.
\label{eq:TPCgeneral}
\end{equation}
Let $z\in\mathbb{R}$ be a real value for which the function $\chi(x;z)$ has real roots $x$, and let
$x_-(z)$ denote the most negative of these roots.  In particular, $\chi(x;z)>0$ for all $x<x_-(z)$.
Define signs $\sigma_\pm$ by
\begin{equation}
\sigma_\pm:=\mathrm{sgn}(4z-1+\alpha u_\pm),
\label{eq:sigmapmdef}
\end{equation}
and then set
\begin{equation}
\omega(x;z):=\sigma_-\sqrt{\chi(x;z)},\quad x<x_-(z).
\label{eq:omegadef}
\end{equation}
For each appropriate value of $z\in\mathbb{R}$, this quantity has the fixed sign $\sigma_-$ throughout its domain of definition, and $\omega(x;z)\to \omega_-(z)$ as $x\to -\infty$,
where
\begin{equation}
\omega_-(z):=4z-1+\alpha u_-.
\label{eq:omegaminusdef}
\end{equation}
Then, for all real $z$ for which $\chi(x;z)$ has real roots, we 
let
\begin{equation}
\gamma(x;z):=\sqrt{-\chi(x;z)}>0,\quad \chi(x;z)\le 0,
\label{eq:gammadef}
\end{equation}
and then 
\begin{equation}
\tau(z) 
= \frac{1}{\alpha}
\int_{y:\chi(y;z)<0}\gamma(y;z)\,dy
\label{eq:tauWKB}
\end{equation}
is a generalized formula that agrees exactly with \eqref{eq:taudef} when $A_0(x)$ and $S_0(x)$
are the particular functions \eqref{eq:data} considered in this paper.
Next, we have
\begin{equation}
\Phi(z)= 
\frac{1}{\alpha}\int_{-\infty}^{x_-(z)}\left[\omega(y;z)-\omega_-(z)\right]\,dy +\frac{1}{\alpha}\omega_-(z)x_-(z) +S_-.
\label{eq:PhidefWKB}
\end{equation}
This formula is a generalization of \eqref{eq:PhidefStirling} in the sense that the latter is the analytic continuation to the upper half $z$-plane of \eqref{eq:PhidefWKB} in the special case
of the functions $A_0(x)$ and $S_0(x)$ given by \eqref{eq:data}.    For these functions it turns out that there
are generically either two distinct roots (turning points) $x_-(z)<x_+(z)$ of $\chi(x;z)$ or none for $z\in\mathbb{R}$.  In such cases it can be shown that if $z_0$ is a real value at which $x_\pm(z)$
coalesce in a square-root sense, then the functions $\Phi(z)$ and $\tau(z)$ can be analytically continued about such a point, and the following monodromy relations hold:
\begin{equation}
\Phi(z_0+(z-z_0)e^{\pm 2\pi i})=\Phi(z)\pm i\sigma_-\tau(z),
\label{eq:Phimonodromy}
\end{equation}
\begin{equation}
\tau(z_0+(z-z_0)e^{\pm 2\pi i})=\tau(z)
\label{eq:taumonodromy}
\end{equation}
We will make use of  the concrete formulae \eqref{eq:taudef} and \eqref{eq:PhidefStirling} and also the equivalent WKB formulae  \eqref{eq:tauWKB} and \eqref{eq:PhidefWKB} at various points in our analysis.

\section{Riemann-Hilbert Problem of Inverse Scattering}
\label{sec:RHP}
\subsection{Basic problem}
We take as a starting point the formulation of the inverse-scattering problem as a matrix Riemann-Hilbert problem for an unknown $\mathbf{M}(k;x,t)$ as described in the Appendix of \cite{DiFrancoM08}.  
This matrix has jump discontinuities across the real and imaginary axes in the complex $k$-plane
and satisfies the holomorphic involution $\mathbf{M}(-k;x,t)=i^{\sigma_3}\mathbf{M}(k;x,t)i^{-\sigma_3}$ 
and the antiholomorphic involution $\mathbf{M}(-k^*;x,t)^*=\sigma_1\mathbf{M}(k;x,t)\sigma_1$.
If for $\Im\{k^2\}=0$, $\mathbf{M}_\pm(k;x,t)$ denotes the boundary value taken by $\mathbf{M}$
from a quadrant where $\pm\Im\{k^2\}<0$, then 
\begin{equation}
\mathbf{M}_+(k;x,t)=\mathbf{M}_-(k;x,t)e^{i\theta(k^2;x,t)\sigma_3/\epsilon}\begin{bmatrix}
1\pm |r(k)|^2 & r(k)\\\pm r(k)^* & 1\end{bmatrix}e^{-i\theta(k^2;x,t)\sigma_3/\epsilon},\quad \pm k^2>0,
\end{equation}
where
\begin{equation}
\theta(z)=\theta(z;x,t):=-\frac{2}{\alpha}\left(z-\frac{1}{4}\right)x -\frac{4}{\alpha^2}\left(z-\frac{1}{4}\right)^2t
\label{eq:thetadef}
\end{equation}
and where $r(k)=-r(-k)$ is the reflection coefficient defined 
generally in terms of the scattering coefficients $S_{12}(k)$ and $S_{22}(k)$ by \eqref{eq:reflectioncoefficient1}.
The matrix $\mathbf{M}(k;x,t)$ is normalized to the identity at $k=0$:  
\begin{equation}
\lim_{k\to 0}\mathbf{M}(k;x,t)=\mathbb{I}
\end{equation}
regardless of the quadrant from which the limit is taken, and is required to have a well-defined
limiting value (a matrix-valued function of $(x,t)$) as $k\to\infty$.  In the case under study in this paper that there
are no eigenvalues of the scattering problem, the matrix $\mathbf{M}(k;x,t)$ is analytic for  $\Im\{k^2\}\neq 0$ and the above properties uniquely determine $\mathbf{M}(k;x,t)$ given $r(\cdot)$.  The corresponding solution of the MNLS equation is then obtained from $\mathbf{M}(k;x,t)$ by the
formula
\begin{equation}
\phi_\epsilon(x,t):=\lim_{k\to\infty}\frac{2k}{\alpha}\frac{M_{12}(k;x,t)}{M_{22}(k;x,t)}.
\label{eq:phirecover}
\end{equation}
(Note that $\mathbf{M}(k;x,t)$ is asymptotically diagonal as $k\to\infty$.  See \cite[page 990]{DiFrancoM08}.)

\subsection{Symmetry reduction}
It will be convenient to take advantage of the $k\mapsto -k$ symmetry to define a new unknown by setting
$\mathbf{N}(z;x,t):=\mathbf{M}(i(-z)^{1/2};x,t)$
for $\Im\{z\}\neq 0$.  
Since 
$(-z)^{1/2}$
denotes the principal
branch of the square root, the matrix $\mathbf{M}(k;x,t)$ as defined in the first 
(second)
quadrant of the $k$-plane is equivalent to the matrix $\mathbf{N}(z;x,t)$ in the upper (lower) half $z$-plane.  For $z\in\mathbb{R}$, let
$\mathbf{N}_\pm(z;x,t)$ denote the boundary value of $\mathbf{N}$ taken from $\mathbb{C}_\pm$.
Then, if $z>0$, since 
$\mathbf{N}_\pm(z;x,t)=\mathbf{M}_\mp(\pm z^{1/2};x,t)$ and according to the holomorphic symmetry of $\mathbf{M}(k;x,t)$ we have $\mathbf{M}_+(-z^{1/2};x,t)=i^{\sigma_3}\mathbf{M}_+(z^{1/2};x,t)i^{-\sigma_3}$,
\begin{equation}
\mathbf{N}_+(z;x,t)=
i^{\sigma_3}
\mathbf{N}_-(z;x,t)
i^{-\sigma_3}
e^{i\theta(z;x,t)\sigma_3/\epsilon}
\begin{bmatrix}1 & -r(z^{1/2})\\-r(z^{1/2})^* &1+|r(z^{1/2})|^2\end{bmatrix}e^{-i\theta(z;x,t)\sigma_3/\epsilon},\quad z>0.
\label{eq:Njumpzgt0}
\end{equation}
On the other hand, if $z<0$, then 
$\mathbf{N}_\pm(z;x,t)=\mathbf{M}_\mp(i(-z)^{1/2};x,t)$,
so
\begin{equation}
\mathbf{N}_+(z;x,t)=
\mathbf{N}_-(z;x,t)
e^{i\theta(z;x,t)\sigma_3/\epsilon}
\begin{bmatrix}1 & -r(i(-z)^{1/2})\\r(i(-z)^{1/2})^* & 1-|r(i(-z)^{1/2})|^2\end{bmatrix}
e^{-i\theta(z;x,t)\sigma_3/\epsilon},\quad z<0.
\label{eq:Njumpzlt0}
\end{equation}
The function $e^{i\pi/4}(-iz)^{1/2}$ (principal branch) is analytic in the upper half $z$-plane and agrees with $z^{1/2}$ for $z>0$ and with $i(-z)^{1/2}$ for $z<0$.  
Note that both $r(z^{1/2})$ for $z>0$ and $r(i(-z)^{1/2})$ for $z<0$ correspond to the same function $s(z)$ defined for all $z\in\mathbb{R}$ by \eqref{eq:s-define}.
We frequently omit the explicit dependence on the parameters $(x,t)$ (along with $\epsilon$, $\alpha$, $\delta$, and $\mu$) and simply write $\mathbf{N}(z)=\mathbf{N}(z;x,t)$.  The matrix $\mathbf{N}(z)$ necessarily satisfies the Schwarz symmetry property:
\begin{equation}
\mathbf{N}(z^*)=\sigma_1\mathbf{N}(z)^*\sigma_1
\label{eq:NSchwarz}
\end{equation}
as well as the normalization condition
\begin{equation}
\mathop{\lim_{z\to 0}}_{\Im\{z\}\neq 0}\mathbf{N}(z)=\mathbb{I},
\label{eq:Nnorm}
\end{equation}
and we require $\mathbf{N}(z)$ to have a well-defined (necessarily diagonal) limiting value as $z\to\infty$.

\section{The Functions $g(z)$ and $h(z)$}
\label{sec:gfunction}
In this section, we construct two related analytic functions of $z$ that will depend parametrically
on $(x,t)\in\mathbb{R}^2$.  For notational convenience, we will frequently not write the parameters explicitly as arguments of these functions, but the reader should be aware that the dependence is there.  We wish to stress from the outset that \emph{all objects considered in this section are independent of the basic asymptotic parameter $\epsilon>0$}.
\subsection{The WKB turning point curve.}
\label{sec:tpc}
In the case of the special initial data under consideration the turning point curve defined in general by \eqref{eq:TPCgeneral} is given by the
equation
\begin{equation}
(Z-A+BT)^2+Z(1-T^2)=0,\quad z=\alpha^2 Z,\quad T:=\tanh(x),
\label{eq:tpc}
\end{equation}
where the parameters $A$ and $B$ are defined by \eqref{eq:ABdefs}.  This is quadratic in $Z$ for fixed $T$ and also is quadratic in $T$ for fixed $Z$.  In solving for $Z$, one easily finds that the discriminant
is $-(1-T^2)(T^2-4BT+4A-1)$, whose second factor is proportional by a negative constant to the quantity $Q$ at $t=0$ as defined by \eqref{eq:specialQ}.  Therefore under
the conditions \eqref{eq:condition-transsonic}--\eqref{eq:condition-no-zeros}, the only (simple) root in the range $|T|<1$ is $T=T_\mathrm{c}$ defined by \eqref{eq:Tcrit}.
(The other root of the quadratic factor exceeds $T=1$ for $(A,B)$ in the admissible region shown in Figure~\ref{fig:ABplane}.)  Moreover, the discriminant is positive (only) for $T\in(T_\mathrm{c},1)$.  The point $x_\mathrm{c}=\mathrm{arctanh}\,(T_\mathrm{c})$, corresponding to $z=z_\mathrm{c}:=(\alpha (\delta+\mu T_\mathrm{c})-1)/4 = 
(\alpha u_0(x_\mathrm{c})-1)/4$, is the only point on the 
turning point curve in the $(x,z)$-plane where there is a vertical tangent.  For $x\ge x_\mathrm{c}$
the turning point curve has two real branches that we denote as $z=\mathfrak{a}(x)$ and
$z=\mathfrak{b}(x)$ with $\mathfrak{a}(x)\le\mathfrak{b}(x)$, with equality holding only at $x=x_\mathrm{c}$ and in the limit $x\to +\infty$ as illustrated in Figure~\ref{fig:tpc}.

Horizontal tangents to the turning point curve correspond to simultaneous solutions of \eqref{eq:tpc}
and the identity
\begin{equation}
B(Z-A+BT)-ZT=0
\end{equation}
obtained from \eqref{eq:tpc} by differentiating implicitly with respect to $T$ and setting $dZ/dT=0$.  Using this latter relation to eliminate $Z$ from \eqref{eq:tpc} we obtain an equation for $T$-values of horizontal tangents in the form
\begin{equation}
\frac{BT-A}{(T-B)^2}\left[(B^2-A)T^2+BT-B^2\right]=0
\end{equation}
Therefore, either $T=A/B$ or $T=T_\pm(A,B)$ where
\begin{equation}
T_\pm(A,B):=\pm\frac{2B}{\sqrt{4B^2-4A+1}\pm 1}.
\end{equation}
Since $B\ge 0$ and $B^2\ge A$ both hold throughout the domain of values of $(A,B)$ consistent with
conditions \eqref{eq:condition-transsonic}--\eqref{eq:condition-no-zeros} (see Figure~\ref{fig:ABplane}), it is clear that $T_+(A,B)\ge 0$ while $T_-(A,B)\le 0$.  Moreover, $T_+(A,B)=1$ implies that $B=A$ while $T_-(A,B)=-1$ implies that $B=-A$, both conditions that only occur on the boundary of the admissible region for $(A,B)$.   Also, from the asymptotic expansions of $T_\pm(A,B)$ in the limit  $B\uparrow\infty$ with $A\in\mathbb{R}$ fixed one concludes
that $0\le T_+(A,B)<1$ and $T_-(A,B)<-1$ both hold in the interior of the admissible region, and therefore of $T_\pm(A,B)$ 
only $T_+(A,B)$ corresponds to a value of $x\in\mathbb{R}$.  Since $B>|A|$ holds in the interior of the
admissible region, $T=A/B$ also corresponds to a real value of $x$.  Similar analysis proves that both
$T=A/B$ and $T=T_+(A,B)$ exceed $T_\mathrm{c}$ for $(A,B)$ in the admissible region, so only the real
part of the turning point curve for $T>T_\mathrm{c}$ has any horizontal tangents (defined by $dZ/dT=0$),
and there are exactly two of them.  Also, it is easy to check that throughout the admissible region of the $(A,B)$-plane,
the horizontal tangent at $T=A/B$ occurs on the more positive of the two branches of $Z(T)$ (with value $z=0$) while the horizontal tangent at $T=T_+(A,B)$ occurs on the more negative of the two branches (with value $z_\mathrm{L}$).  
\begin{figure}[h]
\begin{center}
\includegraphics{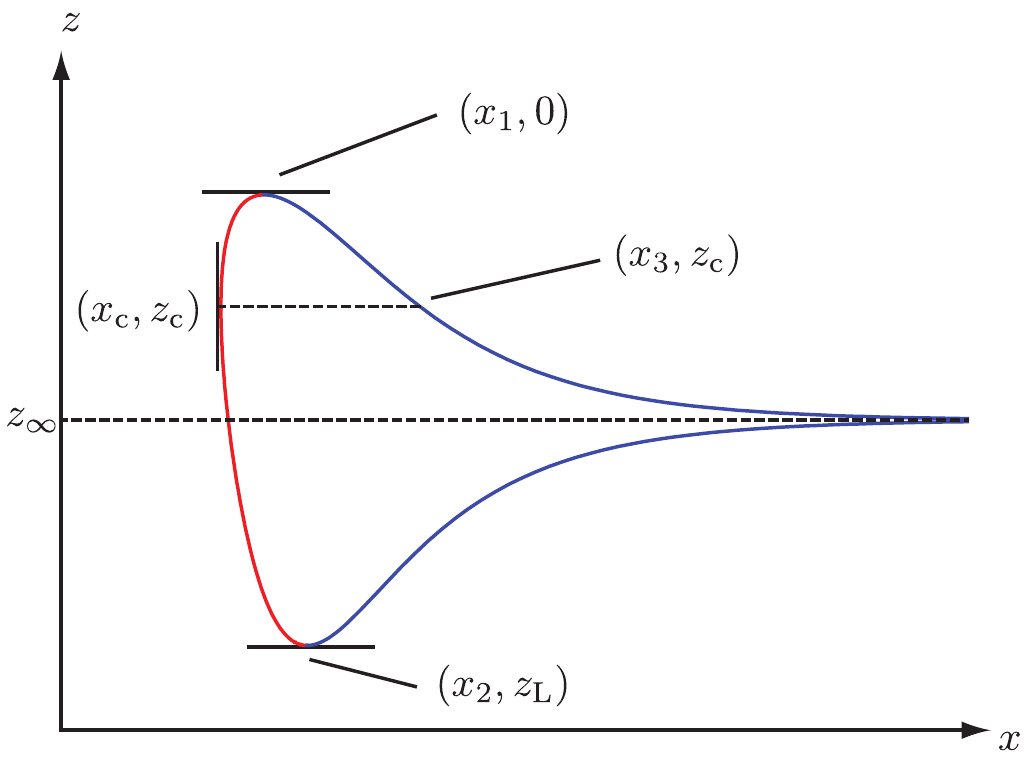}
\end{center}
\caption{The real part of the turning point curve for $(A,B)$ in the admissible region described by conditions \eqref{eq:condition-transsonic}--\eqref{eq:condition-no-zeros}.  The red part of the curve is the graph of a function $x=x_-(z)$ and the blue part of the curve is the graph of a function $x=x_+(z)$ both over the interval $z_\mathrm{L}<z<0$.  For such $z$, $x_\pm(z)$
are real turning points of WKB theory.  The upper curve is the graph of a function $z=\mathfrak{b}(x)$ and the lower curve is the graph of a function $z=\mathfrak{a}(x)$, both defined for $x\ge x_\mathrm{c}$.  The singularity of $x_+(z)$ occurs when $Z=A-B$ or $z=z_\infty:=(1-\alpha\delta-\alpha\mu)/4$, the break point of the equation \eqref{eq:taudef} for $\tau(z)$ in the interval $z_\mathrm{L}<z<0$.}
\label{fig:tpc}
\end{figure}
We define $x_1$ and $x_2$ by
\begin{equation}
x_1:=\mathrm{arctanh}\left(\frac{A}{B}\right)\quad\text{and}\quad x_2:=\mathrm{arctanh}(T_+(A,B)).
\end{equation}
Note that under the conditions \eqref{eq:condition-transsonic}--\eqref{eq:condition-no-zeros}, we have the inequalities $x_2>x_1>x_\mathrm{c}$.
It is also possible to show that these conditions imply that $z_\mathrm{c}\ge z_\infty:=(1-\alpha\delta-\alpha\mu)/4$ ($z_\infty$ as originally defined in \S\ref{sec:asymptotic-properties} is also the point at which the branches of the turning point curve coalesce in the limit $x\to +\infty$, or equivalently, the singularity of $x_+(z)$), with equality occurring only along the part of the boundary of the admissible region of the
parameter space illustrated in Figure~\ref{fig:ABplane} with $A\ge 0$.  This in turn implies that there
exists exactly one value of $x>x_1$ that we denote by $x_3$, at which point $\mathfrak{b}(x_3)=z_\mathrm{c}$.  For $x_\mathrm{c}<x<x_3$ we have $\mathfrak{a}(x)<z_\mathrm{c}<\mathfrak{b}(x)$ and for $x>x_3$ we have $\mathfrak{a}(x)<\mathfrak{b}(x)<z_\mathrm{c}$.
It also follows from the conditions \eqref{eq:condition-transsonic}--\eqref{eq:condition-no-zeros}  that  in fact $x_3>x_2$.
These relationships are illustrated in Figure~\ref{fig:tpc}.

For $x<x_\mathrm{c}$, the turning point curve is complex, with two distinct complex-conjugate branches:
$z=\mathfrak{z}(x)$ with $\Im\{\mathfrak{z}(x)\}>0$ and $z=\mathfrak{z}(x)^*$.  The two branches coincide (with a real value)
only in the limits $x\uparrow x_\mathrm{c}$ (with value $z=z_\mathrm{c}$) and $x\downarrow -\infty$
(with value $Z=A+B>0$ or $z=z_\mathrm{P}$).  Moreover, it is easy to prove that for $(A,B)$ in the admissible region, $\Re\{\mathfrak{z}(x)\}$ is monotone decreasing in $x$ for $-\infty<x<x_\mathrm{c}$,
and 
\begin{equation}
\sup_{x<x_\mathrm{c}}\Re\{\mathfrak{z}(x)\}=\lim_{x\to -\infty}\Re\{\mathfrak{z}(x)\} = z_\mathrm{P}.
\end{equation}
In particular, $\Re\{\mathfrak{z}(x)\}<z_\mathrm{P}$ for all finite $x<x_\mathrm{c}$.  Note also that
for $x<x_\mathrm{c}$, 
\begin{equation}
\Re\{\mathfrak{z}(x)\} = -\frac{1}{4}\left(\alpha u_0(x)-1+2\alpha^2\rho_0(x)\right).
\label{eq:Rezeta}
\end{equation}

\subsection{Basic construction of $g(z)$ and $h(z)$.  The auxiliary function $Y(z)$.}
\label{sec:BasicFormulae}
Let a complex number $q$ with $\Im\{q\}>0$ be given, along with a simple oriented arc $\mathcal{B}$ from $z=q$ to $z=z_\mathrm{P}$ in the open upper half-plane.  Let $S(z)=S(z;q,q^*,\mathcal{B})$ be the function uniquely defined by the following properties:
\begin{itemize}
\item $S(z)$ is defined and analytic for $z\in\mathbb{C}\setminus (\mathcal{B}\cup \mathcal{B}^*)$.
\item $S(z)^2 = (z-q)(z-q^*)$.
\item $S(z)=z + O(1)$ as $z\to\infty$.
\end{itemize}
At times we will want to think of $q$ and $q^*$ as being independent complex variables, but when they are linked by complex conjugation, $S(z)$ is obviously a Schwarz-symmetric function:  $S(z^*)=S(z)^*$.  Note that $S(z)$ changes sign across its branch cut $\mathcal{B}\cup \mathcal{B}^*$.  

Let us assume (this will be completely clarified in \S\ref{sec:q}) that given $(x,t)\in\mathbb{R}^2$,
the complex number $q$ satisfies the equations
\begin{equation}
\begin{split}
M_0(q,q^*;x,t)&:=\frac{2\pi}{\alpha^2}(2(q+q^*)t+\alpha x-t)+I_0(q,q^*)=0\\
M_1(q,q^*;x,t)&:=\frac{\pi}{\alpha^2}((3q^2+2|q|^2+3q^{*2})t+(q+q^*)(\alpha x-t))+I_1(q,q^*)=0,
\end{split}
\label{eq:momentscx}
\end{equation}
where
\begin{equation}
I_p(q,q^*):=\Im\left\{\int_\mathcal{B}\frac{\Phi'(s)s^p\,ds}{S_+(s)}\right\},\quad p=0,1.
\label{eq:Ip}
\end{equation}
As usual, the notation $S_+(z)$ denotes the boundary value taken from the left side of $\mathcal{B}$ as the arc is traversed according to its orientation from $q$ to $z_\mathrm{P}$.  Now, let $p(z)$ be the function analytic for $z\in\mathbb{C}\setminus (\mathcal{B}\cup \mathcal{B}^*)$ given by the Cauchy-type integral formula
\begin{equation}
p(z):=\frac{S(z)}{2\pi i}\left[\int_\mathcal{B}\frac{\Phi'(s)\,ds}{S_+(s)(s-z)}-(*)^*\right]
+\theta'(z) +\frac{8tS(z)}{\alpha^2}, \quad z\in\mathbb{C}\setminus (\mathcal{B}\cup \mathcal{B}^*),
\label{eq:gprime1}
\end{equation}
where the notation $(*)^*$ indicates the Schwarz reflection $w(z^*)^*$ of the function $w(z)$ immediately preceding the minus sign, and where $\theta(z)=\theta(z;x,t)$ is defined by \eqref{eq:thetadef}.  The conditions \eqref{eq:momentscx} obviously imply that
\begin{equation}
p(z)=O\left(\frac{1}{z^2}\right),\quad z\to\infty,
\end{equation}
and hence the contour integral
\begin{equation}
g(z):=\int_0^z p(s)\,ds,\quad z\in\mathbb{C}\setminus (\mathcal{B}\cup \mathcal{B}^*)
\end{equation}
is independent of path (as long as the path avoids the branch cut $\mathcal{B}\cup \mathcal{B}^*$ of $p$) and defines a function analytic in the same domain as $p$, with a well-defined limiting value $g(\infty)$.
In terms of $g(z)$ we define a related function $h$ as follows:
\begin{equation}
h(z):=\theta(z)+\frac{1}{2}\Phi(z)-g(z),\quad z\in\mathbb{C}\setminus((-\infty,z_\mathrm{L}]\cup[0,+\infty)\cup \mathcal{B}\cup \mathcal{B}^*).
\label{eq:hdef}
\end{equation}
Because $\theta(z)$ is entire, $h(z)$ is analytic exactly in the intersection of the domains of analyticity of $\Phi(z)$
and $g(z)$ as explicitly indicated.  Both $g$ and $h$ are Schwarz-symmetric functions:  $g(z^*)=g(z)^*$ and $h(z^*)=h(z)^*$.

Let the function $Y(z)=Y(z;q,q^*)$ be defined for $\Im\{z\}>0$ and $z$ near $\mathcal{B}$ by the following Cauchy-type integral:
\begin{equation}
Y(z;q,q^*):=\left[\frac{1}{2\pi i}\oint_L\frac{\Phi'(s)\,ds}{S(s)(s-z)} + (*)^*\right] + \frac{16t}{\alpha^2},
\label{eq:Yfdefine}
\end{equation}
where $L$ is a clockwise (negatively) oriented loop in the upper half $s$-plane beginning and ending at $z_\mathrm{P}$ and encircling both $s=z$ and the arc $\mathcal{B}$ exactly once.  Obviously $Y$ is an analytic function of $z$ in the interior of the loop $L$.  Exploiting more detailed information about $\Phi$ (as is available from the specific formula \eqref{eq:PhidefStirling} but that may or may not be available for more general initial data when one uses instead the WKB formula \eqref{eq:PhidefWKB}) allows us define a global analytic continuation of $Y(z;q,q^*)$ as we will now show.  Indeed, using the fact that $\Phi'(z)$ is analytic in the open upper half-plane and satisfies $\Phi'(z^*)^*=\Phi'(z)$, along with the asymptotic estimate $\Phi'(z)=O(\log|z|)$ as $z\to\infty$ allows us to deform the contour $L$
to a contour along the real axis; therefore
\begin{equation}
Y(z;q,q*)=\frac{16t}{\alpha^2}-\frac{1}{2\pi i}\int_\mathbb{R}\frac{\Phi'_+(s)-\Phi'_-(s)}{S(s)(s-z)}\,ds
=\frac{16t}{\alpha^2}-\frac{1}{\pi}\int_\mathbb{R}\frac{\Im\{\Phi'_+(s)\}}{S(s)(s-z)}\,ds
\end{equation}
defines the analytic continuation of $Y$ to the domain $z\in\mathbb{C}\setminus\mathbb{R}$.
Now, we recall \eqref{eq:ImPhiPlus} along with the definitions \eqref{eq:Omegadef} of $\Omega(z)$ and \eqref{eq:Rdef} of $R(z)$ and the fact that for $s\in\mathbb{R}$,  $S(s)=-\sqrt{(s-\Re\{q\})^2+\Im\{q\}^2}$ for $s<z_\mathrm{P}$ and $S(s)=\sqrt{(s-\Re\{q\})^2+\Im\{q\}^2}$ for $s>z_\mathrm{P}$, and we therefore find that when $\Phi$ is given by \eqref{eq:PhidefStirling},
\begin{multline}
Y(z;q,q^*):=\frac{16t}{\alpha^2}\\{}-\frac{2}{\alpha}\int_{-\infty}^{z_\mathrm{L}}\frac{ds}{\sqrt{(s-\Re\{q\})^2+\Im\{q\}^2}(s-z)} +4\int_{-\infty}^{z_\mathrm{L}}\frac{ds}{\sqrt{\mu^2-16s}\sqrt{(s-\Re\{q\})^2+\Im\{q\}^2}(s-z)}\\
{}-4\int_0^{z^+}\frac{ds}{\sqrt{\mu^2-16s}\sqrt{(s-\Re\{q\})^2+\Im\{q\}^2}(s-z)}+\frac{2}{\alpha}\int_{z^+}^{+\infty}\frac{ds}{\sqrt{(s-\Re\{q\})^2+\Im\{q\}^2}(s-z)}.
\label{eq:Ydef}
\end{multline}
This completes the extension of $Y(z;q,q^*)$ as an analytic function from $z$ near $\mathcal{B}$ to the maximal slit domain $\mathbb{C}_+\cup\mathbb{C}_-\cup (z_\mathrm{L},0)$.

The function $Y(z;q,q^*)$ has many purposes in our analysis.  For example, it is easy to check with the use of elementary contour deformations that the following identities hold:
\begin{equation}
g_+'(z)-g_-'(z)=S_+(z)Y(z;q,q^*),\quad z\in \mathcal{B},
\label{eq:gplusminusprime}
\end{equation}
and
\begin{equation}
2h'(z)=-S(z)Y(z;q,q^*),\quad z\in\mathbb{C}\setminus ((-\infty,z_\mathrm{L}]\cup[0,+\infty)\cup \mathcal{B}\cup \mathcal{B}^*).
\label{eq:2hprime}
\end{equation}
The key property of the analytic function $Y(z;q,q^*)$ that we will require frequently is the following.
\begin{proposition}
For $t> 0$ there exists exactly one simple zero of $Y(z;q,q^*)$ in its domain of analyticity (and hence by Schwarz symmetry necessarily located in the real interval $z_\mathrm{L}<z<0$).  For $t<0$ there exist exactly three zeros  (counting with multiplicities) of $Y(z;q,q^*)$ in its domain of analyticity, and in this case given any $M>0$, for $-t$ sufficiently small exactly one simple zero lies in the real interval $z_\mathrm{L}<z<0$ while there is an additional simple zero in each of the upper and lower half-planes with $|z|>M$. (In the boundary
case of $t=0$ the non-real zeros are at $z=\infty$.)
\label{prop:Yzeroscx}
\end{proposition}
\begin{proof}
The domain of analyticity of $Y(z;q,q^*)$ is the slit domain $z\in\mathbb{C}\setminus ((-\infty,z_\mathrm{L}]\cup[0,+\infty))$, and we note 
that $Y$ as given by \eqref{eq:Ydef} satisfies the Schwarz symmetry condition $Y(z^*;q,q^*)^*=Y(z;q,q^*)$.   From these facts and the Plemelj formula we find that upon taking a boundary value as $z$ approaches the real axis from the upper half-plane,
\begin{equation}
\Im\{Y_+(z;q,q^*)\} = \begin{cases}
\displaystyle \frac{4\pi\alpha-2\pi\sqrt{\mu^2-16z}}{\alpha\sqrt{\mu^2-16z}\sqrt{(z-\Re\{q\})^2+\Im\{q\}^2}},&\quad z<z_\mathrm{L}\\\\
0,&\quad z_\mathrm{L}<z<0\\\\
\displaystyle \frac{-4\pi}{\sqrt{\mu^2-16z}\sqrt{(z-\Re\{q\})^2+\Im\{q\}^2}},&\quad 0<z<z^+\\\\
\displaystyle\frac{2\pi}{\alpha\sqrt{(z-\Re\{q\})^2+\Im\{q\}^2}},&\quad z>z^+.
\end{cases}
\label{eq:ImYPlus}
\end{equation}
Also, from \eqref{eq:Ydef} it is clear that the complex boundary value $Y_+(z;q,q^*)$ is actually
analytic except at the points $z\in\{z_\mathrm{L},0,z^+\}$, where it becomes infinite in magnitude
with logarithmic singularities.  We may combine this information to show that, when the parameters correspond to the admissible region shown in Figure~\ref{fig:ABplane}, the following
statements are valid.
\begin{enumerate}
\item There is a positive constant $c>0$ such that as $z$ increases from $z=0$ to $z=z^+$, $\Im\{Y_+(z;q,q^*)\}\le -c<0$, while $\Re\{Y_+(z;q,q^*)\}\to -\infty$ as $z\downarrow 0$ and $\Re\{Y_+(z;q,q^*)\}\to +\infty$ as $z\uparrow z^+$.
\label{page:cxYstatements}
\item As $z$ increases from $z=z^+$ to $z=+\infty$, $\Im\{Y_+(z;q,q^*)\}\ge 0$ with equality \emph{only} in the limit $z\uparrow +\infty$, while $\Re\{Y_+(z;q,q^*)\}\to +\infty$ as $z\downarrow z^+$
and $\Re\{Y_+(z;q,q^*)\}\to 16t/\alpha^2$ as $z\uparrow +\infty$.
\item As $z$ increases from $z=-\infty$ to $z=z_\mathrm{L}$, $\Im\{Y_+(z;q,q^*)\}\le 0$ with
equality \emph{only} in the limit $z\downarrow -\infty$, while $\Re\{Y_+(z;q,q^*)\}\to 16t/\alpha^2$
as $z\downarrow -\infty$ and $\Re\{Y_+(z;q,q^*)\}\to +\infty$ as $z\uparrow z_\mathrm{L}$.
\end{enumerate}
In particular, as $z$ increases from $z=0$ to $z=+\infty$ and then from $z=-\infty$ to $z=z_\mathrm{L}$, the complex boundary value $Y_+(z;q,q^*)$ traces out a curve in the $Y$-plane
that crosses the real axis only at one finite point, $Y=16t/\alpha^2$, corresponding to $z=\pm\infty$.  The three intervals in the $z$-plane and their images under $Y_+$ in the $Y$-plane are shown in red, with
numbers corresponding to the above enumeration, in Figure~\ref{fig:Ycurve}.
\begin{figure}[h]
\begin{center}
\includegraphics{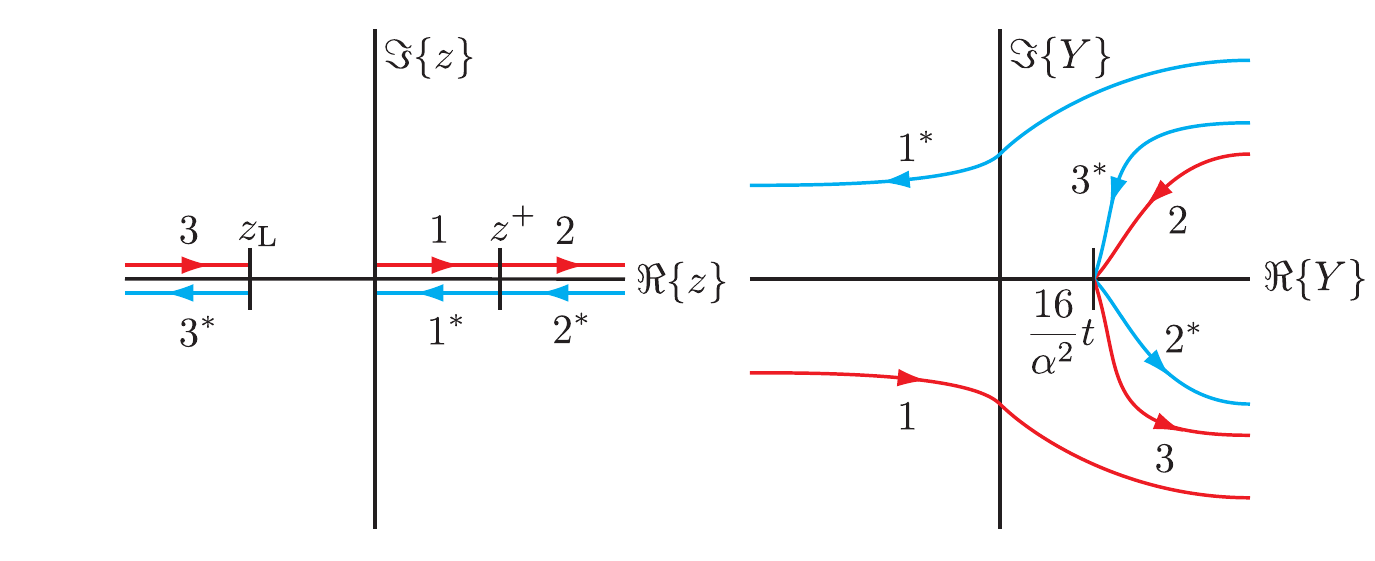}
\end{center}
\caption{The domain of analyticity of $Y(z;q,q^*)$ is $\mathbb{C}\setminus ((-\infty,z_\mathrm{L}]\cup [0,+\infty))$, and a qualitative sketch of the image of the boundary of this slit domain in the $Y$-plane. 
}
\label{fig:Ycurve}
\end{figure}
The desired result then follows from an elementary application of the Argument Principle using the information enumerated
above regarding the boundary value $Y_+(z;q,q^*)$ on the cut (and corresponding information
regarding the boundary value $Y_-(z;q,q^*)$ follows by Schwarz symmetry).  It is obvious
that as the boundary of the slit domain is traversed once in the positive sense (in the order $1$, $2$, $3$, $3^*$, $2^*$, and $1^*$ as indicated in Figure~\ref{fig:Ycurve}), the image curve
in the $Y$-plane encircles the origin exactly once in the positive sense for $t>0$ and exactly three times in the positive sense for $t<0$.  See Figure~\ref{fig:Ycurve}.  This gives the count of the zeros in the domain of analyticity.  The proof is complete upon using continuity of $Y$ with respect to $t$ and Schwarz symmetry.
\end{proof}

It is not difficult to obtain the asymptotic expansion of $Y(z;q,q^*)$ as $z\to\infty$ with $\Im\{z\}>0$.  The 
easiest way to do this is to analytically continue the representation \eqref{eq:Yfdefine} to $z$ outside the loop $L$, by extracting a residue:
\begin{equation}
Y(z;q,q^*)=\left[\frac{1}{2\pi i}\oint_L\frac{\Phi'(s)\,ds}{S(s)(s-z)}+(*)^*\right]-\frac{\Phi'(z)}{S(z)} +\frac{16t}{\alpha^2},\quad \text{$z$ outside of $L$ with $\Im\{z\}>0$.}
\label{eq:Yexpandlast}
\end{equation}
Now we can let $z$ tend to infinity by expanding the
Cauchy kernel in a geometric series for $L$ fixed.  Thus:
\begin{equation}
\begin{split}
\frac{1}{2\pi i}\oint_L\frac{\Phi'(s)\,ds}{S(s)(s-z)} +(*)^* &=- \Im\left\{\frac{1}{\pi}\oint_L\frac{\Phi'(s)\,ds}{S(s)}\right\}\frac{1}{z} + O\left(\frac{1}{z^2}\right)\\
&=-\frac{2}{\pi}\Im\left\{\int_B\frac{\Phi'(s)\,ds}{S_+(s)}\right\}\frac{1}{z}+O\left(\frac{1}{z^2}\right),\quad z\to\infty.
\end{split}
\end{equation}
Assuming that $q$ is chosen so that at least the first of the equations \eqref{eq:momentscx} holds,
we therefore find that
\begin{equation}
 \frac{1}{2\pi i}\oint_L\frac{\Phi'(s)\,ds}{S(s)(s-z)} +(*)^*= \frac{4}{\alpha^2}(2(q+q^*)t+\alpha x-t)\frac{1}{z}+O\left(\frac{1}{z^2}\right),\quad z\to\infty.
\label{eq:fplusfstarstar}
\end{equation}
The final step is to recall the expansions \eqref{eq:PhiPrimeExpansion} and \eqref{eq:OneOverSExpansion} to find
\begin{equation}
-\frac{\Phi'(z)}{S(z)} = \frac{2\pi i}{z}+O(z^{-3/2}),\quad z\to\infty,\quad\Im\{z\}>0.
\label{eq:minusphiprimeovers}
\end{equation}
Using  \eqref{eq:fplusfstarstar} with \eqref{eq:minusphiprimeovers} in \eqref{eq:Yexpandlast} then gives
\begin{equation}
Y(z;q,q^*)=\frac{16t}{\alpha^2} +\left[\frac{4}{\alpha^2}(2(q+q^*)t+\alpha x-t)+2\pi i\right]\frac{1}{z}+O\left(\frac{1}{z^{3/2}}\right),\quad z\to\infty,\quad\Im\{z\}>0.
\label{eq:Yexpansion}
\end{equation}

\subsection{The endpoint $q$ as a function of $x$ and $t$}
\label{sec:q}
Now we return to the conditions \eqref{eq:momentscx} supposed to be satisfied by the endpoint $q$ of the arc $\mathcal{B}$, given $(x,t)\in\mathbb{R}^2$.  First we dispense with the special case of $t=0$.
\begin{proposition}
The equations \eqref{eq:momentscx} hold for $t=0$ and $x\le x_\mathrm{c}$ (for any arc $\mathcal{B}$ connecting $q$ to $z_\mathrm{P}$ lying in the quadrant given by the inequalities $\Im\{z\}>0$
and $\Re\{z\}<z_\mathrm{P}$) with $q=\mathfrak{z}(x)$ and $q^*=\mathfrak{z}(x)^*$ being the two branches of the complex part of the turning point curve.
\label{prop:momentscxtzero}
\end{proposition}
\begin{proof}
Set $t=0$ and assume that $q=\mathfrak{z}(x)$.  Whatever the arc $\mathcal{B}$ that is the path of integration in both integrals in \eqref{eq:momentscx} actually is,
it is homotopic (with orientation preserved) to the image of the map $s=\mathfrak{z}(y)$ as $y$ decreases from
$x$ to $-\infty$, and the homotopy avoids the line of phantom poles $\Re\{s\}=z_\mathrm{P}$.  Since $\Phi'(s)$ is analytic for $\Im\{s\}>0$ and since the boundary value $S_+(s)$ is analytic for $s\in \mathcal{B}$,
we may deform the path of integration from $\mathcal{B}$ to the image of $s=\mathfrak{z}(y)$.  In the remainder of this proof, we suppose that this deformation is done, and that $S(z)$ has been redefined with its branch cut as this new contour and its Schwarz reflection in the lower half-plane.  

Therefore, for $p=0$ or $p=1$, we may parametrize the contour integrals appearing in
\eqref{eq:momentscx} by $-\infty<y<x$:
\begin{equation}
I_p(\mathfrak{z}(x),\mathfrak{z}(x)^*)=\Im\left\{\int_{q=\mathfrak{z}(x)}^{z_\mathrm{P}}\frac{\Phi'(s)s^p\,ds}{S_+(s)}\right\} = 
\Im\left\{\int_{-\infty}^x\frac{\Phi'(\mathfrak{z}(y))\mathfrak{z}(y)^p\mathfrak{z}'(y)\,dy}{S_+(\mathfrak{z}(y))}\right\}.
\label{eq:Ipdef}
\end{equation}
(To get the sign correct it is important to take into account that $S_+(s)$ denotes the boundary value taken on the integration contour from the left according to the indicated direction of orientation; since $S$ changes sign across its branch cuts the limits of integration in the $s$-integral can therefore be exchanged without changing the sign of the integral.)
We will prove the proposition by showing that
\begin{equation}
I_0(\mathfrak{z}(x),\mathfrak{z}(x)^*)= -\frac{2\pi x}{\alpha}\quad\text{and}\quad I_1(\mathfrak{z}(x),\mathfrak{z}(x)^*)=-\frac{2\pi x}{\alpha}\Re\{q\}.
\label{eq:ImIpanswers}
\end{equation}

At this point (having used analyticity of $\Phi(s)$ apparent from the formula \eqref{eq:PhidefStirling} 
obtained from Stirling asymptotics to
deform contours) it becomes more convenient to use the alternate formula \eqref{eq:PhidefWKB} for $\Phi(s)$ that is obtained from WKB theory.  The WKB formula \eqref{eq:PhidefWKB} is initially defined for real $z$ for which there
exist two turning points (in the present case this is the interval $z_\mathrm{L}<z<0$) but for the particular
initial data under consideration it has an analytic continuation to $\Im\{z\}>0$ that agrees exactly with
the Stirling formula \eqref{eq:PhidefStirling}.  By differentiation of $\Phi(s)$ as given by \eqref{eq:PhidefWKB}, we have
\begin{equation}
\Phi'(s)= \frac{1}{\alpha}\frac{d\omega_-}{ds}(s)x_-(s) + \frac{1}{\alpha}\int_{-\infty}^{x_-(s)}\left[\frac{\partial\omega}{\partial s}(v;s)-\frac{d\omega_-}{ds}(s)\right]\,dv,\quad z_\mathrm{L}<s<0,
\label{eq:Phiprime}
\end{equation}
where $\omega$ and $\omega_-$ are defined in \eqref{eq:sigmapmdef}--\eqref{eq:omegaminusdef}, 
and where $x_-(s)$ is the real function defined for $z_\mathrm{L}<s<0$ by the portion of the turning
point curve illustrated in red in Figure~\ref{fig:tpc}.  Now, $x_-(s)$ admits analytic continuation into the complex $s$-plane from its interval of definition, and since it has a simple critical point at $s=z_\mathrm{c}$ (a local minimum), there is a curve passing vertically through the point $s=z_\mathrm{c}$
in the complex $s$-plane along which the analytic continuation of $x_-(s)$ is real and decreasing
away from $s=z_\mathrm{c}$.  Clearly, the portion of this curve in the upper half $s$-plane coincides
with the contour $s=\mathfrak{z}(y)$ for $-\infty<y<x_\mathrm{c}$.  Therefore, when $s=\mathfrak{z}(y)$, we have
$x_-(s)=y$, where $x_-(s)$ denotes the analytic continuation of the similarly-named function from the interval $z_\mathrm{L}<s<0$.  Using this information and substituting the analytic continuation of \eqref{eq:Phiprime} along the path $s=\mathfrak{z}(y)$ for $-\infty<y<x<x_\mathrm{c}$ into
\eqref{eq:Ipdef} we obtain
\begin{equation}
I_p(\mathfrak{z}(x),\mathfrak{z}(x)^*)= \Im\left\{
\frac{4}{\alpha}\int_{-\infty}^x\frac{\mathfrak{z}(y)^p\mathfrak{z}'(y)}{S_+(\mathfrak{z}(y))}y\,dy
+\frac{1}{\alpha}\int_{-\infty}^x\frac{\mathfrak{z}(y)^p\mathfrak{z}'(y)}{S_+(\mathfrak{z}(y))}\int_{-\infty}^{y}
\left[\frac{\partial\omega}{\partial s}(v;\mathfrak{z}(y))-4\right]\,dv\,dy
\right\},
\label{eq:Ipxrewrite}
\end{equation}
where we have also used the fact that $d\omega_-(s)/ds = 4$ (independently of the functions $\rho(\cdot)$ and $u(\cdot)$).

We simplify the first term on the right-hand side of \eqref{eq:Ipxrewrite}, writing it in terms of a double integral as follows:
\begin{equation}
\frac{4}{\alpha}\int_{-\infty}^x\frac{\mathfrak{z}(y)^p\mathfrak{z}'(y)}{S_+(\mathfrak{z}(y))}y\,dy = 
\frac{4x}{\alpha}\int_{-\infty}^x\frac{\mathfrak{z}(y)^p\mathfrak{z}'(y)}{S_+(\mathfrak{z}(y))}\,dy -
\frac{4}{\alpha}\int_{-\infty}^x\frac{\mathfrak{z}(y)^p\mathfrak{z}'(y)}{S_+(\mathfrak{z}(y))}\int_y^x\,dv\,dy.
\end{equation}
Exchanging the order of integration in the double integral gives
\begin{equation}
\frac{4}{\alpha}\int_{-\infty}^x\frac{\mathfrak{z}(y)^p\mathfrak{z}'(y)}{S_+(\mathfrak{z}(y))}y\,dy = \frac{4x}{\alpha}\int_{-\infty}^x
\frac{\mathfrak{z}(y)^p\mathfrak{z}'(y)}{S_+(\mathfrak{z}(y))}\,dy -\frac{4}{\alpha}\int_{-\infty}^x\int_{-\infty}^v\frac{\mathfrak{z}(y)^p\mathfrak{z}'(y)}{S_+(\mathfrak{z}(y))}\,dy\,dv.
\end{equation}
Reparametrizing the integrals over $y$ by the map $s=\mathfrak{z}(y)$ this becomes
\begin{equation}
\frac{4}{\alpha}\int_{-\infty}^x\frac{\mathfrak{z}(y)^p\mathfrak{z}'(y)}{S_+(\mathfrak{z}(y))}y\,dy = \frac{4x}{\alpha}\int_{z_\mathrm{P}}^{q=\mathfrak{z}(x)}
\frac{s^p\,ds}{S_+(s)} -\frac{4}{\alpha}\int_{-\infty}^x\int_{z_\mathrm{P}}^{\mathfrak{z}(v)}\frac{s^p\,ds}{S_+(s)}\,dv.
\end{equation}
Now we simplify the second term on the right-hand side of \eqref{eq:Ipxrewrite}, by first exchanging the order of integration and then
reparametrizing the inner integral by $s=\mathfrak{z}(y)$:
\begin{equation}
\begin{split}
\frac{1}{\alpha}\int_{-\infty}^x\frac{\mathfrak{z}(y)^p\mathfrak{z}'(y)}{S_+(\mathfrak{z}(y))}\int_{-\infty}^y
\left[\frac{\partial\omega}{\partial s}(v;\mathfrak{z}(y))-4\right]\,dv\,dy &=\frac{1}{\alpha}\int_{-\infty}^x\int_{v}^x\frac{\mathfrak{z}(y)^p\mathfrak{z}'(y)}{S_+(\mathfrak{z}(y))}\left[
\frac{\partial\omega}{\partial s}(v;\mathfrak{z}(y))-4\right]\,dy\,dv\\
&=\frac{1}{\alpha}\int_{-\infty}^x\int_{\mathfrak{z}(v)}^{q=\mathfrak{z}(x)}\frac{s^p}{S_+(s)}\left[\frac{\partial\omega}{\partial s}(v;s)-4\right]\,ds\,dv.
\end{split}
\end{equation}
Combining these results shows that
\begin{equation}
I_p(\mathfrak{z}(x),\mathfrak{z}(x)^*) = \Im\left\{\frac{4x}{\alpha}\int_{z_\mathrm{P}}^{q=\mathfrak{z}(x)}\frac{s^p\,ds}{S_+(s)}+\frac{1}{\alpha}\int_{-\infty}^x\left[\int_{\mathfrak{z}(v)}^{q=\mathfrak{z}(x)}\frac{\partial\omega}{\partial s}(v;s)\frac{s^p\,ds}{S_+(s)}-4\int_{z_\mathrm{P}}^{q=\mathfrak{z}(x)}\frac{s^p\,ds}{S_+(s)}\right]\,dv
\right\}.
\end{equation}
(The integral over $v$ cannot be broken up without sacrificing convergence.)  Here the paths of integration in the $s$-integrals all lie along the image of the mapping $s=\mathfrak{z}(y)$, $-\infty<y<x_\mathrm{c}$.  Since $S(s^*)=S(s)^*$ and since for $v\in\mathbb{R}$ we have $\omega(v,s^*)=\omega(v,s)^*$ and hence $\partial\omega/\partial s(v,s^*)=\partial\omega/\partial s(v,s)^*$ where
$\omega(v,s)$ is defined for complex $s$ by analytic continuation from the interval $z_\mathrm{L}<s<0$,
we easily obtain
\begin{equation}
I_p(\mathfrak{z}(x),\mathfrak{z}(x)^*)= \frac{2x}{i\alpha}J_p(x)
+\frac{1}{2i\alpha}\int_{-\infty}^x\left[\int_{\mathfrak{z}(v)}^{\mathfrak{z}(x)}\frac{\partial\omega}{\partial s}(v;s)
\frac{s^p\,ds}{S_+(s)} + \int_{\mathfrak{z}(x)^*}^{\mathfrak{z}(v)^*}\frac{\partial\omega}{\partial s}(v;s)\frac{s^p\,ds}{S_+(s)}-4J_p(x)\right]\,dv.
\label{eq:ImIp}
\end{equation}
where
\begin{equation}
J_p(x):=\int_{\mathfrak{z}(x)^*}^{\mathfrak{z}(x)}\frac{s^p\,ds}{S_+(s)}.
\end{equation}
Here the paths of integration all lie along the image of the mapping $s=\mathfrak{z}(y)$, $-\infty<y<x_\mathrm{c}$ and its Schwarz reflection (which in particular contains the branch cut of the re-defined function $S(s)$).

Now we evaluate all of the integrals over $s$ by contour integration arguments (in particular, this will prove that the integrand of the $v$-integral in  \eqref{eq:ImIp} vanishes identically).  Since $S$ changes sign across its branch cut, and since the latter connects $\mathfrak{z}(x)^*$ to $\mathfrak{z}(x)$, we have
\begin{equation}
J_p(x)=-\frac{1}{2}\oint\frac{s^p\,ds}{S(s)}
\end{equation}
where the path of integration is a closed, positively-oriented loop that encloses the branch cut of $S$.  We may now calculate $J_p(x)$ in terms of residues at $s=\infty$.  
Since 
\begin{equation}
\frac{1}{S(s)}=\frac{1}{s}+\Re\{q\}\frac{1}{s^2} + O\left(\frac{1}{s^{3}}\right),\quad s\to\infty,
\label{eq:OneOverSExpansion}
\end{equation}
we see easily that
\begin{equation}
J_0(x)=-i\pi\quad\text{and}\quad J_1(x)=-i\pi\Re\{q\}.
\label{eq:Jvalue}
\end{equation}
Also, by implicit differentiation of the identity $\omega(v;s)^2 = 16\alpha^2s\rho_0(v)+(4s-1+\alpha u_0(v))^2$ we easily obtain
\begin{equation}
\frac{\partial\omega}{\partial s}(v;s)=\frac{8\alpha^2\rho_0(v)+4(4s-1+\alpha u_0(v))}{\omega(v;s)}.
\end{equation}
By definition of the complex part of the turning point curve, the roots of the quadratic $\omega(v;s)^2$ are exactly $s=\mathfrak{z}(v)$ and $s=\mathfrak{z}(v)^*$ for $v\in\mathbb{R}$ with $v<x_\mathrm{c}$; thus $\omega(v;s)^2= 16(s-\mathfrak{z}(v))(s-\mathfrak{z}(v)^*)$.  To determine the function $\omega(v;s)$ as a function of $s$ for fixed $v<x_\mathrm{c}$ (that is, to determine the proper way to take the square root of $\omega(v;s)^2$), we proceed as follows.  By definition, $\omega(v;s)$ is real and has the sign $\sigma_-$ for $z_\mathrm{L}<s<0$, because $v$ lies to the left of both real turning points $x_\pm(s)$ in this case.  The value of $\sigma_-$ is determined from \eqref{eq:sigmapmdef} under the assumption that $z$ lies in the interval in which there exist two real turning points:  $z_\mathrm{L}<z<0$.  But \eqref{eq:sigmapmdef} can be written in the form
\begin{equation}
\sigma_-=\mathrm{sgn}(4(z-z_\mathrm{P}))
\end{equation}
so since $z_\mathrm{P}>0$ we have $\sigma_-=-1$.  Now since $v<x_\mathrm{c}$, the roots of $\omega(v;s)^2$ form a complex-conjugate pair in the complex $s$-plane, and as $\omega(v;s)$ is a negative
real function of $s$ for $z_\mathrm{L}<s<0$ while $z_\mathrm{P}>0$, we may construct the analytic
continuation from the interval $s\in (z_\mathrm{L},0)$ to the complement of a branch cut connecting
the two roots $\mathfrak{z}(v)$ and $\mathfrak{z}(v)^*$ coinciding with a sub-arc of the branch cut for $S(s)$, and normalized so that $\omega(v;s)=4s+O(1)$ as $s\to\infty$.  With the branch cut for $\omega(v;s)$ defined in this way, we can write $\omega(v;s)S_+(s) = 4\tilde{S}_+(v;s)$ where $\tilde{S}(v;s):=\omega(v;s)S(s)/4=s^2+O(s)$
as $s\to\infty$, and $\tilde{S}(s;v)$ is analytic except for \emph{two} cuts:  one connecting $\mathfrak{z}(v)$ to $q=\mathfrak{z}(x)$ and lying along the contour $s=\mathfrak{z}(y)$ for $-\infty<y<x_\mathrm{c}$, and the other being the Schwarz reflection of the first one.  At last we are in a position to evaluate the integral involving $\partial\omega/\partial s(v;s)$ by residues; we first use the fact that $\tilde{S}(v;s)$ changes sign across the contours of integration to write
\begin{equation}
\int_{\mathfrak{z}(v)}^{\mathfrak{z}(x)}\frac{\partial\omega}{\partial s}(v;s)\frac{s^p\,ds}{S_+(s)} +
\int_{\mathfrak{z}(x)^*}^{\mathfrak{z}(v)^*}\frac{\partial\omega}{\partial s}(v;s)\frac{s^p\,ds}{S_+(s)} = 
-\frac{1}{2}\oint\frac{8\alpha^2\rho_0(v)+4(4s-1+\alpha u_0(v))}{4\tilde{S}(v;s)}s^p\,ds
\end{equation}
where the path of integration is a closed, positively-oriented loop that encircles both of the branch cuts of $\tilde{S}(v;s)$.  Now since 
\begin{equation}
\frac{1}{\tilde{S}(v;s)}=\frac{1}{s^2}+\frac{\Re\{q\}+\Re\{\mathfrak{z}(v)\}}{s^3} +O\left(\frac{1}{s^4}\right),\quad s\to\infty,
\end{equation}
a residue calculation shows that, for $p=0$,
\begin{equation}
\int_{\mathfrak{z}(v)}^{\mathfrak{z}(x)}\frac{\partial\omega}{\partial s}(v;s)\frac{ds}{S_+(s)} +
\int_{\mathfrak{z}(x)^*}^{\mathfrak{z}(v)^*}\frac{\partial\omega}{\partial s}(v;s)\frac{ds}{S_+(s)} = -4\pi i
\label{eq:intIIp0}
\end{equation}
and, for $p=1$,
\begin{equation}
\begin{split}
\int_{\mathfrak{z}(v)}^{\mathfrak{z}(x)}\frac{\partial\omega}{\partial s}(v;s)\frac{s\,ds}{S_+(s)} +
\int_{\mathfrak{z}(x)^*}^{\mathfrak{z}(v)^*}\frac{\partial\omega}{\partial s}(v;s)\frac{s\,ds}{S_+(s)} &= -4\pi i\left[\Re\{q\}+\Re\{\mathfrak{z}(v)\}+\frac{1}{2}\alpha^2\rho_0(v)+\frac{\alpha u_0(v)-1}{4}\right] \\ &= -4\pi i\Re\{q\},
\end{split}
\label{eq:intIIp1}
\end{equation}
where on the second line we have used \eqref{eq:Rezeta}.  Using \eqref{eq:Jvalue}, \eqref{eq:intIIp0},
and \eqref{eq:intIIp1} in \eqref{eq:ImIp} we arrive at \eqref{eq:ImIpanswers},
so the proof is complete.
\end{proof}

Now we show that the solution of the equations \eqref{eq:momentscx} can be continued to nearby $(x,t)$ under certain conditions involving the function $Y(z;q,q^*)$.
\begin{proposition}
Suppose that $q$ is a solution of the equations \eqref{eq:momentscx} at $x=x_0$ and $t=t_0$ with $\Im\{q\}>0$
and for which $Y(q;q,q^*)\neq 0$.  Then the equations \eqref{eq:momentscx} have a unique smooth
solution $q(x,t)$ valid for $(x,t)$ near $(x_0,t_0)$ for which $q(x_0,t_0)=q$.
\label{prop:cxcontinue}
\end{proposition}
\begin{proof}
We construct the Jacobian determinant of $M_0$ and $M_1$ with respect to $q$ and $q^*$ (here viewed as independent complex variables).  We first rewrite $M_0$ and $M_1$ in the form
\begin{equation}
\begin{split}
M_0(q,q^*,x,t) &= \frac{2\pi}{\alpha^2}(2(q+q^*)t+\alpha x - t) -\frac{i}{4}\left[\oint_L\frac{\Phi'(s)\,ds}{S(s)}-\oint_{L^*}\frac{\Phi'(s)\,ds}{S(s)}\right]\\
M_1(q,q^*,x,t) &=\frac{\pi}{\alpha^2}((3q^2 + 2qq^* + 3q^{*2})t+(q+q^*)(\alpha x-t))-\frac{i}{4}
\left[\oint_L\frac{\Phi'(s)s\,ds}{S(s)}-\oint_{L^*}\frac{\Phi'(s)s\,ds}{S(s)}\right],
\end{split}
\end{equation}
where $L$ is a negatively-oriented contour beginning and ending at $z_\mathrm{P}$ and encircling the contour arc $B$ exactly once, and $L^*$ denotes the Schwarz reflection of this contour in the lower half-plane (with orientation induced from $L$ by reflection, that is, $L^*$ is positively-oriented). 
Now holding $L$ fixed, the only dependence in the integrals on $q$ and $q^*$ enters parametrically through the
function $S(s)$.  Since by direct calculation,
\begin{equation}
\frac{\partial}{\partial q}\frac{1}{S(s)} = \frac{1}{2S(s)(s-q)}\quad\text{and}\quad
\frac{\partial}{\partial q^*}\frac{1}{S(s)} = \frac{1}{2S(s)(s-q^*)}
\end{equation}
we find that
\begin{equation}
\begin{split}
\frac{\partial M_0}{\partial q} &= \frac{4\pi t}{\alpha^2} -\frac{i}{8}\left[\oint_L\frac{\Phi'(s)\,ds}{S(s)(s-q)}
-\oint_{L^*}\frac{\Phi'(s)\,ds}{S(s)(s-q)}\right]\\
\frac{\partial M_0}{\partial q^*} &= \frac{4\pi t}{\alpha^2} -\frac{i}{8}\left[\oint_L\frac{\Phi'(s)\,ds}{S(s)(s-q^*)}
-\oint_{L^*}\frac{\Phi'(s)\,ds}{S(s)(s-q^*)}\right].
\end{split}
\end{equation}
Comparing with the original formula \eqref{eq:Yfdefine} for $Y(z;q,q^*)$, we see that
\begin{equation}
\frac{\partial M_0}{\partial q} = \frac{\pi}{4}Y(q;q,q^*)\quad\text{and}\quad\frac{\partial M_0}{\partial q^*}=
\frac{\pi}{4}Y(q^*;q,q^*).
\label{eq:M0cxpartials}
\end{equation}
Similarly, for partial derivatives of $M_1$ we obtain
\begin{equation}
\begin{split}
\frac{\partial M_1}{\partial q}&=\frac{\pi}{\alpha^2}\left[6tq + 2tq^*+\alpha x-t\right]-\frac{i}{8}
\left[\oint_L\frac{\Phi'(s)s\,ds}{S(s)(s-q)}-\oint_{L^*}\frac{\Phi'(s)s\,ds}{S(s)(s-q)}\right]\\
\frac{\partial M_1}{\partial q^*}&=\frac{\pi}{\alpha^2}\left[6tq^* + 2tq+\alpha x-t\right]-\frac{i}{8}
\left[\oint_L\frac{\Phi'(s)s\,ds}{S(s)(s-q^*)}-\oint_{L^*}\frac{\Phi'(s)s\,ds}{S(s)(s-q^*)}\right].
\end{split}
\end{equation}
The partial derivatives of $M_1$ can be rewritten as
\begin{equation}
\begin{split}
\frac{\partial M_1}{\partial q}&=\frac{\pi}{\alpha^2}\left[6tq + 2tq^*+\alpha x-t\right]-
\frac{iq}{8}
\left[\oint_L\frac{\Phi'(s)\,ds}{S(s)(s-q)}-\oint_{L^*}\frac{\Phi'(s)\,ds}{S(s)(s-q)}\right]\\
&\quad\quad\quad{}
-\frac{i}{8}\left[\oint_L\frac{\Phi'(s)\,ds}{S(s)}-\oint_{L^*}\frac{\Phi'(s)\,ds}{S(s)}\right]\\
\frac{\partial M_1}{\partial q^*}&=\frac{\pi}{\alpha^2}\left[6tq^* + 2tq+\alpha x-t\right]-
\frac{iq^*}{8}
\left[\oint_L\frac{\Phi'(s)\,ds}{S(s)(s-q^*)}-\oint_{L^*}\frac{\Phi'(s)\,ds}{S(s)(s-q^*)}\right]\\
&\quad\quad\quad{}
-\frac{i}{8}\left[\oint_L\frac{\Phi'(s)\,ds}{S(s)}-\oint_{L^*}\frac{\Phi'(s)\,ds}{S(s)}\right].
\end{split}
\end{equation}
Setting $x=x_0$ and $t=t_0$, we use the fact that $q$ is a solution of $M_0=0$ to eliminate the terms
on the second line in each case:
\begin{equation}
\begin{split}
\left.\frac{\partial M_1}{\partial q}\right|_{(x_0,t_0)} &= \frac{4\pi t_0}{\alpha^2}q  -
\frac{iq}{8}
\left[\oint_L\frac{\Phi'(s)\,ds}{S(s)(s-q)}-\oint_{L^*}\frac{\Phi'(s)\,ds}{S(s)(s-q)}\right]\\
&= q\left.\frac{\partial M_0}{\partial q}\right|_{(x_0,t_0)}\\
\left.\frac{\partial M_1}{\partial q^*}\right|_{(x_0,t_0)} &= \frac{4\pi t_0}{\alpha^2}q ^* -
\frac{iq^*}{8}
\left[\oint_L\frac{\Phi'(s)\,ds}{S(s)(s-q^*)}-\oint_{L^*}\frac{\Phi'(s)\,ds}{S(s)(s-q^*)}\right]\\
&=q^*\left.\frac{\partial M_0}{\partial q^*}\right|_{(x_0,t_0)}.
\end{split}
\label{eq:M1cxpartials}
\end{equation}
Therefore, the Jacobian determinant is
\begin{equation}
\begin{split}
\mathscr{J}(q,q^*,x_0,t_0):=
\left.\left[\frac{\partial M_1}{\partial q}\frac{\partial M_0}{\partial q^*}-\frac{\partial M_0}{\partial q}
\frac{\partial M_1}{\partial q^*}\right]\right|_{(x_0,t_0)} &= \frac{\pi^2}{16}Y(q;q,q^*)Y(q^*;q,q^*)(q-q^*) \\ &=
\frac{i\pi^2}{8}|Y(q;q,q^*)|^2\Im\{q\}.
\end{split}
\end{equation}
Since this is nonzero by assumption, the solution to the simultaneous equations $M_0=M_1=0$
can be continued from $q$ to $q(x,t)$ by the Implicit Function Theorem.
\end{proof}

\begin{corollary}
Let $x<x_\mathrm{c}$.  Then, for $|t|$ sufficiently small, there exists a unique solution $q=q(x,t)$
of the equations \eqref{eq:momentscx} satisfying $q(x,0)=\mathfrak{z}(x)$, where $\mathfrak{z}(x)$ is the complex part of the
turning point curve.  For $t>0$ the solution can be continued (even for $x\ge x_\mathrm{c}$) until $\Im\{q(x,t)\}$ becomes zero.
\end{corollary}
\begin{proof}
From Proposition~\ref{prop:momentscxtzero} we have that the equations \eqref{eq:momentscx}
are satisfied by $q=\mathfrak{z}(x)$ when $t=0$ and $x<x_\mathrm{c}$, and we note that $\Im\{\mathfrak{z}(x)\}>0$
for $x<x_\mathrm{c}$.  By Proposition~\ref{prop:Yzeroscx} the only possible zeros of $Y(z;q,q^*)$ in 
the open upper half-plane are far from $q$ when $|t|$ is small, which implies that $Y(\mathfrak{z}(x);\mathfrak{z}(x),\mathfrak{z}(x)^*)\neq 0$.  Moreover, if $t>0$ then $Y(q;q,q^*)$ can only become
zero if $\Im\{q\}=0$.
Therefore, by Proposition~\ref{prop:cxcontinue} the Corollary is proved.
\end{proof}

It follows that for $t\ge 0$, the only obstruction to smooth continuation of the endpoint function $q(x,t)$ from its values $q(x,0)=\mathfrak{z}(x)$ for $x<x_\mathrm{c}$ may be a collision of $q$ with the real axis.  We now completely resolve the scope of the continuation.
\begin{proposition}
There exists 
a well-defined smooth curve $x=x_\mathrm{c}(t)$ for 
$t\ge 0$ satisfying $x_\mathrm{c}(0)=x_\mathrm{c}$,
such that the solution of \eqref{eq:momentscx} with $q(x,0)=\mathfrak{z}(x)$ for $x<x_\mathrm{c}$ can be uniquely continued to the domain 
$t>0$ and $x<x_\mathrm{c}(t)$.  Also, $\Im\{q(x,t)\}>0$ whenever $x<x_\mathrm{c}(t)$, but 
$\Im\{q(x,t)\}\downarrow 0$ as $x\uparrow x_\mathrm{c}(t)$.
\label{prop:xc-continue}
\end{proposition}
\begin{proof}
First, we argue that it is impossible for $\Im\{q\}$ to vanish for bounded $(x,t)$ unless $q$ tends to a real point $z$ in the interval $(z_\mathrm{L},0)$.  Consider the combination of $M_0$ and $M_1$ given by:
\begin{equation}
\begin{split}
F(q,q^*;x,t):=&M_1(q,q^*;x,t)-\frac{1}{2}(q+q^*)M_0(q,q^*;x,t)\\
{}=&\frac{\pi t}{\alpha^2}(q-q^*)^2 + I_1(q,q^*)-\frac{1}{2}(q+q^*)I_0(q,q^*).
\end{split}
\end{equation}
Clearly, we must have $F(q,q^*;x,t)=0$ for any solution of the two equations $M_0=M_1=0$.  Now for $t$ bounded the first term on the second line above obviously converges to zero as $\Im\{q\}\downarrow 0$.  By definition of $I_p(q,q^*)$ we have 
\begin{equation}
I_1(q,q^*)-\frac{1}{2}(q+q^*)I_0(q,q^*)=\Im\left\{\int_\mathcal{B} \Phi'(s)\frac{s-\frac{1}{2}(q+q^*)}{S_+(s)}\,ds\right\}.
\end{equation}
A careful dominated convergence argument (in which the integration from $q$ is replaced locally by a loop around $s=q$ that is deformed to the real axis in a neighborhood of $z\in\mathbb{R}$, a contour on which $(s-\frac{1}{2}(q+q^*))/S(s)$ converges pointwise to $1$ and is uniformly bounded) shows that
\begin{equation}
\mathop{\lim_{q\to z\in\mathbb{R}}}_{\Im\{q\}>0} \left[I_1(q,q^*)-\frac{1}{2}(q+q^*)I_0(q,q^*)\right]=\Im\left\{
\int_z^{z_\mathrm{P}}\Phi_+'(s)\,ds\right\} = -\Im\{\Phi_+(z)\}
\end{equation}
because $\Im\{\Phi_+(z_\mathrm{P})\}=0$ according to \eqref{eq:ImPhiPlus}.  But again according to \eqref{eq:ImPhiPlus} we see that $\Im\{\Phi_+(z)\}=0$ only for $z=z_\mathrm{P}$ and for $z\in [z_\mathrm{L},0]$.  Therefore, for bounded $t$ we can only have $\Im\{q\}\downarrow 0$ if $q\to z\in [z_\mathrm{L},0]$ or if $q\to z_\mathrm{P}$.  But a similar argument also shows that $M_0(q,q^*;x,t)\to\infty$ if $q\to z_\mathrm{L}$, $q\to 0$, or if $q\to z_\mathrm{P}$ with $\Im\{q\}>0$ and $(x,t)$ bounded, so $q$ can only tend to a point $z\in(z_\mathrm{L},0)$ for bounded $(x,t)$.

Now fix $z\in(z_\mathrm{L},0)$.  Then $F(q,q^*;x,t)\to 0$ as $q\to z$ with $\Im\{q\}>0$, so the two conditions $M_0=M_1=0$ degenerate to a single condition:  $M_0^\mathrm{c}(z;x,t)=0$,
where
\begin{equation}
M_0^\mathrm{c}(z;x,t):=\frac{2\pi}{\alpha^2}(4zt+\alpha x-t) + I_0^\mathrm{c}(z),\quad z_\mathrm{L}<z<0,
\end{equation}
where
\begin{equation}
\begin{split}
I_0^\mathrm{c}(z):=&-\frac{\pi}{2}\Phi'(z)+\Im\left\{
\int_0^{z_\mathrm{P}}\frac{\Phi_+'(s)\,ds}{s-z}\right\}\\
{}=&\frac{\pi}{2}\Phi'(z)+\Im\left\{\int_{z_\mathrm{L}}^{z_\mathrm{P}}\frac{\Phi'(s)\,ds}{s-z}\right\},
\end{split}
\end{equation}
where on the second line the path of integration from $z_\mathrm{L}$ to $z_\mathrm{P}$ lies in the upper half $s$-plane.
Since $\Phi'(z)\to +\infty$ as $z\downarrow z_\mathrm{L}$ while $\Phi'(z)\to -\infty$ as $z\uparrow 0$,
the expression on the first line shows that $I_0^\mathrm{c}(z)\to -\infty$ as $z\downarrow z_\mathrm{L}$, while that on the second line shows that $I_0^\mathrm{c}(z)\to -\infty$ as $z\uparrow 0$.  Both expressions are real-valued, finite, and equivalent for $z\in (z_\mathrm{L},0)$.

For each $z\in (z_\mathrm{L},0)$, the condition $M_0^\mathrm{c}(z;x,t)=0$ defines a straight line in the $(x,t)$-plane.  We will now show that the union of these lines is an unbounded region of the form $x\ge x_\mathrm{c}(t)$ for all 
$t\ge 0$, where $x_\mathrm{c}(t)$ is a smooth function of $t$ with $x_\mathrm{c}(0)=x_\mathrm{c}$.  More precisely, through each point $(x,t)$ with $x>x_\mathrm{c}(t)$ and 
$t\ge 0$ there will pass exactly two of these lines corresponding to distinct values of $z\in (z_\mathrm{L},0)$, and the curve $x=x_\mathrm{c}(t)$ will appear as the caustic formed by the intersection of infinitesimally neighboring lines.  Given 
$t\ge 0$ fixed, the possible $x$-values lying on the lines are given by
\begin{equation}
x=\frac{1-4z}{\alpha}t-\frac{\alpha}{2\pi}I_0^\mathrm{c}(z),\quad z_\mathrm{L}<z<0.
\label{eq:xline}
\end{equation}
The first term is obviously bounded for $t$ fixed and $z\in (z_\mathrm{L},0)$, but the second term tends to $+\infty$ as $z\downarrow z_\mathrm{L}$ or $z\uparrow 0$.  This shows that for 
$t\ge 0$ fixed the region occupied by the lines is exactly
\begin{equation}
x\ge x_\mathrm{c}(t):=\min_{z_\mathrm{L}<z<0}\left[\frac{1-4z}{\alpha}t-\frac{\alpha}{2\pi}I_0^\mathrm{c}(z)\right]<+\infty.
\end{equation}
To determine $x_\mathrm{c}(t)$ we show that for 
each $t\ge 0$, the function to be minimized has exactly one critical point in $(z_\mathrm{L},0)$, which is then (by asymptotics as $z\downarrow z_\mathrm{L}$ and as $z\uparrow 0$) necessarily the unique minimizer.  Indeed,
\begin{equation}
\frac{\partial}{\partial z}\left[\frac{1-4z}{\alpha}t-\frac{\alpha}{2\pi}I_0^\mathrm{c}(z)\right]=-\frac{\alpha}{2\pi}\frac{\partial M_0^\mathrm{c}}{\partial z}(z;x,t) = -\frac{\alpha}{4}Y^\mathrm{c}(z;z),
\end{equation}
where $Y^\mathrm{c}(w;z)$ is by definition for $z_\mathrm{L}<w<0$ the limiting value of the function $Y(w;q,q^*)$ as $q\to z\in (z_\mathrm{L},0)$ with $\Im\{q\}>0$:
\begin{equation}
Y^\mathrm{c}(w;z):=\frac{16t}{\alpha^2}-\frac{\Phi'(w)-\Phi'(z)}{w-z}+\frac{2}{\pi}\Im\left\{
\int_0^{z_\mathrm{P}}\frac{\Phi_+'(s)\,ds}{(s-z)(s-w)}\right\},\quad z_\mathrm{L}<w<0,\quad z_\mathrm{L}<z<0.
\end{equation}
It is obvious that $Y^\mathrm{c}(w;z)$ is symmetric in its arguments:  $Y^\mathrm{c}(w;z)=Y^\mathrm{c}(z;w)$, and for each $z\in (z_\mathrm{L},0)$ there is, by Proposition~\ref{prop:Yzeroscx}, a unique simple zero $w=w(z)\in (z_\mathrm{L},0)$ of $Y^\mathrm{c}(w;z)$. 
The graph $w=w(z)$ must be symmetric with respect to reflection through the diagonal $w=z$ because otherwise there would necessarily be multiple roots $w\in (z_\mathrm{L},0)$ of $Y^\mathrm{c}(w;z)$ for at least some $z$ in the interval $(z_\mathrm{L},0)$.  Therefore, this graph either coincides with the diagonal line $w=z$ or it connects $(z_\mathrm{L},0)$ with $(0,z_\mathrm{L})$ and crosses the diagonal exactly once orthogonally.  But the case $w(z)=z$ can be ruled out easily by symmetry of $Y^\mathrm{c}(w;z)$ and the fact that $Y^\mathrm{c}(w;z)\uparrow +\infty$ as $w\downarrow z_\mathrm{L}$ while $Y^\mathrm{c}(w;z)\downarrow -\infty$ as $w\uparrow 0$.  Therefore, $Y^\mathrm{c}(z;z)$ has exactly one simple root in the interval $(z_\mathrm{L},0)$ for all $t>-T$, and this establishes the uniqueness of the critical point and hence of the minimizer.  Let us denote this critical point by $z=z_\mathrm{c}(t)$, and the corresponding value of $x$ defined for $z=z_c(t)$ by \eqref{eq:xline} as $x=x_\mathrm{c}(t)$.  By consideration of a different type of configuration for the $g$-function in which the branch points of the square root $S(z)$ are both real, and examination of the limit in which these two points coalesce (see \cite{DiFrancoM12b}), it can be shown that 
in fact $z_\mathrm{c}(0)=z_\mathrm{c}\in (z_\mathrm{L},0)$, where $z_\mathrm{c}$ is the specific value 
at which the complex branches $\mathfrak{z}(x)$ and $\mathfrak{z}(x)^*$ of the turning point curve become become real (see Figure~\ref{fig:tpc}).
It then follows from Proposition~\ref{prop:momentscxtzero} that $x_\mathrm{c}(0)=x_\mathrm{c}$.    
\end{proof}
The curve $x=x_\mathrm{c}(t)$ is plotted for the case of $\alpha=\delta=1$ and $\mu=2$ in Figure~\ref{fig:xc-of-t}.
\begin{figure}[h]
\begin{center}
\includegraphics[width=0.6\linewidth]{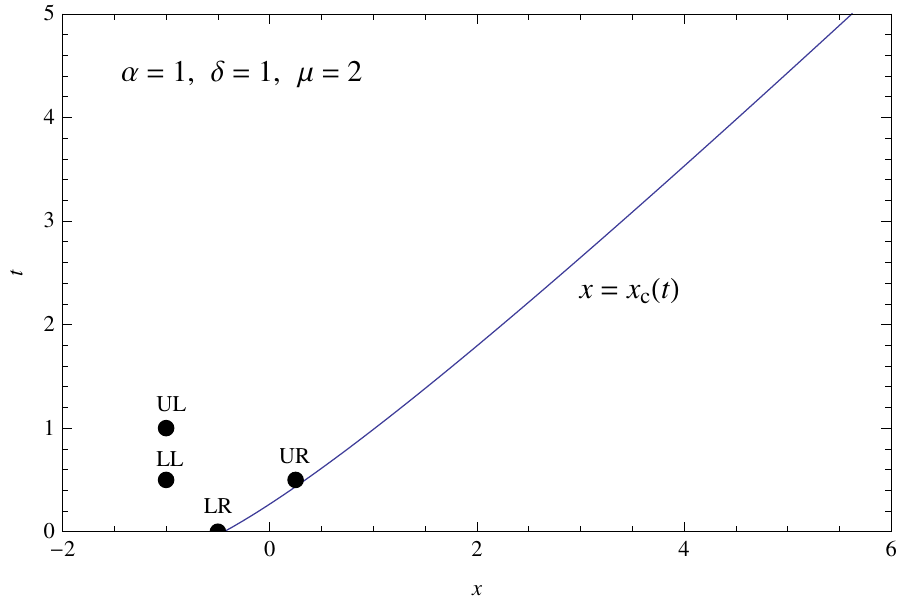}
\end{center}
\caption{The curve $x_\mathrm{c}(t)$ for $\alpha=\delta=1$ and $\mu=2$.  The four marked points in the domain $x<x_\mathrm{c}(t)$ correspond to the $(x,t)$-values in the four indicated panels of Figure~\ref{fig:Im-h-plots} below.}
\label{fig:xc-of-t}
\end{figure}
In Figure~\ref{fig:xc-compare1} we plot the differences between $x_\mathrm{c}(t)$ for various parameter values and the specific function $x_\mathrm{c}(t)$ shown in Figure~\ref{fig:xc-of-t}.
\begin{figure}[h]
\begin{center}
\includegraphics[width=0.475\linewidth]{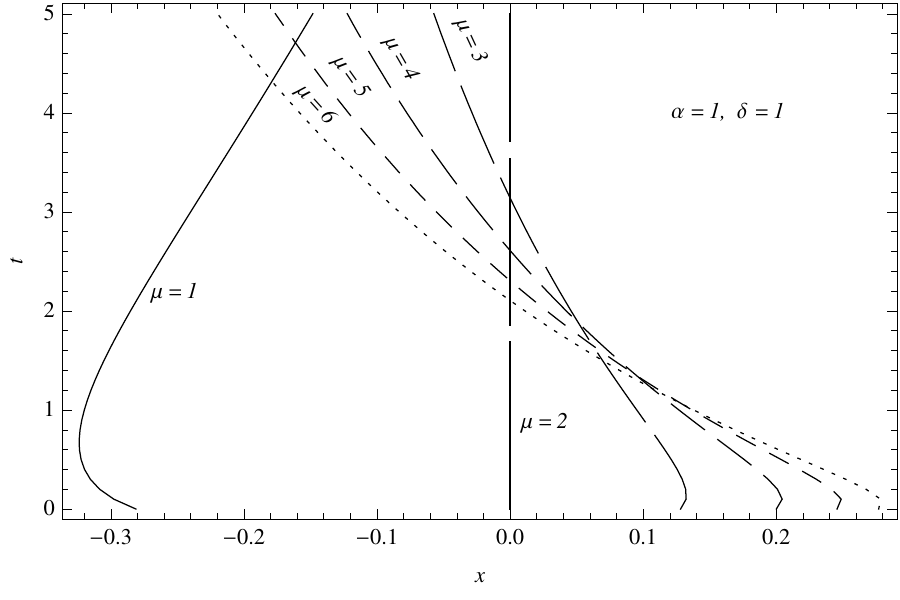}%
\includegraphics[width=0.475\linewidth]{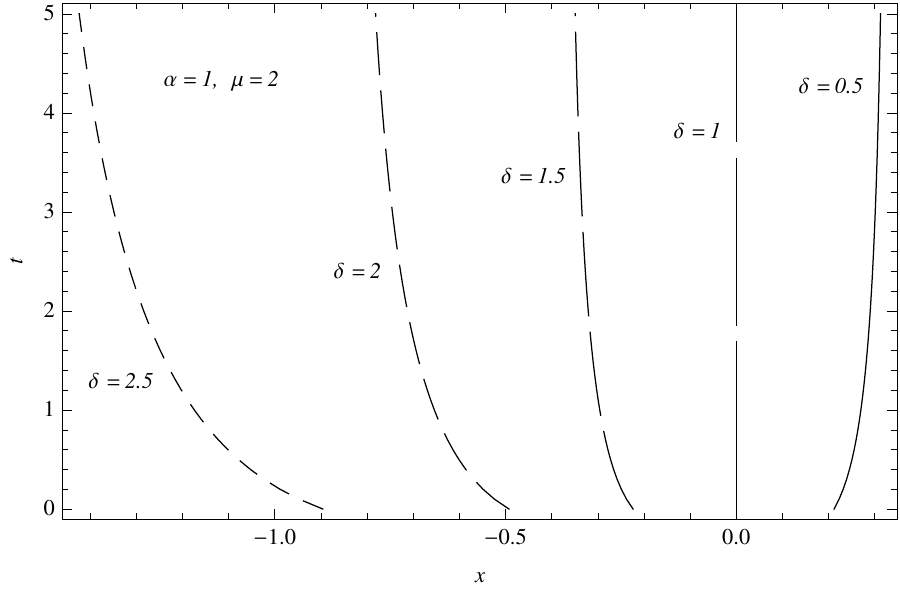}
\end{center}
\caption{Left:  for $\alpha=\delta=1$, curves $x=x_\mathrm{c}(t)$ for various values of $\mu$ compared with the curve in Figure~\ref{fig:xc-of-t}.  Right:  for $\alpha=1$ and $\mu=2$, curves $x=x_\mathrm{c}(t)$ for various values of $\delta$ compared with the curve in Figure~\ref{fig:xc-of-t}.}
\label{fig:xc-compare1}
\end{figure}
In particular, these plots clearly show that $x_\mathrm{c}(t)$ is not (in general at least) a linear function of $t$.

Finally, we show that the complex-valued function $q(x,t)$ defined for all $t\ge 0$ and $x\le x_\mathrm{c}(t)$ satisfies certain canonical partial differential equations (Whitham equations)
in the interior of the domain of definition.
\begin{proposition}
Suppose the equations \eqref{eq:momentscx} admit a solution $q=q(x,t)$ differentiable with respect to $x$ and $t$.  Then
\begin{equation}
\frac{\partial q}{\partial t} +\frac{1}{\alpha}\left[1-3q-q^*\right]\frac{\partial q}{\partial x}=0\quad\text{and}\quad
\frac{\partial q^*}{\partial t} +\frac{1}{\alpha}\left[1-3q^*-q\right]\frac{\partial q^*}{\partial x}=0.
\label{eq:cxWhitham}
\end{equation}
\label{prop:cxWhitham}
\end{proposition}
\begin{proof}
This follows from the relations $M_0(q(x,t),q^*(x,t),x,t)=M_1(q(x,t),q^*(x,t),x,t)=0$ by implicit differentiation.  Indeed, since
\begin{equation}
\frac{\partial M_0}{\partial x}=\frac{2\pi}{\alpha}\quad\text{and}\quad
\frac{\partial M_0}{\partial t}=\frac{2\pi}{\alpha^2}[2(q+q^*)-1]
\end{equation}
while
\begin{equation}
\frac{\partial M_1}{\partial x}=\frac{\pi}{\alpha}(q+q^*)\quad\text{and}\quad
\frac{\partial M_1}{\partial t}=\frac{\pi}{\alpha^2}[3q^2+2qq^*+3q^{*2}-q-q^*],
\end{equation}
it follows immediately from \eqref{eq:M0cxpartials} and \eqref{eq:M1cxpartials} given in the proof of Proposition~\ref{prop:cxcontinue} that
\begin{equation}
\frac{\partial q}{\partial x}=-\frac{4}{\alpha Y(q;q,q^*)}\quad\text{and}\quad
\frac{\partial q^*}{\partial x}=-\frac{4}{\alpha Y(q^*;q,q^*)},
\label{eq:qqstarx}
\end{equation}
and also that
\begin{equation}
\frac{\partial q}{\partial t}=\frac{4}{\alpha^2Y(q;q,q^*)}\left[1-3q-q^*\right]\quad\text{and}\quad
\frac{\partial q^*}{\partial t}=\frac{4}{\alpha^2Y(q^*;q,q^*)}\left[1-3q^*-q\right].
\label{eq:qqstart}
\end{equation}
It is then obvious that \eqref{eq:cxWhitham} holds true if $Y(q;q,q^*)$ and $Y(q^*;q,q^*)$ are nonzero, a condition for the existence of a smooth solution $q(x,t)$.
\end{proof}
The equations \eqref{eq:cxWhitham} are actually equivalent to the dispersionless MNLS system
\eqref{eq:DispersionlessMNLS} in the locally subsonic case that $Q<0$, as can be seen by means of a straightforward calculation using the substitutions 
\begin{equation}
\rho = \frac{4}{\alpha^2}(\Im\{q^{1/2}\})^2\quad\text{and}\quad u = \frac{1}{\alpha}(1-4|q|).
\end{equation}
This shows that the variables $q$ and $q^*$ are complex \emph{Riemann invariants} for the elliptic dispersionless MNLS system.

\subsection{Choice of the arc $\mathcal{B}$}
Supposing that $t\ge 0$ and $x< x_\mathrm{c}(t)$, we determine the endpoint $q=q(x,t)$ with
$\Im\{q(x,t)\}>0$ as explained in \S\ref{sec:q} and hence obtain as explained in \S\ref{sec:BasicFormulae} the functions $g(z)$ and $h(z)$ (parametrized by $(x,t)$ of course).  
These functions are actually only well-defined once we specify a particular contour arc $\mathcal{B}$
in the upper half-plane connecting $z=q$ with $z=z_\mathrm{P}$.  We now describe how $\mathcal{B}$
is to be chosen.

\begin{proposition}
Let $t\ge 0$ and $x<x_\mathrm{c}(t)$, and let $q=q(x,t)$ be determined as explained in \S\ref{sec:q}.  Then there is a unique choice of the arc $\mathcal{B}$ connecting $z=q$ with $z=z_\mathrm{P}$ for which $\Im\{h(z)\}$ can be extended by continuity to the domain $z\in\mathbb{C}\setminus (-\infty,z_\mathrm{L}]\cup [0,+\infty)$.
\end{proposition}
\begin{proof}
The key observation is that as the imaginary part of an analytic function, $\Im\{h(z)\}$ is a harmonic function in the domain $\mathbb{C}\setminus ((-\infty,z_\mathrm{L}]\cup[0,+\infty)\cup \mathcal{B}\cup \mathcal{B}^*)$.
Generally there is a jump discontinuity across $\mathcal{B}\cup \mathcal{B}^*$ that we wish to avoid by choice of $\mathcal{B}$.  

Firstly, regardless of how $\mathcal{B}$ is chosen, we can easily see that $\Im\{h(z)\}=0$ for $z\in [z_\mathrm{L},0]$, for $z=z_\mathrm{P}$, for $z=q$, and for $z=q^*$.  Indeed, by Schwarz symmetry of $h$ and the fact that $(z_\mathrm{L},0)$ is in the domain of analyticity of $h$, it follows immediately that $\Im\{h(z)\}=0$ for $z_\mathrm{L}<z<0$.  But since $S(z)$ is Schwarz symmetric and analytic on the real axis except at its jump discontinuity point $z=z_\mathrm{P}$, and since
according to \eqref{eq:ImYPlus}, the boundary values $\Im\{Y_\pm(z;q,q^*)\}$ taken by $\Im\{Y(z;q,q^*)\}$ on the real axis are bounded, it follows from \eqref{eq:2hprime} that the corresponding boundary values of $\Im\{h(z)\}$ are real differentiable functions of real $z$ except at the point $z=z_\mathrm{P}$ (which is nonetheless a point of Lipschitz continuity of the boundary values of $\Im\{h(z)\}$ for $z\in\mathbb{R}$).  In particular this implies that the limit points $z=z_\mathrm{L}$ and $z=0$ of the open interval $(z_\mathrm{L},0)$ are also points where $\Im\{h(z)\}=0$.  The fact that the well-defined value $\Im\{h(z_\mathrm{P})\}$ vanishes is a consequence of
the Fundamental Theorem of Calculus, the formula \eqref{eq:2hprime}, and the formula \eqref{eq:ImYPlus}; since $S(z)=-\sqrt{(z-\Re\{q\})^2+\Im\{q\}^2}$ holds
for $z<z_\mathrm{P}$, we integrate along the top edge of the branch cut for $Y(z;q,q^*)$ to
obtain
\begin{equation}
\begin{split}
\Im\{h(z_\mathrm{P})\}&=\Im\left\{h(0) -\frac{1}{2}\int_0^{z_\mathrm{P}}S(s)Y_+(s;q,q^*)\,ds\right\}\\
&=\frac{1}{2}\int_0^{z_\mathrm{P}}\sqrt{(s-\Re\{q\})^2+\Im\{q\}^2}\Im\{Y_+(s;q,q^*)\}\,ds\\
&=- 2\pi\int_0^{z^+}\frac{ds}{\sqrt{\mu^2-16s}} +\frac{\pi}{\alpha}\int_{z^+}^{z_\mathrm{P}}\,ds.
\end{split}
\end{equation}
By direct evaluation of these integrals and the use of the definitions of $z^+$ and $z_\mathrm{P}$ we obtain the claimed result that $\Im\{h(z_\mathrm{P})\}=0$.  By Schwarz symmetry of $Y(z;q,q^*)$ the same result holds had we integrated instead along the lower edge of the branch cut for $Y(z;q,q^*)$.  With this information we can express $\Im\{h(q)\}$ as an integral again using the Fundamental Theorem of Calculus.  Indeed, integrating along the left edge of the branch cut $\mathcal{B}$ in the direction of its orientation from $z=q$ to $z=z_\mathrm{P}$, we get (using $\Im\{h(z_\mathrm{P})\}=0$)
\begin{equation}
\Im\{h(q)\}=\frac{1}{2}\Im\left\{\int_\mathcal{B}S_+(s)Y(s;q,q^*)\,ds\right\}.
\end{equation}
Now using the formula \eqref{eq:gplusminusprime} and the fact that $g(z)$ is Schwarz-symmetric and analytic for $z\in \mathbb{C}\setminus (\mathcal{B}\cup \mathcal{B}^*)$ we have
\begin{equation}
\Im\{h(q)\}=\frac{1}{2}\Im\left\{\int_\mathcal{B} (g_+'(s)-g_-'(s))\,ds\right\}=\frac{i}{4}\oint g'(s)\,ds = \frac{i}{4}\oint p(s)\,ds,
\end{equation}
where the closed contour of integration is a large positively-oriented circle of arbitrarily large radius.  It therefore follows from the fact that $q$ satisfies the conditions \eqref{eq:momentscx}
that $p(s)=O(s^{-2})$ as $s\to\infty$ and hence $\Im\{h(q)\}=0$ as claimed.  The fact that $\Im\{h(q^*)\}=0$ then follows by Schwarz symmetry of $h$.

Still regarding the arc $\mathcal{B}$ connecting $z=q$ with $z=z_\mathrm{P}$ as arbitrary, we note that
the zero level set of $\Im\{h(z)\}$ defined by $\mathcal{Z}:=\{z\in\mathbb{C}\setminus ((-\infty,z_\mathrm{L}]\cup [0,+\infty)\cup \mathcal{B}\cup \mathcal{B}^*), \Im\{h(z)\}=0\}$ actually extends by continuity to the excluded branch cut $\mathcal{B}\cup \mathcal{B}^*$ and moreover is independent of $\mathcal{B}$.  This follows from the fact that given $\Im\{h(q)\}=0$ we may express $\Im\{h(z)\}$ in terms of an integral from $s=q$ to $s=z$, and because $s=q$ is the square-root branch point of $S(s)$, a change of choice of branch for $S$ simply amounts to a change of sign of the integral and hence of $\Im\{h(z)\}$.  This clearly leaves $\mathcal{Z}$ invariant.

We will now characterize $\mathcal{Z}$ completely, and prove in particular that $\mathcal{Z}$ contains a smooth arc connecting $z=q$ with $z=z_\mathrm{P}$.  When the proof is finished we will select this arc to be the arc $\mathcal{B}$ and show that this choice makes $\Im\{h(z)\}$ continuous.  We firstly characterize the real points of $\mathcal{Z}$.  We already know that $\mathcal{Z}$ contains the real interval $[z_\mathrm{L},0]$ and the real point $z_\mathrm{P}$.  Moreover, it is easy to see that these points exhaust $\mathcal{Z}\cap\mathbb{R}$.  Indeed, since $S(z)$ is bounded away from zero for real $z\neq z_\mathrm{P}$ and $\mathrm{sgn}(S(z))=\mathrm{sgn}(z-z_\mathrm{P})$, it follows from \eqref{eq:2hprime} and \eqref{eq:ImYPlus} (see also Figure~\ref{fig:Ycurve}) that $\Im\{h_+(z)\}$ is strictly increasing for
$z^+<z<z_\mathrm{P}$ and strictly decreasing for $z<z_\mathrm{L}$, for $0<z<z^+$, and for $z>z_\mathrm{P}$.  

Next we consider the points of $\mathcal{Z}$ in the upper half-plane near the real axis.  According to Proposition~\ref{prop:Yzeroscx}, for $t\ge 0$ there exists a unique simple zero $z=\xi$ of $Y(z;q,q^*)$ in
the interval $(z_\mathrm{L},0)$, and since $S(z)$ is nonzero near this point, $z=\xi$ is obviously a simple saddle point of $\Im\{h(z)\}$ and therefore there is a unique branch of $\mathcal{Z}$ emanating from this point with a vertical tangent into the upper half-plane.  Also, since the real derivative of the boundary
value $\Im\{h_+(z)\}$ has a jump discontinuity at $z=z_\mathrm{P}$ with opposite nonzero left and right limits, there is a unique branch of $\mathcal{Z}$ emanating transversely to the real axis into the upper half-plane from $z=z_\mathrm{P}$ (here the tangent is not necessarily vertical, however).  The points $z=\xi\in (z_\mathrm{L},0)$ and $z=z_\mathrm{P}>0$ are the only real limit points of the part of $\mathcal{Z}$ in the open upper half-plane.  

Next, since $z=q$ is a simple root of $S(z)^2$, a local analysis of $S(z)Y(z;q,q^*)$ near $z=q$ using the fact (see Proposition~\ref{prop:Yzeroscx}) that $Y(q;q,q^*)\neq 0$ shows that
there exist exactly three arcs of $\mathcal{Z}$ emanating from $z=q\in\mathcal{Z}$ separated by angles of $2\pi/3$.

Let us now analyze $\mathcal{Z}$ assuming that $|z|$ is large with $\Im\{z\}>0$.  Since $S(s)=s[1-\tfrac{1}{2}(q+q^*)s^{-1}+O(s^{-2})]$ as $s\to\infty$,  we recall \eqref{eq:2hprime} and \eqref{eq:Yexpansion} to obtain
\begin{equation}
\begin{split}
\Im\{h(z)\}&=-\frac{1}{2}\Im\left\{\int_q^z S(s)Y(s;q,q^*)\,ds\right\}\\ &=\Im\left\{
-\frac{4t}{\alpha^2}z^2-\left(\frac{2}{\alpha^2}(\alpha x-t)+\pi i\right)z + O(z^{1/2})\right\},\quad z\to\infty,\quad \Im\{z\}>0.
\end{split}
\label{eq:Imh-asymp}
\end{equation}
For $t=0$, the part of the level set $\mathcal{Z}$ in the distant upper half-plane is clearly a smooth curve asymptotic to the straight line through the points $z=0$ and $z=-2x/\alpha + \pi i$ as $z\to\infty$.  
For $t\neq 0$ however,  $\mathcal{Z}$ is asymptotic to leading order to the union of the real and imaginary axes, and the computation of higher order corrections is necessary to determine whether the horizontal asymptotes actually correspond to branches of $\mathcal{Z}$ in the upper half-plane.  
Representing a branch of $\mathcal{Z}$ in the form $z=|z|e^{i\phi}$, $\phi=\phi(|z|)$, we substitute into \eqref{eq:Imh-asymp} and set $\Im\{h(z)\}=0$ to obtain the relation
\begin{equation}
-\frac{4t}{\alpha^2}\sin(2\phi)-\frac{2}{\alpha^2|z|}(\alpha x-t)\sin(\phi)-\frac{\pi}{|z|}\cos(\phi) + O(|z|^{-3/2})=0,\quad |z|\to\infty.
\end{equation}
Setting $|z|=\infty$ gives $\sin(2\phi)=0$, all roots $\phi=\phi_0$ of which are simple.  We may therefore apply the Implicit Function Theorem to continue these roots to finite $|z|$, yielding the asymptotic expansion
\begin{equation}
\phi(|z|)=\phi_0-\frac{(\alpha x-t)}{4t |z|}\frac{\sin(\phi_0)}{\cos(2\phi_0)}-\frac{\pi\alpha^2}{8t|z|}
\frac{\cos(\phi_0)}{\cos(2\phi_0)} + O(|z|^{-3/2}),\quad |z|\to\infty,
\label{eq:phi-expansion}
\end{equation}
where $\phi_0$ is a root of $\sin(2\phi_0)=0$.  Since we are only concerned with the upper half $z$-plane, we need to consider only the angles $\phi_0=0,\tfrac{1}{2}\pi,\pi$.  For $\phi_0=\frac{1}{2}\pi$, the asymptotic formula predicts a slight deformation of the angle in a direction depending on $x$ and $t$, but regardless this branch remains in the upper half-plane due to the dominant constant term in the angle $\phi$.  For $\phi_0=0$ and $\phi_0=\pi$, we have $\sin(\phi_0)=0$, and the surviving term at order $|z|^{-1}$ in the asymptotic expansion \eqref{eq:phi-expansion} indicates that for $t>0$ we have $|\tfrac{1}{2}\pi-\phi(|z|)|>0$ for large $|z|$, showing that these solutions \emph{do not} correspond to branches of $\mathcal{Z}$.  Therefore, we conclude that for $t\ge 0$ the part of $\mathcal{Z}$ in the distant upper half plane consists of exactly one smooth curve tending to infinity asymptotic to a non-horizontal straight line.  (If $t<0$ then the same analysis shows that when $\phi_0=0$ or $\phi_0=\pi$ we have instead $|\tfrac{1}{2}\pi-\phi(|z|)|<0$ for large $|z|$
and there are then three distinct branches of $\mathcal{Z}$ in the distant upper half-plane.)

So, we know that in the open upper half-plane there exist branches of $\mathcal{Z}$ emanating from the real axis only from $z=\xi\in(z_\mathrm{L},0)$ and from $z=z_\mathrm{P}$.  There also exist exactly three branches emanating from $z=q$, and (for $t\ge 0$) there is exactly one unbounded branch in the upper half-plane.  Since $\Im\{h(z)\}$ is harmonic for $\Im\{z\}>0$ except along the branch cut $\mathcal{B}$ across which $\Im\{h(z)\}$ simply changes sign, the Maximum Principle prohibits
the existence of any isolated components of  $\mathcal{Z}$ in the open upper half-plane, with the possible exception of curves enclosing $z=q$.  But the latter are also easily excluded from $\mathcal{Z}$ by local arguments.  Therefore, the part of $\mathcal{Z}$ in the upper half-plane consists exactly of the six branches described above, and these must be matched with each other.  The Maximum Principle also prevents any of the arcs emanating from $z=q$ from coinciding with one another.  Finally, since Proposition~\ref{prop:Yzeroscx} rules out any zeros of $Y(z;q,q^*)$ in the open upper half-plane for $t\ge 0$, there are no saddle points of $\Im\{h(z)\}$ for $\Im\{z\}>0$ and this means that the six arcs must be matched \emph{pairwise}, and no crossings are allowed.  Since none of the three branches emanating from $z=q$ can be matched with each other, we conclude that exactly one of these three branches is connected to each of the three distinct points $z=\xi$, $z=z_\mathrm{P}$, and $z=\infty$.  As the zero level set is Schwarz-symmetric, this concludes the complete description of $\mathcal{Z}$ in the case $t\ge 0$.  

Now suppose that the branch cut $\mathcal{B}$ is taken to coincide with the branch (a smooth arc, actually) of $\mathcal{Z}$ connecting $z=q$ with $z=z_\mathrm{P}$.  It then follows that $\Im\{h(z)\}$ can be continuously extended to $z\in \mathcal{B}\cup \mathcal{B}^*$.  Indeed,  since $\Im\{h(q)\}=0$, from \eqref{eq:2hprime} and the Fundamental Theorem of Calculus, the
mismatch $\Im\{h_+(z)\}-\Im\{h_-(z)\}$ of boundary values taken by $\Im\{h(z)\}$ on $\mathcal{B}$ is
\begin{equation}
\begin{split}
\Im\{h_+(z)\}-\Im\{h_-(z)\}&=-\frac{1}{2}\Im\left\{\mathop{\int_q^z}_{s\in \mathcal{B}} S_+(s)Y(s;q,q^*)\,ds-\mathop{\int_q^z}_{s\in\mathcal{B}} S_-(s)Y(s;q,q^*)\,ds\right\}\\
&=-\Im\left\{\mathop{\int_q^z}_{s\in \mathcal{B}} S_+(s)Y(s;q,q^*)\,ds\right\},\quad z\in \mathcal{B}.
\end{split}
\end{equation}
But this is the same as $2\Im\{h_+(z)\}$, which vanishes because $z\in \mathcal{B}\subset \mathcal{Z}$.
\end{proof}
From now on we assume that for all $t\ge 0$ and $x<x_\mathrm{c}(t)$, the arc $\mathcal{B}$ is chosen exactly as explained above, making $\Im\{h(z)\}$ continuous across $\mathcal{B}\cup \mathcal{B}^*$.  The zero level set $\mathcal{Z}$ clearly divides the open upper half-plane into three disjoint open regions.  We denote by $\mathcal{V}$ (the \emph{valley} of $\Im\{h(z)\}$ in the upper half-plane) the union of the two open regions that are separated by $\mathcal{B}$,  and we denote by $\mathcal{M}$ (the \emph{mountain} of $\Im\{h(z)\}$ in the upper half-plane) the remaining open region.  According to \eqref{eq:Imh-asymp} and the fact that $\mathcal{Z}$ contains all of the points for which $\Im\{h(z)\}=0$, we have the following strict inequalities for $\Im\{z\}>0$:
\begin{equation}
\Im\{h(z)\}>0,\quad z\in\mathcal{M}
\label{eq:Imh-Mountain}
\end{equation}
and 
\begin{equation}
\Im\{h(z)\}<0,\quad z\in\mathcal{V}.
\label{eq:Imh-Valley}
\end{equation}
The regions $\mathcal{M}$ and $\mathcal{V}$ are illustrated for several choices $(x,t)$ in
Figure~\ref{fig:Im-h-plots}.
\begin{figure}[h]
\begin{center}
\includegraphics[width=0.49\linewidth]{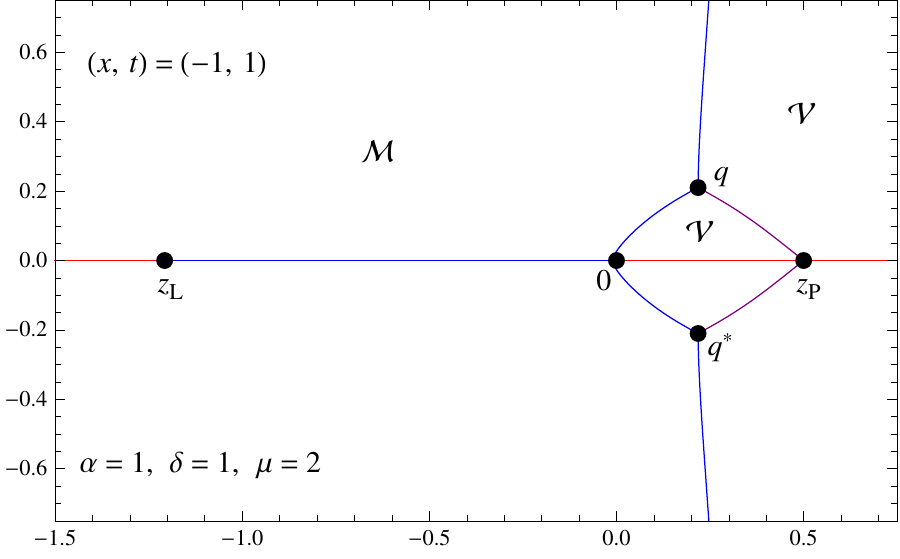}\hspace{0.01\linewidth}%
\includegraphics[width=0.49\linewidth]{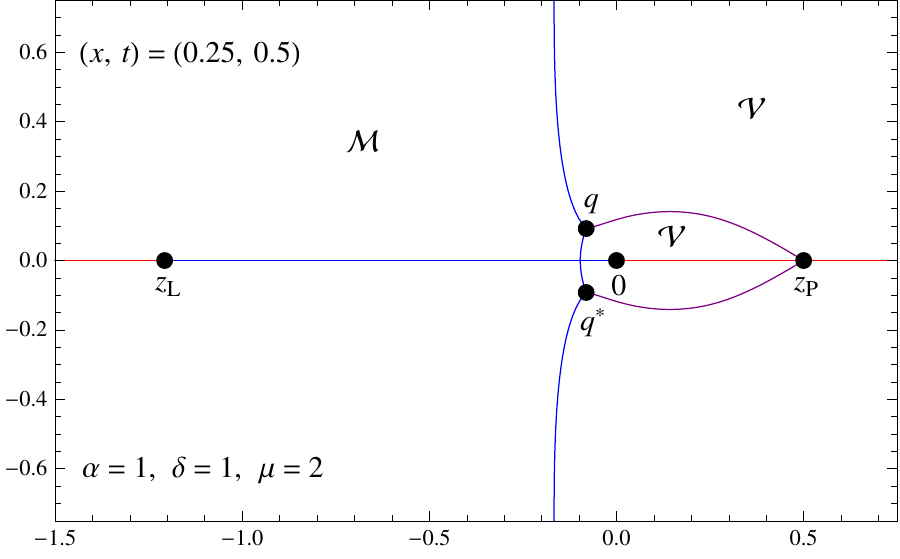}\\\vspace{3 pt}
\includegraphics[width=0.49\linewidth]{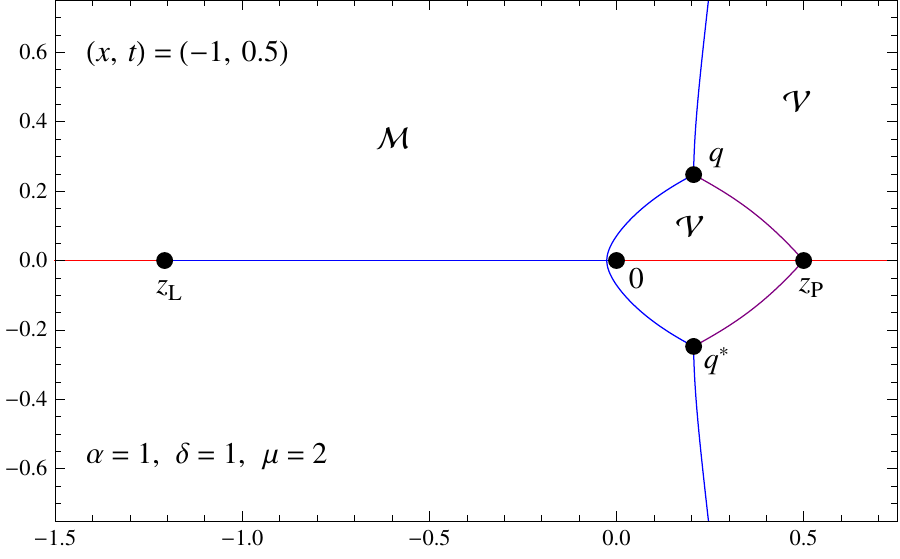}\hspace{0.01\linewidth}%
\includegraphics[width=0.49\linewidth]{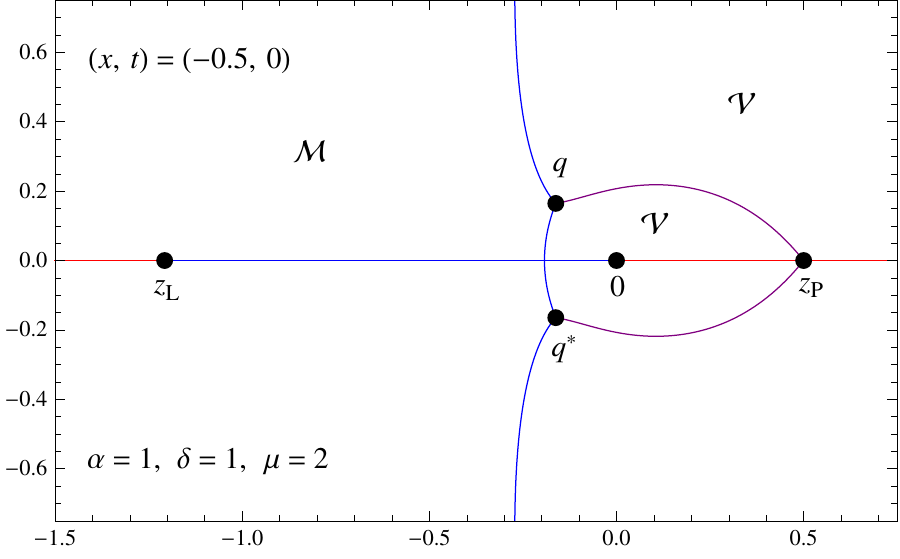}
\end{center}
\caption{Four numerically generated sign charts for $\Im\{h(z)\}$ indicating the regions $\mathcal{M}$ and $\mathcal{V}$, all for parameter values $\alpha=1$, $\delta=1$, and $\mu=2$.  Upper left (UL):  $(x,t)=(-1,1)$.  Upper right (UR):  $(x,t)=(0.25,0.5)$.  Lower left (LL):  $(x,t)=(-1,0.5)$.  Lower right (LR):  $(x,t)=(-0.5,0)$.  The red curves are branch cuts for $h(z)$, the blue curves are zero level curves of $\Im\{h(z)\}$, and the purple curve is $\mathcal{B}\cup\mathcal{B}^*$, which is both a branch cut of $h(z)$ (although $\Im\{h(z)\}$ is continuous there) and a zero level curve of $\Im\{h(z)\}$.}
\label{fig:Im-h-plots}
\end{figure}


\subsection{Key properties of $g(z)$ and $h(z)$}
With $g(z)$ and $h(z)$ now completely determined as analytic functions depending smoothly
on parameters $(x,t)$ with $t\ge 0$ and $x<x_\mathrm{c}(t)$, we now explain all of the important consequences of our finished construction.
\begin{proposition}
The function $g$ is analytic for $z\in\mathbb{C}\setminus (\mathcal{B}\cup \mathcal{B}^*)$, satisfies the Schwarz symmetry condition $g(z^*)=g(z)^*$, and also $g(0)=0$ while $g(\infty)$ is well-defined and finite.  The function $h$ is analytic for $z\in\mathbb{C}\setminus ((-\infty,z_\mathrm{L}]\cup [0,+\infty)\cup \mathcal{B}\cup \mathcal{B}^*)$, satisfies the Schwarz symmetry condition $h(z^*)=h(z)^*$, and $\Im\{h(z)\}$ is continuous for $\Im\{z\}\neq 0$.  There is a real constant $\kappa=\kappa(x,t)$ such that 
\begin{equation}
h_+(z)+h_-(z)=\kappa,\quad z\in \mathcal{B}\cup \mathcal{B}^*.
\label{eq:h-equilibrium}
\end{equation}
Also, the strict inequalities \eqref{eq:Imh-Mountain}--\eqref{eq:Imh-Valley} hold for $\Im\{z\}>0$, with the only excluded points in the upper half-plane satisfying $\Im\{h(z)\}=0$.  From \eqref{eq:Imh-asymp} we have the asymptotic condition
\begin{equation}
\Im\{h(z)\}=-\left[\frac{2\pi}{\alpha^2}+\frac{8t}{\alpha^2}\Im\{z\}\right]\Re\{z\} + O(|\Re\{z\}|^{1/2}),\quad\Re\{z\}\to -\infty,\quad 0\le\Im\{z\}=O(1).
\label{eq:Im-h-tail}
\end{equation}
Finally, $h(z)$ is continuous at $z=q$, and
\begin{equation}
h'(z)^2=(z-q)v(z)
\label{eq:hprimeSquared}
\end{equation}
where $v$ is analytic and non-vanishing in a neighborhood of $z=q$.
\end{proposition}
\begin{proof}
It only remains to prove \eqref{eq:h-equilibrium} and \eqref{eq:hprimeSquared}.  But, these both follow from the representation \eqref{eq:2hprime}.  Indeed, 
taking into account that $S(z)$ changes sign across $\mathcal{B}$ while $Y(z;q,q^*)$ is analytic
in a neighborhood of $\mathcal{B}$, 
we easily see that the sum of boundary values $h_+'(z)+h_-'(z)$ vanishes identically for $z\in\mathcal{B}$.  Therefore, $h_+(z)+h_-(z)$ is constant along $\mathcal{B}$, and since $\mathcal{B}$ is part of the zero level set $\mathcal{Z}$ of $\Im\{h(z)\}$, it is obvious that the constant value taken by $h_+(z)+h_-(z)$ for $z\in\mathcal{B}$ is real.  It then follows by Schwarz symmetry of $h(z)$ that the same identity holds for $z\in\mathcal{B}^*$.  This proves \eqref{eq:h-equilibrium}.  The formula \eqref{eq:hprimeSquared} also follows directly from \eqref{eq:2hprime}
using Proposition~\ref{prop:Yzeroscx}.
\end{proof}

With the parametric dependence of $h$ on $(x,t)$ completely determined, we can consider
the partial derivatives of $h$ with respect to $x$ and $t$.  Recalling the definition \eqref{eq:hdef}
of $h$ in terms of $g$, $\theta$, and $\Phi$, and the fact that $\Phi(z)$ is a function independent of $(x,t)$, we immediately deduce that these partial derivatives are analytic in a larger domain than
is $h$ itself, namely for $z\in\mathbb{C}\setminus (\mathcal{B}\cup\mathcal{B}^*)$.  We claim that
the partial derivatives are given by the following simple explicit formulae:
\begin{equation}
\frac{\partial h}{\partial x}(z)=\frac{1}{2}\frac{\partial\kappa}{\partial x}-\frac{2}{\alpha}S(z),\quad
z\in\mathbb{C}\setminus (\mathcal{B}\cup\mathcal{B}^*),
\label{eq:hx-def}
\end{equation}
and
\begin{equation}
\frac{\partial h}{\partial t}(z)=\frac{1}{2}\frac{\partial\kappa}{\partial t} -\frac{4}{\alpha^2}\left(z+\Re\{q\}-\frac{1}{2}\right)S(z),\quad z\in\mathbb{C}\setminus (\mathcal{B}\cup\mathcal{B}^*).
\label{eq:ht-def}
\end{equation}
Indeed, these formulae exhibit the correct domain of analyticity and capture the correct boundary conditions on $\mathcal{B}\cup\mathcal{B}^*$
following from \eqref{eq:h-equilibrium}:
\begin{equation}
\frac{\partial h_+}{\partial x}(z)+\frac{\partial h_-}{\partial x}(z)=\frac{\partial\kappa}{\partial x}
\quad\text{and}\quad
\frac{\partial h_+}{\partial t}(z)+\frac{\partial h_-}{\partial t}(z)=\frac{\partial\kappa}{\partial t},
\quad z\in\mathcal{B}\cup\mathcal{B}^*.
\end{equation}
Also, since the partial derivatives of $g$ are bounded for $z=\infty$, the principal parts of the partial derivatives of $h$ at $z=\infty$ must agree with those of $\theta$, and by expansion of $S(z)$ for large $z$ it is easy to confirm that the above formulae indeed satisfy
\begin{equation}
\frac{\partial h}{\partial x}(z)=\frac{\partial \theta}{\partial x}(z)+O(1)=-\frac{2}{\alpha}z+O(1)\quad\text{and}\quad
\frac{\partial h}{\partial t}(z)=\frac{\partial\theta}{\partial t}(z)+O(1)=-\frac{4}{\alpha^2}z^2+\frac{2}{\alpha^2}z+O(1),\quad z\to\infty.
\end{equation}
Since $g(0)=0$ for all $(x,t)$ it follows from \eqref{eq:hdef} that the partial derivatives of $h$ must agree with those of $\theta$ exactly at $z=0$.  Imposing these conditions on \eqref{eq:hx-def}
and \eqref{eq:ht-def} proves that the partial derivatives of $\kappa$ can be expressed explicitly in terms of $q(x,t)$.  Indeed, since $S(0)=-|q|$, 
\begin{equation}
0=\frac{\partial h}{\partial x}(0)-\frac{\partial\theta}{\partial x}(0)=\frac{1}{2}\frac{\partial\kappa}{\partial x}+\frac{2}{\alpha}|q|-\frac{1}{2\alpha}\quad\implies\quad \frac{\partial\kappa}{\partial x}=\frac{1}{\alpha}\left(1-4|q|\right),
\label{eq:kappa-x}
\end{equation}
and
\begin{equation}
0=\frac{\partial h}{\partial t}(0)-\frac{\partial\theta}{\partial t}(0)=\frac{1}{2}\frac{\partial\kappa}{\partial t}+\frac{4}{\alpha^2}\left(\Re\{q\}-\frac{1}{2}\right)|q|+\frac{1}{4\alpha^2}\quad\implies\quad
\frac{\partial\kappa}{\partial t}=\frac{4}{\alpha^2}\left(1-2\Re\{q\}\right)|q|-\frac{1}{2\alpha^2}.
\end{equation}

Similar calculations allow us to establish a direct expression of $\kappa=\kappa(x,t)$ itself.  
Indeed, since $g(z)$ is analytic for $z\in\mathbb{C}\setminus (\mathcal{B}\cup\mathcal{B}^*)$,
is well-defined for $z=\infty$, and satisfies $g_+(z)+g_-(z)=2\theta(z)+\Phi(z)-\kappa$ for $z\in\mathcal{B}\cup\mathcal{B}^*$, it follows that it must be given by
\begin{equation}
\begin{split}
g(z)&=\frac{S(z)}{2\pi i}\int_{\mathcal{B}\cup\mathcal{B}^*}\frac{2\theta(s)+\Phi(s)-\kappa}{S_+(s)(s-z)}\,ds\\ & = \theta(z)-\frac{1}{2}\kappa+\frac{2S(z)}{\alpha^2}(\alpha x-t) + \frac{4tS(z)}{\alpha^2}(z+\Re\{q(x,t)\})+\frac{S(z)}{2\pi i}\int_{\mathcal{B}\cup\mathcal{B}^*}\frac{\Phi(s)\,ds}{S_+(s)(s-z)},\quad
z\in\mathbb{C}\setminus (\mathcal{B}\cup\mathcal{B}^*).
\end{split}
\end{equation}
Enforcing the condition that $g(0)=0$ and using the fact that $S(0)=-|q(x,t)|$ then gives a formula for $\kappa(x,t)$:
\begin{equation}
\kappa(x,t)=\frac{1}{2\alpha^2}(2\alpha x-t) -\frac{4}{\alpha^2}|q(x,t)|(\alpha x-t)-\frac{8t}{\alpha^2}|q(x,t)|\Re\{q(x,t)\}-\frac{|q(x,t)|}{\pi i}\int_{\mathcal{B}\cup\mathcal{B}^*}\frac{\Phi(s)\,ds}{sS_+(s)}.
\end{equation}
In particular, for $t=0$ we have
\begin{equation}
\kappa(x,0)=\frac{x}{\alpha}\left(1-4|q(x,0)|\right)-\frac{|q(x,0)|}{\pi i}\int_{\mathcal{B}\cup\mathcal{B}^*}\frac{\Phi(s)\,ds}{sS_+(s)},\quad x<x_\mathrm{c},
\end{equation}
where in the last integral we use $q=q(x,0)$ and the corresponding arc $\mathcal{B}$.  This formula can be simplified even further by taking (carefully) the limit $x\uparrow x_\mathrm{c}$ in which $q(x,0)\to z_\mathrm{c}<0$:
\begin{equation}
\lim_{x\uparrow x_\mathrm{c}}\kappa(x,0)=\frac{x_\mathrm{c}}{\alpha}(1+4z_\mathrm{c})+2\Phi(0)-
\Phi(z_\mathrm{c})+\frac{2z_\mathrm{c}}{\pi}\int_0^{z_\mathrm{P}}\frac{\Im\{\Phi_+(s)\}\,ds}{s(s-z_\mathrm{c})}.
\label{eq:KappaCrit}
\end{equation}
Note that $\Phi(0)$ is a well-defined real number, and the fact that the integral is convergent at $s=0$ follows from \eqref{eq:ImPhiPlus} which shows that $\Im\{\Phi_+(z)\}$ vanishes linearly as $z\downarrow 0$.

\section{Steepest Descent Analysis of the Riemann-Hilbert Problem}
\label{sec:RHP-deform}
\subsection{Opening of Lenses and Introduction of $g(z)$}
We begin by making an explicit modification of the matrix $\mathbf{N}(z)$ by ``opening lenses'' according to the diagram in Figure~\ref{fig:EllipticContours}. 
\begin{figure}[h]
\begin{center}
\includegraphics{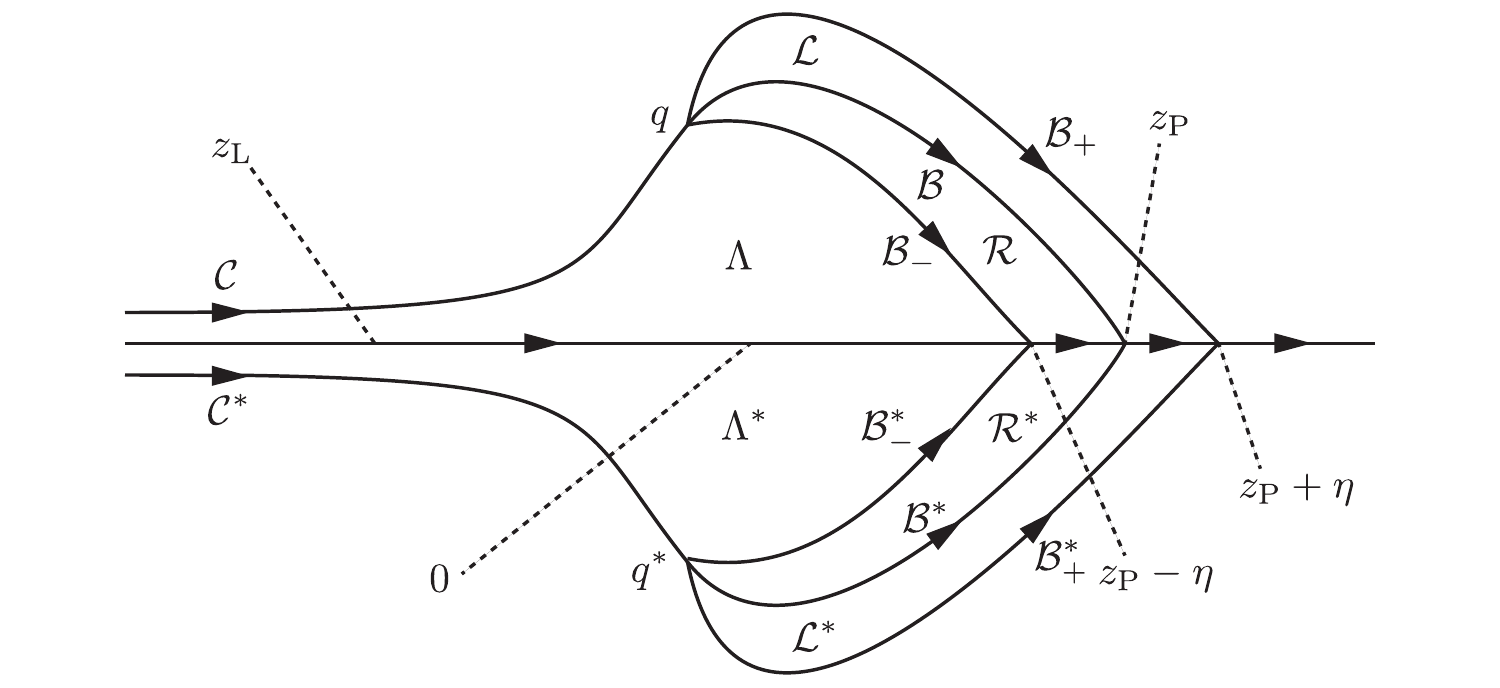}
\end{center}
\caption{The domains $\Lambda$, $\mathcal{L}$, and $\mathcal{R}$, along with their mirror images in the lower half-plane.  Note that $\partial\mathcal{L}\cap\partial\mathcal{R}$ coincides with the contour arc $\mathcal{B}$.  The solid curves comprise a contour denoted $\Sigma$. }
\label{fig:EllipticContours}
\end{figure}
The most important features that we require of the contour in this figure are:
\begin{itemize}
\item The contour arc $\mathcal{C}$ lies entirely within the mountainous region $\mathcal{M}$ for $\Im\{h(z)\}$, and is asymptotically horizontal as $\Re\{z\}\to -\infty$.
\item The contour arc $\mathcal{B}_+$ lies entirely within the unbounded component of the valley region $\mathcal{V}$ for $\Im\{h(z)\}$.
\item The contour arc $\mathcal{B}_-$ lies entirely within the bounded component of the valley region $\mathcal{V}$ for $\Im\{h(z)\}$.
\item The parameter $\eta$ is sufficiently small that $z_\mathrm{P}-\eta>z^+$.
\end{itemize}
The modification is defined as follows:
\begin{equation}
\mathbf{O}(z):=\mathbf{N}(z)e^{i\theta(z;x,t)\sigma_3/\epsilon}\begin{bmatrix}
1 &  
-
\tilde{E}(z)e^{i\Phi(z)/\epsilon}\\ 0 & 1\end{bmatrix}e^{-i\theta(z;x,t)\sigma_3/\epsilon},\quad
z\in\Lambda,
\end{equation}
\begin{equation}
\mathbf{O}(z):=\mathbf{N}(z)e^{i\theta(z;x,t)\sigma_3/\epsilon}\begin{bmatrix}
1 & 0 \\ 
-
\tilde{E}(z)^{-1}e^{-i\Phi(z)/\epsilon} & 1\end{bmatrix}\tilde{E}(z)^{\sigma_3/2}e^{-i\theta(z;x,t)\sigma_3/\epsilon},\quad z\in \mathcal{L},
\end{equation}
\begin{equation}
\mathbf{O}(z):=\mathbf{N}(z)e^{i\theta(z;x,t)\sigma_3/\epsilon}\begin{bmatrix}
0 & 
-
\tilde{E}(z)e^{i\Phi(z)/\epsilon}\\
\tilde{E}(z)^{-1}e^{-i\Phi(z)/\epsilon} & 1\end{bmatrix}
\tilde{E}(z)^{\sigma_3/2}
e^{-i\theta(z;x,t)\sigma_3/\epsilon},\quad z\in \mathcal{R},
\end{equation}
\begin{equation}
\mathbf{O}(z):=\sigma_1\mathbf{O}(z^*)^*\sigma_1,\quad z\in \Lambda^*\cup \mathcal{L}^*\cup \mathcal{R}^*,
\label{eq:cxOSchwarz}
\end{equation}
and elsewhere we set $\mathbf{O}(z):=\mathbf{N}(z)$.  Recalling that $\tilde{E}(z)$ is an analytic nonvanishing function for $z$ in the open upper half-plane, we see that this defines the matrix $\mathbf{O}(z)$ as an analytic function of $z$ in the complement of the system of contours illustrated in Figure~\ref{fig:EllipticContours}.  Then, we introduce the function $g(z)$ defined in \S\ref{sec:gfunction} by setting
\begin{equation}
\mathbf{P}(z):=\mathbf{O}(z)e^{ig(z)\sigma_3/\epsilon}.
\end{equation}
Since $g(z)$ is analytic for $z\in \mathbb{C}\setminus{\mathcal{B}\cup \mathcal{B}^*}$, and since $\mathcal{B}=\partial\mathcal{L}\cap\partial\mathcal{R}$, $\mathbf{P}(z)$ has the same domain of analyticity as does $\mathbf{O}(z)$.  

The explicit transformation of $\mathbf{N}(z)$ into $\mathbf{P}(z)$ differs from a more standard
choice that is based upon the factorization of the jump matrix for $\mathbf{N}(z)$ (say, for $z<0$)
given by
\begin{equation}
\begin{bmatrix}1 & -s(z)e^{2i\theta(z;x,t)/\epsilon}\\s(z)^*e^{-2i\theta(z;x,t)/\epsilon} & 1-|s(z)|^2
\end{bmatrix}=\begin{bmatrix}1 & 0\\ s(z)^*e^{-2i\theta(z;x,t)/\epsilon}& 1\end{bmatrix}\begin{bmatrix}1 &-s(z)e^{2i\theta(z;x,t)/\epsilon} \\ 0& 1\end{bmatrix}.
\label{eq:two-factor}
\end{equation}
This more standard choice was used, for example, in the analysis of Tovbis, Venakides, and Zhou of the semiclassical limit for the focusing nonlinear Schr\"odinger equation \cite{TovbisVZ04} with special initial data for which there were no discrete eigenvalues (poles of $\mathbf{N}(z)$) and for which the reflection coefficient analogous to $s(z)$ was given in terms of Euler gamma functions.
It has the advantage that after opening lenses there remains no jump across the real axis for $z<z_\mathrm{P}$.  However, opening lenses based on the factorization \eqref{eq:two-factor}
actually sets into motion a chain of unfortunate events leading to more difficult analysis and even some ambiguity about the scope of the resulting asymptotic formulae.  The main problem with using the factorization \eqref{eq:two-factor} to open lenses is the presence in the upper half-plane near $z=z_\mathrm{P}$ of the \emph{phantom poles} of $s(z)$.  These poles both constrain the possible location of the contours near $z=z_\mathrm{P}$ (to avoid introducing new singularities into the Riemann-Hilbert problem for $\mathbf{P}(z)$ jump matrices involving $s(z)$ must not be deformed through the phantom poles) and also lead to nonuniformity in the Stirling asymptotics
of $s(z)$ near $z=z_\mathrm{P}$ that ruins the approximation $E(z)\approx 1$.  The authors of \cite{TovbisVZ04} dealt with the latter issue by installing an unusual local parametrix near the point analogous to $z=z_\mathrm{P}$ in their problem; the parametrix is difficult to analyze because it is not explicit (its existence relies on abstract Fredholm theory) and it depends on $\epsilon$ in a way that must be carefully understood.  By contrast, the reader will see that the approach we use here (which is based on extending not $s(z)$ but rather the product $(1+e^{2f(z)/\epsilon})s(z)$ --- related to the function $\tilde{E}(z)$ by \eqref{eq:tildeEdef} --- into the upper half-plane) \emph{removes all difficulties with the phantom poles near $z=z_\mathrm{P}$ and moreover avoids completely the need for any sort of local parametrix near this point}.  In fact, $\tilde{E}(z)$ is analytic throughout the upper half-plane, and it will turn out that in our approach the contribution to the error from the neighborhood of $z=z_\mathrm{P}$ is exponentially small.

We recall that $g(0)=0$ and also that $\tilde{E}_+(0)=0$ (the latter following from the fact that $s(0)=0$ while $\Phi(0)$ is well-defined and finite) and hence it follows that $\mathbf{P}(z)$ satisfies the normalization condition
\begin{equation}
\mathbf{P}(0)=\mathbb{I}
\end{equation}
(the value is the same whether the limit is taken from the upper or lower half-plane).  Since $g(\infty)$ is well-defined, $\mathbf{P}(z)$ tends to a diagonal limit as $z\to\infty$ for each fixed $\epsilon$.  To describe the jump discontinuities of $\mathbf{P}(z)$ across the arcs of the jump contour pictured in Figure~\ref{fig:EllipticContours}, we suppose that all contour arcs are oriented left-to-right (that is, from $-\infty$ toward $q$, $q^*$, or $z_\mathrm{P}-\eta$, 
from $q$ or $q^*$ toward $z_\mathrm{P}-\eta$, $z_\mathrm{P}$, or $z_\mathrm{P}+\eta$, 
from $z_\mathrm{P}-\eta$ toward $z_\mathrm{P}$, from $z_\mathrm{P}$ toward $z_\mathrm{P}+\eta$, and finally from $z_\mathrm{P}+\eta$ toward $+\infty$).  

Then 
the jump conditions satisfied by $\mathbf{P}(z)$ along the real $z$-axis are as follows:
\begin{equation}
\mathbf{P}_+(z)=\mathbf{P}_-(z)\begin{bmatrix}1 & e^{2f(z)/\epsilon}
s(z)e^{2i(\theta(z;x,t)-g(z))/\epsilon}
\\-e^{2f(z)/\epsilon}s(z)^*e^{-2i(\theta(z;x,t)-g(z))/\epsilon}
&
1-e^{4f(z)/\epsilon}|s(z)|^2
\end{bmatrix},\quad z<0,
\label{eq:cxPjumpznegative}
\end{equation}
\begin{multline}
\mathbf{P}_+(z)=i^{\sigma_3}\mathbf{P}_-(z)i^{-\sigma_3}\begin{bmatrix}
1 & e^{2f(z)/\epsilon}s(z)e^{2i(\theta(z;x,t)-g(z))/\epsilon}
\\
e^{2f(z)/\epsilon}s(z)^*e^{-2i(\theta(z;x,t)-g(z))/\epsilon}
&
1+e^{4f(z)/\epsilon}|s(z)|^2
\end{bmatrix},\\
0<z<z_\mathrm{P}-\eta,
\label{eq:cxPjumpzpositive1}
\end{multline}
\begin{equation}
\mathbf{P}_+(z)=i^{\sigma_3}\mathbf{P}_-(z)i^{-\sigma_3}\begin{bmatrix}
A_{11}(z)
& A_{12}(z)e^{2i(\theta(z;x,t)-g(z))/\epsilon}
\\ A_{21}(z)e^{-2i(\theta(z;x,t)-g(z))/\epsilon} &A_{22}(z)\end{bmatrix},\quad
z_\mathrm{P}-\eta<z<z_\mathrm{P},
\label{eq:cxPjumpzpositive2}
\end{equation}
where
\begin{equation}
A_{11}(z):=\frac{1+|s(z)|^2}{(1+e^{2f(z)/\epsilon})|s(z)|}e^{\Im\{\Phi_+(z)\}/\epsilon},
\end{equation}
\begin{equation}
A_{12}(z)=A_{21}(z)^*:=\frac{
1-|s(z)|^2e^{2f(z)/\epsilon}
}{(1+e^{2f(z)/\epsilon})|s(z)|}e^{i\Re\{\Phi_+(z)\}/\epsilon},
\end{equation}
and where $A_{11}(z)A_{22}(z)-A_{12}(z)A_{21}(z)=1$,
\begin{equation}
\mathbf{P}_+(z)=i^{\sigma_3}\mathbf{P}_-(z)i^{-\sigma_3}\begin{bmatrix}
B_{11}(z) & B_{12}(z)e^{2i(\theta(z;x,t)-g(z))/\epsilon}\\
B_{21}(z)e^{-2i(\theta(z;x,t)-g(z))/\epsilon} & B_{22}(z)
\end{bmatrix},\quad
z_\mathrm{P}<z<z_\mathrm{P}+\eta,
\label{eq:cxPjumpzpositive3}
\end{equation}
where
\begin{equation}
B_{22}(z):=\frac{1+|s(z)|^2}{(1+e^{2f(z)/\epsilon})|s(z)|}e^{-\Im\{\Phi_+(z)\}/\epsilon},
\end{equation}
\begin{equation}
B_{12}(z)=B_{21}(z)^*:=\frac{
|s(z)|^2e^{2f(z)/\epsilon}-1
}{(1+e^{2f(z)/\epsilon})|s(z)|}e^{i\Re\{\Phi_+(z)\}/\epsilon},
\end{equation}
and where $B_{11}(z)B_{22}(z)-B_{12}(z)B_{21}(z)=1$, and finally
\begin{equation}
\mathbf{P}_+(z)=i^{\sigma_3}\mathbf{P}_-(z)i^{-\sigma_3}\begin{bmatrix}1 & -s(z)e^{2i(\theta(z;x,t)-g(z))/\epsilon}\\
-s(z)^*e^{-2i(\theta(z;x,t)-g(z))/\epsilon} & 1+|s(z)|^2\end{bmatrix},\quad z>z_\mathrm{P}+\eta.
\label{eq:cxPjumpzpositive4}
\end{equation}

These jump conditions may appear complicated, and this is the price to be paid for eschewing the standard two-factor factorization \eqref{eq:two-factor} in favor of a more complicated one.  However, we will now prove that all of the jump conditions for $\mathbf{P}(z)$ along $\mathbb{R}$ amount to exponentially small jump discontinuities, given the condition $\eta<z_\mathrm{P}-z^+$.  Indeed, using \eqref{eq:sboundznegative} and \eqref{eq:sboundzpositive} and recalling that $\theta(z;x,t)$ and $g(z)$ are real for $z\in\mathbb{R}$, we easily see that \eqref{eq:cxPjumpznegative} takes the form
\begin{equation}
\mathbf{P}_+(z)=\mathbf{P}_-(z)\left(\mathbb{I} + O(e^{2f(z)/\epsilon})\right),\quad z<0,
\label{eq:cxPjumpznegativerewrite}
\end{equation}
that \eqref{eq:cxPjumpzpositive1} takes the form
\begin{equation}
\mathbf{P}_+(z)= i^{\sigma_3}\mathbf{P}_-(z)i^{-\sigma_3}\left(\mathbb{I}+O(e^{f(z)/\epsilon})\right),\quad 0<z<z_\mathrm{P}-\eta,
\label{eq:cxPjumpzpositive1rewrite}
\end{equation}
and that \eqref{eq:cxPjumpzpositive4} takes the form
\begin{equation}
\mathbf{P}_+(z)=i^{\sigma_3}\mathbf{P}_-(z)i^{-\sigma_3}\left(\mathbb{I}+O(e^{-f(z)/\epsilon})\right),\quad z>z_\mathrm{P}+\eta.
\label{eq:cxPjumpzpositive4rewrite}
\end{equation}
In each case, the jump matrix is  a uniformly exponentially small (in $\epsilon$) perturbation of the identity that also decays exponentially as $z\to \pm\infty$.  Also, due to the estimate \eqref{eq:sboundorigin}, we have that 
\begin{equation}
\mathbf{P}_+(z)=\mathbf{P}_-(z)\left(\mathbb{I}+O(|z|^{1/2}\epsilon^{-1}e^{-M/\epsilon})\right),\quad -\eta<z<0,
\end{equation}
and
\begin{equation}
\mathbf{P}_+(z)=i^{\sigma_3}\mathbf{P}_-(z)i^{-\sigma_3}\left(\mathbb{I}+O(|z|^{1/2}\epsilon^{-1}e^{-M/\epsilon})\right),\quad 0<z<\eta
\end{equation}
both hold uniformly for some $\eta>0$ and $M>0$.  
Next, since according to \eqref{eq:ImPhiPlus} we have $\Im\{\Phi_+(z)\}=-|f(z)|$ for $z>z^+$, we see that for $z_\mathrm{P}-\eta<z<z_\mathrm{P}$ we have
\begin{equation}
\begin{split}
A_{11}(z) &= \frac{1+|s(z)|^2}{(1+e^{2f(z)/\epsilon})|s(z)|}e^{f(z)/\epsilon}\\&=1+\frac{(|s(z)|e^{f(z)/\epsilon}-1)(|s(z)|e^{f(z)/\epsilon}-e^{2f(z)/\epsilon})}{|s(z)|e^{f(z)/\epsilon}(1+e^{2f(z)/\epsilon})}\\
&=1+\text{exponentially small in $\epsilon$},\quad\text{uniformly for $z_\mathrm{P}-\eta<z<z_\mathrm{P}$},
\end{split}
\end{equation}
where we have used \eqref{eq:modssquaredlargez} and the fact that $e^{2f(z)/\epsilon}\le 1$ for $z\le z_\mathrm{P}$.  Similarly, from \eqref{eq:modssquaredlargez},
\begin{equation}
\begin{split}
|A_{12}(z)|=|A_{21}(z)|&=\frac{|
1-|s(z)|^2e^{2f(z)/\epsilon}
|}{2|s(z)|e^{f(z)/\epsilon}\cosh(f(z)/\epsilon)}\\
&\le \frac{|
1-|s(z)|^2e^{2f(z)/\epsilon}
|}{2|s(z)|e^{f(z)/\epsilon}}\\ &=\text{exponentially small in $\epsilon$},\quad \text{uniformly for $z_\mathrm{P}-\eta<z<z_\mathrm{P}$}.
\end{split}
\end{equation}
Since $f(z)$ changes sign at $z_\mathrm{P}$, we have $e^{-\Im\{\Phi_+(z)\}/\epsilon}=e^{f(z)/\epsilon}$ for $z>z_\mathrm{P}$, and therefore
\begin{equation}
\begin{split}
B_{22}(z)&=\frac{1+|s(z)|^2}{(1+e^{2f(z)/\epsilon})|s(z)|}e^{f(z)/\epsilon}\\
&=1+\frac{(|s(z)e^{f(z)/\epsilon}-1)(|s(z)e^{f(z)/\epsilon}e^{-2f(z)/\epsilon}-1)}{|s(z)|e^{f(z)/\epsilon}(1+e^{-2f(z)/\epsilon})}\\
&=1+\text{exponentially small in $\epsilon$},\quad\text{uniformly for $z_\mathrm{P}<z<z_\mathrm{P}+\eta$},
\end{split}
\end{equation}
where we have again used \eqref{eq:modssquaredlargez} and the fact that $e^{-2f(z)/\epsilon}\le 1$ for $z\ge z_\mathrm{P}$.  Similarly,
\begin{equation}
\begin{split}
|B_{12}(z)|=|B_{21}(z)|&=\frac{|
|s(z)|^2e^{2f(z)/\epsilon}-1
|}{2|s(z)|e^{f(z)/\epsilon}\cosh(f(z)/\epsilon)}\\
&\le\frac{|
|s(z)|^2e^{2f(z)/\epsilon}-1
|}{2|s(z)|e^{f(z)/\epsilon}}\\
&=\text{exponentially small in $\epsilon$},\quad\text{uniformly for $z_\mathrm{P}<z<z_\mathrm{P}+\eta$}.
\end{split}
\end{equation}
Therefore \eqref{eq:cxPjumpzpositive2} and \eqref{eq:cxPjumpzpositive3} can be combined to read
\begin{equation}
\mathbf{P}_+(z)=i^{\sigma_3}\mathbf{P}_-(z)i^{-\sigma_3}(\mathbb{I}+\text{uniformly exponentially small in $\epsilon$}),\quad |z-z_\mathrm{P}|<\eta.
\label{eq:cxPjumpzpositive2-3rewrite}
\end{equation}

Now we describe the jump conditions for $\mathbf{P}(z)$ across the various non-real contours pictured in Figure~\ref{fig:EllipticContours} (it suffices to consider those in the open upper half-plane only due to the symmetry 
\begin{equation}
\mathbf{P}(z)=\sigma_1\mathbf{P}(z^*)^*\sigma_1
\end{equation}
inherited from \eqref{eq:NSchwarz} and \eqref{eq:cxOSchwarz} via the fact that $g(z^*)=g(z)^*$).
First consider the arc $\mathcal{B}$ separating the regions $\mathcal{L}$ and $\mathcal{R}$ in Figure~\ref{fig:EllipticContours}, oriented from $z=q$ toward $z=z_\mathrm{P}$.  Then recalling the function $h(z)$ defined in terms of $g(z)$ by \eqref{eq:hdef}, we have
\begin{equation}
\mathbf{P}_+(z)=\mathbf{P}_-(z)\begin{bmatrix}0 & 
e^{i(h_+(z)+h_-(z))/\epsilon}\\
-
e^{-i(h_+(z)+h_-(z))/\epsilon} & 0\end{bmatrix}=
\mathbf{P}_-(z)\begin{bmatrix}0 &
e^{i\kappa/\epsilon}\\
-
e^{-i\kappa/\epsilon} & 0\end{bmatrix},\quad z\in \mathcal{B},
\label{eq:cxPjumpB}
\end{equation}
where we have used \eqref{eq:h-equilibrium}. 
 Next, let $\mathcal{B}_+$ denote the oriented arc in Figure~\ref{fig:EllipticContours} from $z=q$ to $z=z_\mathrm{P}+\eta$.  Then
\begin{equation}
\mathbf{P}_+(z)=\mathbf{P}_-(z)\begin{bmatrix}1 & 0 \\ 
e^{-2ih(z)/\epsilon} & 1\end{bmatrix}\tilde{E}(z)^{-\sigma_3/2},\quad z\in \mathcal{B}_+.
\label{eq:cxPjumpBplus}
\end{equation}
Similarly, let $\mathcal{B}_-$ denote the oriented arc in Figure~\ref{fig:EllipticContours} from $z=q$ to $z=z_\mathrm{P}-\eta$. Then
\begin{equation}
\mathbf{P}_+(z)=\mathbf{P}_-(z)\tilde{E}(z)^{\sigma_3/2}\begin{bmatrix}1 & 0 \\ 
e^{-2ih(z)/\epsilon}&1\end{bmatrix},\quad z\in \mathcal{B}_-.
\label{eq:cxPjumpBminus}
\end{equation}
Finally, consider the unbounded arc of $\Sigma$ in the upper half-plane, which we denote by $\mathcal{C}$.  We assume that $\mathcal{C}$ is oriented from $z=-\infty$ toward $z=q$.  Then
\begin{equation}
\mathbf{P}_+(z)=\mathbf{P}_-(z)\tilde{E}(z)^{\sigma_3/2}\begin{bmatrix}1 & 
e^{2ih(z)/\epsilon}\\0 & 1\end{bmatrix}\tilde{E}(z)^{-\sigma_3/2},\quad z\in \mathcal{C}.
\label{eq:cxPjumpCminusB}
\end{equation}
It follows from the inequalities \eqref{eq:Imh-Mountain}--\eqref{eq:Imh-Valley}  and the asymptotic estimate \eqref{eq:tildeEasymp} that all three of the jump conditions \eqref{eq:cxPjumpBplus}--\eqref{eq:cxPjumpCminusB} are of the form $\mathbf{P}_+(z)=\mathbf{P}_-(z)(\mathbb{I}+O(\epsilon))$ uniformly for $z$ on the relevant contours bounded away from the common endpoint $z=q$.  In fact, for such $z$ the dominant contribution to the error comes from the matrix factors $\tilde{E}(z)^{\pm\sigma_3/2}$ as $e^{\pm 2ih(z)/\epsilon}$ is in each case exponentially small.  Moreover, due to the estimate \eqref{eq:Im-h-tail}, along the unbounded contour $\mathcal{C}$, we have that $\mathbf{P}_-(z)^{-1}\mathbf{P}_+(z)-\mathbb{I}$ is exponentially decaying as $z\to\infty$ for each sufficiently small $\epsilon>0$. 

\subsection{Parametrix construction}
Let $D$ denote a small open disk of radius independent of $\epsilon$ centered at the point $z=q$.  We
now define an ad-hoc approximation to $\mathbf{P}(z)$, the \emph{global parametrix} $\dot{\mathbf{P}}(z)$:
\begin{equation}
\dot{\mathbf{P}}(z):=\begin{cases}
\dot{\mathbf{P}}^{\mathrm{in}}(z),&\quad z\in D\setminus\Sigma\\
\sigma_1\dot{\mathbf{P}}^{\mathrm{in}}(z^*)^*\sigma_1,&\quad z\in D^*\setminus\Sigma\\
\dot{\mathbf{P}}^{\mathrm{out}}(z),&\quad z\in\mathbb{C}\setminus (\Sigma\cup\overline{D}\cup\overline{D}^*).
\end{cases}
\end{equation}
We call $\dot{\mathbf{P}}^\mathrm{out}(z)$ the \emph{outer parametrix} and we call
$\dot{\mathbf{P}}^\mathrm{in}(z)$ the \emph{inner parametrix}.  We emphasize that there is no special parametrix needed in any neighborhood of the point $z=z_\mathrm{P}$ from which
the phantom poles of $s(z)$ emerge into the upper half-plane.
\subsubsection{The outer parametrix}
The outer parametrix is easy to write down explicitly:
\begin{equation}
\dot{\mathbf{P}}^{\mathrm{out}}(z):=e^{i\kappa\sigma_3/(2\epsilon)}
\mathbf{U}\beta(i(-z)^{1/2})^{\sigma_3}\mathbf{U}^\dagger e^{-i\kappa\sigma_3/(2\epsilon)},\quad \mathbf{U}:=\frac{1}{\sqrt{2}}\begin{bmatrix}e^{
i\pi/4} & e^{
-
i\pi/4}\\e^{
-
i\pi/4} & e^{
i\pi/4}\end{bmatrix},\quad
\mathbf{U}^{-1}=\mathbf{U}^\dagger.
\end{equation}
Here, $i(-z)^{1/2}\in\mathbb{C}_+$ and $\beta(k)$ is the function whose fourth power is given explicitly by
\begin{equation}
\beta(k)^4:=\frac{(k-q^{1/2})(k+q^{*1/2})}{(k-q^{*1/2})(k+q^{1/2})}
\end{equation}
where the principal branch of the square root is meant in all cases,  whose branch cuts lie along the arcs of $\pm(\mathcal{B}\cup \mathcal{B}^*)^{1/2}$, and for which the specific branch of the one-fourth power is selected so that $\beta(0)=1$.  It is easy to see that $\beta(k)$ is analytic along the imaginary $k$-axis, and since $\beta(k)^4>0$ for imaginary $k$, $\beta(k)>0$ for imaginary $k$ as well, leading to the conclusion that $\beta(\infty)=1$.  This allows us to asymptotically expand $\beta(k)$ for large $k$, leading to the corresponding expansion of $\dot{\mathbf{P}}^\mathrm{out}(z)$:
\begin{equation}
\dot{\mathbf{P}}^\mathrm{out}(z)=\mathbb{I} +\begin{bmatrix}0 &
\Im\{q^{1/2}\}e^{i\kappa/\epsilon}\\
-
\Im\{q^{1/2}\}e^{-i\kappa/\epsilon} & 0\end{bmatrix}\frac{1}{i(-z)^{1/2}} + O(z^{-1}),\quad z\to\infty.
\label{eq:cxdotPoutasymp}
\end{equation}
It is easy to check that $\dot{\mathbf{P}}^{\mathrm{out}}(z)$ is analytic for $z\in\mathbb{C}\setminus (\mathcal{B}\cup \mathcal{B}^*\cup\mathbb{R}_+)$, and that it is bounded independently of $\epsilon$ uniformly for $z\in\mathbb{C}\setminus (D\cup D^*)$ and has determinant one.  Also, $\dot{\mathbf{P}}^\mathrm{out}(0)=\mathbb{I}$, $\dot{\mathbf{P}}^\mathrm{out}(z^*)^*=\sigma_1\dot{\mathbf{P}}^\mathrm{out}(z)\sigma_1$, and 
\begin{equation}
\dot{\mathbf{P}}^\mathrm{out}_+(z)=\dot{\mathbf{P}}^\mathrm{out}_-(z)\begin{bmatrix}
0 &
e^{i\kappa/\epsilon}\\
-
e^{-i\kappa/\epsilon} & 0\end{bmatrix},\quad z\in B,
\end{equation}
which should be compared with \eqref{eq:cxPjumpB}, while $\dot{\mathbf{P}}^\mathrm{out}_+(z)=i^{\sigma_3}\dot{\mathbf{P}}^\mathrm{out}_-(z)i^{-\sigma_3}$ for $z>0$, which should be compared with \eqref{eq:cxPjumpzpositive1rewrite}, \eqref{eq:cxPjumpzpositive4rewrite}, and \eqref{eq:cxPjumpzpositive2-3rewrite}.

\subsubsection{The inner parametrix}
We now construct the inner parametrix in terms of Airy functions \cite{DLMF}, following closely the discussion in section 5.2 of \cite{SG1}. Note that as a consequence of the continuity of $h(z)$ at $z=q$ and the condition \eqref{eq:h-equilibrium} we have $2h(q)=\kappa$.  Then, using \eqref{eq:hprimeSquared} shows that the equation
\begin{equation}
y^3=(2ih(z)-i\kappa)^2
\end{equation}
defines three different univalent functions $y$ of $z$ in a neighborhood of $z=q$ each of which maps $z=q$ to $y=0$ and which differ by factors of the cube roots of unity; we choose the 
branch $y=y(z)$ for which the image of $\mathcal{B}\cap D$ is a segment of the negative real $y$-axis that abuts the origin.  We then set $\zeta:=\epsilon^{-2/3}y(z)$.  Then for $z\in \mathcal{B}_+\cup \mathcal{B}_-\cup\mathcal{C}$ we have $2ih(z)-i\kappa=-\epsilon\zeta^{3/2}$ where the principal branch of the $3/2$ power is meant.  Note that in terms of $\zeta$ we may write the outer parametrix in the form
\begin{equation}
\dot{\mathbf{P}}^\mathrm{out}(z)=e^{i\kappa\sigma_3/(2\epsilon)}\mathbf{H}(z)\epsilon^{\sigma_3/6}\zeta^{\sigma_3/4}\mathbf{U}^\dagger e^{-i\kappa\sigma_3/(2\epsilon)},\quad z\in \overline{D},
\end{equation}
where $\mathbf{H}(z)$ is an $\epsilon$-independent analytic matrix function with determinant one, defined for $z\in \overline{D}$ by
\begin{equation}
\mathbf{H}(z):=\mathbf{U}[y(z)^{-1/4}\beta(i(-z)^{1/2})]^{\sigma_3},\quad z\in\overline{D}.
\end{equation}
That $\mathbf{H}(z)$ is analytic follows from the fact that $y(z)^{-1/4}\beta(i(-z)^{1/2})$ extends from $D\setminus \mathcal{B}$ to all of $D$ as a single-valued non-vanishing analytic function.  Now for convenience write
$\xi:=(\frac{3}{4})^{2/3}\zeta$ and define a matrix in the $\zeta$-plane by the formulae:
\begin{equation}
\mathbf{Z}(\zeta):=\sqrt{2\pi}\left(\frac{4}{3}\right)^{\sigma_3/6}\begin{bmatrix}
e^{-3\pi i/4}\mathrm{Ai}'(\xi) & e^{11\pi i/12}\mathrm{Ai}'(\xi e^{-2\pi i/3})\\
e^{-\pi i/4}\mathrm{Ai}(\xi) & e^{\pi i/12}\mathrm{Ai}(\xi e^{-2\pi i/3})\end{bmatrix}
e^{2\xi^{3/2}\sigma_3/3},\quad 0<\arg(\zeta)<\frac{2\pi}{3},
\end{equation}
\begin{equation}
\mathbf{Z}(\zeta):=\sqrt{2\pi}\left(\frac{4}{3}\right)^{\sigma_3/6}\begin{bmatrix}
e^{-5\pi i/12}\mathrm{Ai}'(\xi e^{2\pi i/3}) & e^{11\pi i/12}\mathrm{Ai}'(\xi e^{-2\pi i/3})\\
e^{-7\pi i/12}\mathrm{Ai}(\xi e^{2\pi i/3}) & e^{\pi i/12}\mathrm{Ai}(\xi e^{-2\pi i/3})\end{bmatrix}
e^{2\xi^{3/2}\sigma_3/3},\quad\frac{2\pi}{3}<\arg(\zeta)<\pi,
\end{equation}
\begin{equation}
\mathbf{Z}(\zeta):=\sqrt{2\pi}\left(\frac{4}{3}\right)^{\sigma_3/6}\begin{bmatrix}
e^{11\pi i/12}\mathrm{Ai}'(\xi e^{-2\pi i/3}) & e^{7\pi i/12}\mathrm{Ai}'(\xi e^{2\pi i/3})\\
e^{\pi i/12}\mathrm{Ai}(\xi e^{-2\pi i/3}) & e^{5\pi i/12}\mathrm{Ai}(\xi e^{2\pi i/3})\end{bmatrix}
e^{2\xi^{3/2}\sigma_3/3},\quad -\pi<\arg(\zeta)<-\frac{2\pi}{3},
\end{equation}
\begin{equation}
\mathbf{Z}(\zeta):=\sqrt{2\pi}\left(\frac{4}{3}\right)^{\sigma_3/6}\begin{bmatrix}
e^{-3\pi i/4}\mathrm{Ai}'(\xi) & e^{7\pi i/12}\mathrm{Ai}'(\xi e^{2\pi i/3})\\
e^{-\pi i/4}\mathrm{Ai}(\xi) & e^{5\pi i/12}\mathrm{Ai}(\xi e^{2\pi i/3})\end{bmatrix}
e^{2\xi^{3/2}\sigma_3/3},\quad -\frac{2\pi}{3}<\arg(\zeta)<0.
\end{equation}
From the well-known asymptotic behavior \cite{DLMF} of $\mathrm{Ai}(\xi)$ and $\mathrm{Ai}'(\xi)$, it follows
that
\begin{equation}
\mathbf{Z}(\zeta)\mathbf{U}\zeta^{-\sigma_3/4}=\mathbb{I}+\begin{bmatrix}
O(\zeta^{-3/2}) & O(\zeta^{-1})\\O(\zeta^{-2}) & O(\zeta^{-3/2})\end{bmatrix},
\label{eq:Airyasymp}
\end{equation}
as $\zeta\to\infty$ in all directions of the complex plane.  We define the inner parametrix by setting
\begin{equation}
\dot{\mathbf{P}}^\mathrm{in}(z):=e^{i\kappa\sigma_3/(2\epsilon)}\mathbf{H}(z)\epsilon^{\sigma_3/6}\mathbf{Z}(\zeta)e^{-i\kappa\sigma_3/(2\epsilon)}\mathbf{D}(z),\quad z\in D\setminus\Sigma,
\end{equation}
where $\mathbf{D}(z)$ is a near-identity diagonal matrix factor given by
\begin{equation}
\mathbf{D}(z):=\begin{cases}\mathbb{I},&\quad z\in D\cap(\mathcal{L}\cup\mathcal{R})\\
\tilde{E}(z)^{-\sigma_3/2},&\quad\text{elsewhere in $D\setminus\Sigma$}.
\end{cases}
\end{equation}
It is a calculation to confirm that the following jump conditions hold, which should be compared with
\eqref{eq:cxPjumpB}--\eqref{eq:cxPjumpCminusB}:
\begin{equation}
\dot{\mathbf{P}}^\mathrm{in}_+(z)=\dot{\mathbf{P}}^\mathrm{in}_-(z)\begin{bmatrix}0 & 
e^{i\kappa/\epsilon}\\
-
e^{-i\kappa/\epsilon} & 0\end{bmatrix},\quad
z\in \mathcal{B}\cap D,
\end{equation}
\begin{equation}
\dot{\mathbf{P}}^\mathrm{in}_+(z)=\dot{\mathbf{P}}^\mathrm{in}_-(z)\begin{bmatrix}
1 & 0\\
e^{-2ih(z)/\epsilon} & 1\end{bmatrix}\tilde{E}(z)^{-\sigma_3/2},\quad z\in \mathcal{B}_+\cap D,
\end{equation}
\begin{equation}
\dot{\mathbf{P}}^\mathrm{in}_+(z)=\dot{\mathbf{P}}^\mathrm{in}_-(z)\tilde{E}(z)^{\sigma_3/2}
\begin{bmatrix}1 & 0\\
e^{-2ih(z)/\epsilon} & 1\end{bmatrix},\quad z\in \mathcal{B}_-\cap D,
\end{equation}
and
\begin{equation}
\dot{\mathbf{P}}^\mathrm{in}_+(z)=\dot{\mathbf{P}}^\mathrm{in}_-(z)
\tilde{E}(z)^{\sigma_3/2}\begin{bmatrix}1 & 
e^{2ih(z)/\epsilon}\\0 & 1\end{bmatrix}
\tilde{E}(z)^{-\sigma_3/2},\quad z\in \mathcal{C}\cap D.
\end{equation}
That is, the inner parametrix satisfies \emph{exactly} the same jump conditions satisfied by $\mathbf{P}(z)$ in a neighborhood $D$ of $z=q$.
\subsubsection{Mismatch of the global parametrix across $\partial D$}
It is not difficult to see that the inner and outer parametrices match each other quite well on the boundary of the disk $D$.  Indeed, we have
\begin{multline}
\dot{\mathbf{P}}^\mathrm{in}(z)\dot{\mathbf{P}}^\mathrm{out}(z)^{-1}=e^{i\kappa\sigma_3/(2\epsilon)}\mathbf{H}(z)\epsilon^{\sigma_3/6}\mathbf{Z}(\zeta)\mathbf{U}
\zeta^{-\sigma_3/4}\epsilon^{-\sigma_3/6}\mathbf{H}(z)^{-1}e^{-i\kappa\sigma_3/\epsilon}\cdot\dot{\mathbf{P}}^\mathrm{out}(z)\mathbf{D}(z)\dot{\mathbf{P}}^\mathrm{out}(z)^{-1},\\z\in\partial D.
\end{multline}
But, since $\mathbf{D}(z)=\mathbb{I}+O(\epsilon)$ while $\dot{\mathbf{P}}^\mathrm{out}(z)$ and its inverse are uniformly bounded for $z\in \partial D$, the latter factor is obviously $\mathbb{I}+O(\epsilon)$.  Also, since $z\in \partial D$ corresponds to $\zeta^{-1}=O(\epsilon^{2/3})$, from \eqref{eq:Airyasymp} we have that $\epsilon^{\sigma_3/6}\mathbf{Z}(\zeta)\mathbf{U}\zeta^{-\sigma_3/4}\epsilon^{-\sigma_3/6}=\mathbb{I}+O(\epsilon)$ for $z\in\partial D$.  It follows that
\begin{equation}
\dot{\mathbf{P}}^\mathrm{in}(z)\dot{\mathbf{P}}^\mathrm{out}(z)^{-1}=\mathbb{I}+O(\epsilon),\quad
\text{uniformly for $z\in\partial D$.}
\end{equation}
\subsection{Error analysis.}  Let $\mathbf{E}(z):=\mathbf{P}(z)\dot{\mathbf{P}}(z)^{-1}$ denote the
mismatch between the (not explicitly known) matrix function $\mathbf{P}(z)$ and its explicit global parametrix $\dot{\mathbf{P}}(z)$.  Because $\mathbf{P}(z)$ and $\dot{\mathbf{P}}(z)$ satisfy exactly the same jump conditions on the contours $\mathcal{B}$, $\mathcal{B}^*$, and $\Sigma\cap (D\cup D^*)$,  $\mathbf{E}(z)$ may be extended analytically to the latter contours.  On the other hand, $\mathbf{E}(z)$ has jumps across the disk boundaries $\partial D$ and $\partial D^*$ where $\mathbf{P}(z)$ has no jump.
That is, $\mathbf{E}(z)$ may be regarded as an analytic matrix function of $z\in\mathbb{C}\setminus\Sigma_\mathbf{E}$, where $\Sigma_\mathbf{E}$ is the contour shown in Figure~\ref{fig:EllipticError}.
\begin{figure}[h]
\begin{center}
\includegraphics{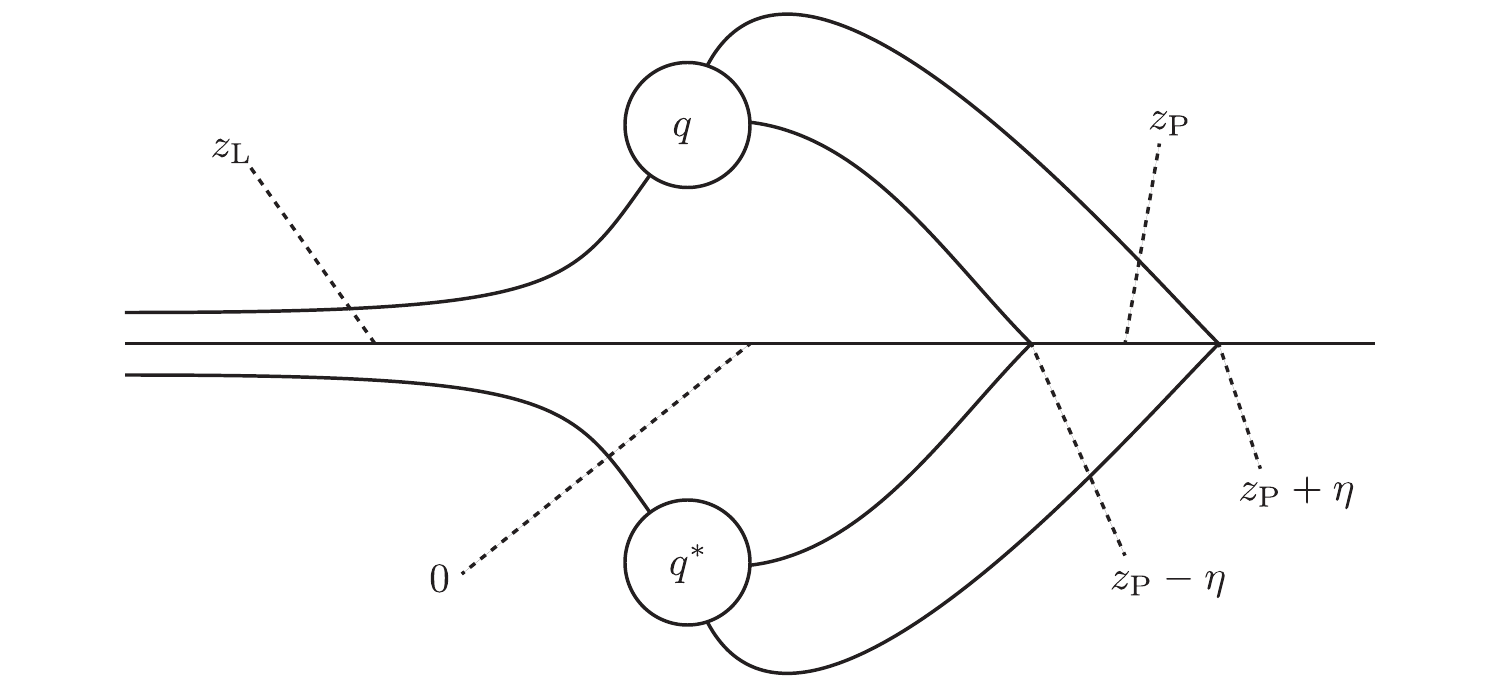}
\end{center}
\caption{The jump contour $\Sigma_\mathbf{E}$ for $\mathbf{E}(z)$.  The disk boundaries are taken to be oriented in the clockwise direction, and all other contour arcs are oriented left-to-right. }
\label{fig:EllipticError}
\end{figure}
By combining the jump conditions satisfied by $\mathbf{P}(z)$ with those known to be satisfied by the parametrix $\dot{\mathbf{P}}(z)$, it is a simple calculation to see that on the complex arcs of $\Sigma_\mathbf{E}$ as well as for $z<0$, we have that $\mathbf{E}_-(z)^{-1}\mathbf{E}_+(z)-\mathbb{I}$ is $O(\epsilon)$ in both the $L^\infty(\Sigma_\mathbf{E}\setminus\mathbb{R}_+)$ and the $L^2(\Sigma_\mathbf{E}\setminus\mathbb{R}_+)$ senses.  For $z>0$ we have that $[i^{\sigma_3}\mathbf{E}_-(z)i^{-\sigma_3}]^{-1}\mathbf{E}_+(z)-\mathbb{I}$ is $O(\epsilon)$ (in fact, exponentially small) in both the $L^\infty(\mathbb{R}_+)$
and $L^2(\mathbb{R}_+)$ senses.  Because both $\mathbf{P}(0)=\mathbb{I}$ and also $\dot{\mathbf{P}}(0)=\dot{\mathbf{P}}^\mathrm{out}(0)=\mathbb{I}$, we have $\mathbf{E}(0)=\mathbb{I}$.

The unusual nature of the jump condition on $\mathbb{R}_+$ can be understood by introducing the
matrix function $\mathbf{F}(w)$ defined by
\begin{equation}
\mathbf{F}(w)=\begin{cases}i^{\sigma_3}\mathbf{E}(w^{-2})i^{-\sigma_3},&\quad\Im\{w\}>0
\\\mathbf{E}(w^{-2}),&\quad\Im\{w\}<0.
\end{cases}
\end{equation}
The jump contour $\Sigma_\mathbf{F}$ for $\mathbf{F}(w)$ consists of reciprocals of the positive and negative square roots of the complex arcs of $\Sigma_\mathbf{E}$ along with the real and imaginary $w$-axes.  It follows directly that $\mathbf{V}(w)-\mathbb{I}=O(\epsilon)$ in both the $L^\infty(\Sigma_\mathbf{F})$ and $L^2(\Sigma_\mathbf{F})$ senses, where $\mathbf{V}(w):=\mathbf{F}_-(w)^{-1}\mathbf{F}_+(w)$ for $w\in\Sigma_\mathbf{F}$.  Also we unambiguously calculate the value $\mathbf{F}(\infty)=\mathbb{I}$.  This problem is solved by writing $\mathbf{F}(w)$ in terms of a Cauchy integral
\begin{equation}
\mathbf{F}(w)=\mathbb{I}+\frac{1}{2\pi i}\int_{\Sigma_\mathbf{F}}\frac{(\mathbb{I}+\mathbf{X}(s))(\mathbf{V}(s)-\mathbb{I})}{s-w}\,ds,\quad w\in\mathbb{C}\setminus\Sigma_\mathbf{F},
\end{equation}
where $\mathbf{X}\in L^2(\Sigma_\mathbf{F})$ solves the singular integral equation
\begin{equation}
\mathbf{X}(w)-\frac{1}{2\pi i}\int_{\Sigma_\mathbf{F}}\frac{\mathbf{X}(s)(\mathbf{V}(s)-\mathbb{I})}{s-w_-}\,ds=\frac{1}{2\pi i}\int_{\Sigma_\mathbf{F}}\frac{\mathbf{V}(s)-\mathbb{I}}{s-w_-}\,ds,\quad
w\in\Sigma_\mathbf{F}.
\end{equation}
Here the the notation $w_-$ indicates that the integrals are to be evaluated for $w\not\in\Sigma_\mathbf{F}$ and the limit of $w$ approaching a point on $\Sigma_\mathbf{F}$ is to be taken from the right side according to the orientation of each arc.  By standard theory \cite{Zhou89}, this equation has a unique solution whose norm is $O(\epsilon)$.  Now since $\mathbf{V}(w)-\mathbb{I}$ vanishes to all orders as $w\to 0$ along $\Sigma_\mathbf{F}$, $\mathbf{F}(w)$ has a nonconvergent but asymptotic power series representation
\begin{equation}
\mathbf{F}(w)-\mathbb{I}\sim\sum_{n=0}^\infty\mathbf{F}_nw^n,\quad w\to 0,
\end{equation}
with coefficients given by
\begin{equation}
\mathbf{F}_n:=\frac{1}{2\pi i}\int_{\Sigma_\mathbf{F}}(\mathbb{I}+\mathbf{X}(s))(\mathbf{V}(s)-\mathbb{I})s^{-n-1}\,ds,\quad n\ge0.
\end{equation}
Now in addition to the estimates $\|\mathbf{X}\|_{L^2(\Sigma_\mathbf{F})}=O(\epsilon)$, 
$\|\mathbf{V}-\mathbb{I}\|_{L^\infty(\Sigma_\mathbf{F})}=O(\epsilon)$, and
$\|\mathbf{V}-\mathbb{I}\|_{L^2(\Sigma_\mathbf{F})}=O(\epsilon)$, we have from \eqref{eq:cxPjumpznegativerewrite} and \eqref{eq:cxPjumpzpositive4rewrite} that
$\mathbf{V}(w)-\mathbb{I}=O(e^{-M/(\epsilon |w|^2)})$ holds for sufficiently small (independent of $\epsilon$) $|w|$, $w\in\Sigma_\mathbf{F}$.  Assuming without loss of generality that $|w|<1$, 
this implies that $(\mathbf{V}(w)-\mathbb{I})w^{-n-1}=O(\epsilon^{n+1})$ uniformly for small $w$.
These estimates combined with the Cauchy-Schwarz inequality yield that $\mathbf{F}_n=O(\epsilon)$ for all $n\ge 0$.  Note that $\mathbf{F}_0$ is necessarily a diagonal matrix (for consistency with the jump conditions of $\mathbf{E}(z)$ for large $z$).

From this information, we can now obtain the asymptotic expansion of $\mathbf{P}(z)=\mathbf{E}(z)\dot{\mathbf{P}}(z)$ as $z\to\infty$:  taking $\Im\{w\}<0$ we have $w=-i(-z)^{-1/2}$, and so $\mathbf{E}(z)=\mathbf{F}(-i(-z)^{-1/2})$.  Therefore, since $\dot{\mathbf{P}}(z)=\dot{\mathbf{P}}^\mathrm{out}(z)$ for large $z$, using \eqref{eq:cxdotPoutasymp} gives
\begin{equation}
\mathbf{P}(z)=\mathbb{I}+\mathbf{F}_0 +\left( (\mathbb{I}+\mathbf{F}_0)\begin{bmatrix}
0 & 
\Im\{q^{1/2}\}e^{i\kappa/\epsilon}\\
-
\Im\{q^{1/2}\}e^{-i\kappa/\epsilon} & 0\end{bmatrix}-\mathbf{F}_1\right)\frac{1}{i(-z)^{1/2}} + O(z^{-1}),\quad z\to\infty.
\end{equation}
Finally, since for large $z\not\in \Lambda\cup\Lambda^*$, we have $\mathbf{M}(i(-z)^{1/2})=\mathbf{N}(z)=\mathbf{O}(z)=\mathbf{P}(z)e^{-ig(z)\sigma_3/\epsilon}$, we see from \eqref{eq:phirecover} that the solution of the MNLS Cauchy problem is
\begin{equation}
\phi(x,t)=\mathop{\lim_{z\to\infty}}_{z\not\in\Lambda\cup\Lambda^*}\frac{2i(-z)^{1/2}}{\alpha}\frac{M_{12}(i(-z)^{1/2})}{M_{22}(i(-z)^{1/2})}=
\mathop{\lim_{z\to\infty}}_{z\not\in\Lambda\cup\Lambda^*}\frac{2i(-z)^{1/2}}{\alpha}\frac{P_{12}(z)}{P_{22}(z)} = 
\frac{2}{\alpha}\Im\{q^{1/2}\}e^{i\kappa/\epsilon} + O(\epsilon),\quad\epsilon\to 0.
\end{equation}
This asymptotic formula has the form written in the statement of Theorem~\ref{thm:main} where
\begin{equation}
A(x,t):=\frac{2}{\alpha}\Im\{q(x,t)^{1/2}\}\quad\text{and}\quad S(x,t):=\kappa(x,t).
\label{eq:A-S-xt}
\end{equation}
That the derived fields $\rho(x,t):=A(x,t)^2$ and $u(x,t):=S_x(x,t)$ solve the dispersionless MNLS system \eqref{eq:DispersionlessMNLS} has already been shown, since
according to \eqref{eq:kappa-x} we can express $u(x,t)$ explicitly in terms of $q(x,t)$, and hence
\begin{equation}
\rho(x,t)=\frac{4}{\alpha^2}\left(\Im\{q(x,t)^{1/2}\}\right)^2\quad\text{and}\quad u(x,t)=\frac{1}{\alpha}\left(1-4|q(x,t)|\right).
\label{eq:rho-u-xt}
\end{equation}
We have already seen (see the discussion at the end of \S\ref{sec:q}) that these substitutions reduce the dispersionless MNLS system to diagonal (Riemann invariant) form, reproducing exactly the partial differential equations \eqref{eq:cxWhitham} that were proved in Proposition~\ref{prop:cxWhitham} to be satisfied by $q(x,t)$ and $q(x,t)^*$.  Next, we show that $A(x,0)=A_0(x)$ and that $S(x,0)=S_0(x)$.  To see this, we note that according to Proposition~\ref{prop:momentscxtzero}, we have $q(x,0)=\mathfrak{z}(x)$ for $x<x_\mathrm{c}$, where
$z=\mathfrak{z}(x)$ is the solution in the upper half-plane of the turning point equation 
\eqref{eq:TPCgeneral}; this implies that the quadratic $\chi(x;z)$ is a constant multiple of the factored form $(z-\mathfrak{z}(x))(z-\mathfrak{z}(x)^*)=(z-q(x,0))(z-q(x,0)^*)$, which in turn implies the identities:
\begin{equation}
\alpha^2\rho_0(x)+\frac{1}{2}(\alpha u_0(x)-1)=-\left[q(x,0)+q(x,0)^*\right]\quad\text{and}\quad
\frac{1}{16}\left(\alpha u_0(x)-1\right)^2=|q(x,0)|^2.
\label{eq:two-identities}
\end{equation}
Under the condition that $u_0(x)<1/\alpha$, it is easy to see that \eqref{eq:rho-u-xt} and \eqref{eq:two-identities} imply that 
\begin{equation}
\rho(x,0)=\rho_0(x)\quad\text{and}\quad u(x,0)=u_0(x).
\end{equation}
But because $\rho_0(x)>0$, the condition $u_0(x)<1/\alpha$ follows from the inequality $Q<0$ that holds for $x<x_\mathrm{c}$ at $t=0$.  It is obvious from \eqref{eq:A-S-xt} that $A(x,t)>0$ for all $x<x_\mathrm{c}(t)$ and $t\ge 0$, and hence we can take a positive square root to find that
$A(x,0)=A_0(x)=\mathrm{sech}(x)$.  Since we already have $u(x,0)=u_0(x)$, to prove that $S(x,0)=S_0(x)$ we simply need to establish the latter identity at some point $x<x_\mathrm{c}$.  In fact, we examine the limit $x\uparrow x_\mathrm{c}$; since  according to \eqref{eq:data} we have $S_0(x_\mathrm{c})=\delta x_\mathrm{c}+\mu\log(\cosh(x_\mathrm{c}))$, we simply compare with the explicit expression for the limiting value of $S(x,0)=\kappa(x,0)$ as $x\uparrow x_\mathrm{c}$
as given by \eqref{eq:KappaCrit}; thus the problem is reduced to checking the equality of two explicit functions of the three real parameters $\alpha$, $\delta$, and $\mu$ over the region constrained by the three inequalities \eqref{eq:condition-transsonic}--\eqref{eq:condition-no-zeros}.   We have confirmed this equality numerically, and we therefore conclude that indeed $S(x,0)=S_0(x)=\delta x+\mu\log(\cosh(x))$ holds for all $x<x_\mathrm{c}$.


\begin{thebibliography}{99}
\bibitem{DLMF} Digital Library of Mathematical Functions. \textit{Release date 2012-3-23.} National Institute of Standards and Technology from \texttt{http://dlmf.nist.gov/}.
\bibitem{SG1} R. J. Buckingham and P. D. Miller, ``The sine-Gordon equation in the semiclassical limit:  dynamics of fluxon condensates,'' to appear in \textit{Memoirs AMS}.  \texttt{arXiv:1103.0061}.
\bibitem{ClaeysG09}
T. Claeys and T. Grava, ``Universality of the break-up profile for the KdV equation in the small dispersion limit using the Riemann-Hilbert approach,'' \textit{Comm. Math. Phys.},  \textbf{286}, 979--1009, 2009.
\bibitem{ClaeysG10a}
T. Claeys and  T. Grava, ``Solitonic asymptotics for the Korteweg-de Vries equation in the small dispersion limit,''  \textit{SIAM J. Math. Anal.},  \textbf{42}, 2132--2154, 2010.
 \bibitem{ClaeysG10b}
T. Claeys and T. Grava, ``Painlev\'e II asymptotics near the leading edge of the oscillatory zone for the Korteweg-de Vries equation in the small-dispersion limit,'' \textit{Comm. Pure Appl. Math.}, \textbf{63}, 203--232,  2010.
\bibitem{DeiftVZ97} P. Deift, S. Venakides, and X. Zhou, ``New results in small dispersion KdV by an extension of the steepest descent method for Riemann-Hilbert problems,'' \textit{Internat. Math. Res. Notices}, \textbf{1997}, 286--299, 1997.
\bibitem{DeiftZ93}  P. Deift and X. Zhou, ``A steepest descent method for oscillatory Riemann-Hilbert problems:  asymptotics for the mKdV equation'', \textit{Ann. Math.}, \textbf{137}, 295--370, 1993.
\bibitem{DiFrancoM08} J. C. DiFranco and P. D. Miller, ``The semiclassical modified nonlinear Schr\"odinger equation I:  Modulation theory and spectral analysis,'' \textit{Physica D}, \textbf{237}, 947--997, 2008.
\bibitem{DiFrancoM12b} J. C. DiFranco and P. D. Miller, ``The semiclassical modified nonlinear Schr\"odinger equation III:  asymptotic analysis of the Cauchy problem.  The hyperbolic region for transsonic initial data,'' in preparation.
\bibitem{DiFrancoMM11} J. C. DiFranco, P. D. Miller, and B. K. Muite, ``On the modified nonlinear Schr\"odinger equation in the semiclassical limit:  supersonic, subsonic, and transsonic behavior,''
\textit{Acta Math. Sci.}, \textbf{31B}, 2342--2377, 2011.
\bibitem{Doktorov02}
E. V. Doktorov, ``The modified nonlinear Schr\"odinger equation:  facts and artefacts,'' \textit{European Phys. J. B}, \textbf{29}, 227--231, 2002.
\bibitem{Doktorov06}
E. V. Doktorov, ``Dynamics of a subpicosecond dispersion-managed soliton in a fibre:  a perturbative analysis,'' \textit{J. Mod. Optics}, \textbf{53}, 2701--2723, 2006.
\bibitem{DoktorovK01}
E. V. Doktorov and I. S. Kuten, ``The Gordon-Haus effect for modified NLS solitons,'' \textit{Europhys. Lett.}, \textbf{53}, 22--28, 2001.
\bibitem{KamvissisMM03} S. Kamvissis, K. D. T.-R. McLaughlin, and P. D. Miller,
\textit{Semiclassical soliton ensembles for the focusing nonlinear Schr\"odinger equation}, 
Annals of Math. Studies, \textbf{154}, Princeton University Press, 2003.
\bibitem{KaupN78} D. J. Kaup and A. C. Newell, ``An exact solution for a derivative nonlinear Schr\"odinger equation,'' \textit{J. Math. Phys.}, \textbf{19}, 798--801, 1978.
\bibitem{KitaevV97}  A. V. Kitaev and A. H. Vartanian, ``Leading-order temporal asymptotics of the modified nonlinear Schr\"odinger equation:  solitonless sector,'' \textit{Inverse Prob.}, \textbf{13}, 1311--1339, 1997.
\bibitem{KitaevV99} A. V. Kitaev and A. H. Vartanian, ``Asymptotics of solutions to the modified nonlinear Schr\"odinger equation:  solitons on a nonvanishing continuous background,'' \textit{SIAM J. Math. Anal.}, \textbf{30}, 787--832, 1999.
\bibitem{KuvshinovL94} B. N. Kuvshinov and V. P. Lakhin, ``The Riemann invariants and characteristic velocities of Whitham equations for the derivative nonlinear Schr\"odinger equation,''
 \textit{Phys. Scr.}, \textbf{49}, 257--260, 1994.
\bibitem{Madelung26} E. Madelung, ``Quantum theory in hydrodynamic form,'' \textit{Zeitschrift f\"ur Physik}, \textbf{40}, 322--326, 1926.
\bibitem{Ramond96} T. Ramond, ``Semiclassical study of quantum scattering on the line,'' \textit{Comm. Math. Phys.},  \textbf{177}, 221--254, 1996.
\bibitem{TovbisVZ04} A. Tovbis, S. Venakides, and X. Zhou, ``On semiclassical (zero dispersion limit) solutions of the focusing nonlinear Schr\"odinger equation,'' \textit{Comm. Pure Appl. Math.},
\textbf{57}, 877--985, 2004.
\bibitem{WadatiKI79} M. Wadati, K. Konno, and Y. Ichikawa, ``Generalization of inverse scattering method,'' \textit{J. Phys. Soc. Japan}, \textbf{46}, 1965--1966, 1979.
\bibitem{Zhou89} X. Zhou, ``The Riemann-Hilbert problem and inverse scattering,'' \textit{SIAM J. Math. Anal.}, \textbf{20}, 966--986, 1989.
\end{thebibliography}
\end{document}